\documentclass[10pt,a4paper,twoside,reqno]{amsart}

\pdfoutput=1
\usepackage{subcaption}
\usepackage[T1]{fontenc}
\usepackage{float}
\usepackage{latexsym}
\usepackage{amsmath,amstext,amssymb,amsfonts,
amscd,bm,array,multirow,amsbsy,mathrsfs,calc}
\usepackage{t1enc}
\usepackage{indentfirst}
\usepackage{pb-diagram}
\usepackage{graphicx}
\usepackage{enumerate}
\usepackage{color}
\usepackage[all]{xy}
\usepackage{hyperref}
\hypersetup{colorlinks=false,pdfborderstyle={/S/U/W 0}}
\usepackage[usenames,dvipsnames]{xcolor}
\usepackage{url}
\usepackage{upgreek}

\theoremstyle{plain}
\newtheorem{thm}{Theorem}[section]

 \newtheorem{prop}[thm]{Proposition}

 \newtheorem{lemma}[thm]{Lemma}
\newtheorem{cor}[thm]{Corollary}

\theoremstyle{definition}
 \newtheorem{definition}[thm]{Definition}

 \theoremstyle{remark}
 \newtheorem{remark}[thm]{Remark}

\newcommand{\be}{\begin{equation*}}
\newcommand{\ee}{\end{equation*}}
\newcommand{\ben}{\begin{equation}}
\newcommand{\een}{\end{equation}}
\newcommand{\beqa}{\begin{eqnarray*}}
\newcommand{\eeqa}{\end{eqnarray*}}
\newcommand{\beqan}{\begin{eqnarray}}
\newcommand{\eeqan}{\end{eqnarray}}
\newcommand{\nn}{\nonumber}
\newcommand{\la}{\langle}
\newcommand{\ra}{\rangle}

\def\i{\mathbf{i}}

\def\Z{\mathbb{Z}}
\def\C{\mathbb{C}}
\def\R{\mathbb{R}}

\def\End{\mathrm{End}}
\def\Hess{\mathrm{Hess}}

\def\Crit{\mathrm{Crit}}

\def\pd{\partial}

\def\reg{\mathrm{reg}}
\def\sing{\mathrm{sing}}
\def\crit{\mathrm{crit}}
\def\noncrit{\mathrm{noncrit}}

\def\fM{\mathfrak{M}}

\def\rT{\mathrm{T}}

\def\fX{\mathfrak{X}}

\def\Cx{\mathrm{Cx}}

\def\dd{\mathrm{d}}

\def\fc{\mathfrak{c}}

\def\fd{\mathfrak{d}}

\newcommand{\id}{\mathrm{id}}
\newcommand{\Tr}{\mathrm{Tr}}
\newcommand{\tr}{\mathrm{tr}}

\newcommand{\sign}{\mathrm{sign}}

\def\GL{\mathrm{GL}}

\def\cA{\mathcal{A}}
\def\cB{\mathcal{B}}
\def\cC{\mathcal{C}}
\def\cD{\mathcal{D}}

\def\cF{\mathcal{F}}
\def\cG{\mathcal{G}}
\def\cH{\mathcal{H}}
\def\cL{\mathcal{L}}
\def\cM{\mathcal{M}}

\def\Tw{\mathrm{Tw}}

\def\cV{\mathcal{V}}

\def\Met{\mathrm{Met}}
\def\cA{\mathcal{A}}

\def\U{\mathrm{U}}

\def\cO{\mathcal{O}}

\def\PSL{\mathrm{PSL}}

\def\GL{\mathrm{GL}}

\def\rS{\mathrm{S}}

\def\fM{\mathfrak{M}}

\def\p{{\bf p}}

\def\SO{\mathrm{SO}}

\def \hess{\mathrm{hess}}

\def\Riemm{\mathfrak{Riemm}}

\newcommand{\eqdef}{\stackrel{{\rm def.}}{=}}

\def\grad{\mathrm{grad}}
\def\Sym{\mathrm{Sym}}
\def\End{\mathrm{End}}

\def\Re{\mathrm{Re}}
\def\Im{\mathrm{Im}}

\def\O{\mathrm{O}}

\newcommand{\twopartdef}[4]
{
	\left\{
	\begin{array}{ll}
		#1 & \mbox{if } #2 \\
		#3 & \mbox{if } #4
	\end{array}
	\right.
}

\def\vol{\mathrm{vol}}

\begin{document}

\title[Strong rapid turn inflation and contact Hamilton-Jacobi equations]{
 Strong rapid turn inflation and contact Hamilton-Jacobi equations}

\author{E. M. Babalic$^{1}$, C. I. Lazaroiu$^{1,2}$, V. O. Slupic$^{1}$}

\address{$^{1}$Horia Hulubei National
  Institute for R\&D in Physics and Nuclear Engineering\\
  Department of Theoretical Physics,
  Reactorului 30, Bucharest-Magurele, 077125, Romania\\
mbabalic@theory.nipne.ro, lcalin@theory.nipne.ro, victor.slupic@theory.nipne.ro\\
$^{2}$Departamento de Matematicas, Universidad UNED - Madrid\\
 Calle de Juan del Rosal 10, 28040, Madrid, Spain\\
clazaroiu@mat.uned.es}

\maketitle

\begin{abstract}
We consider the consistency condition for ``strong'' sustained rapid
turn inflation with third order slow roll (SRRT) in two-field
cosmological models with oriented scalar manifold as a geometric PDE
which constrains the scalar field metric and potential of such models. When
supplemented by appropriate boundary conditions, the equation
determines one of these objects in terms of the other and hence
selects ``fiducial'' models for strong SRRT inflation. When the scalar
potential is given, the equation can be simplified by fixing the
conformal class of the scalar field metric (equivalently, fixing a
complex structure which makes the scalar manifold into a complex
Riemann surface). Then the consistency equation becomes a contact
Hamilton-Jacobi PDE which determines the scalar field metric within
the given conformal class. We analyze this equation with standard
methods of PDE theory, discuss its quasilinearization near a non-degenerate
critical point of the scalar potential and extract natural asymptotic
conditions for its solutions at such points. We also give numerical
examples of solutions to a simple Dirichlet problem. For the case of
elliptic curves relevant to two-field axion cosmology, we determine
the general symmetry-adapted solution of the equation for potentials
with a single charge vector.
\end{abstract}



\tableofcontents

\section*{Introduction}

Multifield inflationary cosmology is a subject of renewed interest
given its natural connection to fundamental theories of gravity and
matter such as string theory. In fact, recent arguments suggest
\cite{S1,S2,S3,S4} that such models may be preferred in all consistent
theories of quantum gravity. The simplest examples are provided by
two-field models, which already display some of the general features
of multifield cosmology.

Two-field cosmological models have much richer phenomenology than
their one-field counterparts, due to existence of non-geodesic
trajectories which allow for new inflationary regimes such as those
characterized by ``large'' turning rates. Cosmological trajectories
with sustained ``rapid'' turn and slow roll are of significant
phenomenological interest \cite{RT1, RT2, RT3, RT4, RT5, RT6, RT7,
  RT8, RT9, RT10, RT11, LA, LA2, APR, Chakraborty}. Such models also
display much more involved dynamics than their one-field counterparts,
some general aspects of which were explored in \cite{ren,grad,
  natobs}. Noether symmetries of two-field models were discussed in
\cite{Noether1, Noether2, Hesse}.

A conceptually important problem in the study of multifield
inflationary scenarios is to find useful conditions on a given model
which tell us when cosmological trajectories satisfying the desired
regime may exist. This is essential for developing {\em selection
  criteria} which can be used to rule out phenomenologically
irrelevant models without attempting to compute trajectories in each
candidate model ``by hand'' -- given that every candidate model admits
a continuous infinity of solutions depending on the choice of initial data.

This problem was first addressed systematically in reference
\cite{cons} (see \cite{consproc} for a summary) in the context of
so-called {\em strong} rapid turn inflation with third order slow roll
(SRRT). The qualifier ``strong'' means that the square of the
dimensionless turning rate of trajectories is required to be much
larger than one, as has been the tradition in studies of rapid-turn
inflation.

It was shown in \cite{cons} that such a dynamical regime can exist
only in those regions of the field space where the scalar field metric
and scalar potential of the two-field cosmological model approximately
satisfy the so-called {\em strong\footnote{Later work \cite{Achucarro,
    LiliaCons} derived a similar -- but less restrictive --
  consistency condition for so-called {\em weak} rapid-turn inflation
  with third order slow roll, where the qualifier ``weak'' means that
  one requires only that the squared dimensionless turning rate
  $\eta_\perp^2$ is much larger that the magnitudes of the first three
  slow roll parameters (though it could be smaller that one, so this
  is not ``rapid'' turn in the traditional sense). This reduces to the
  consistency condition derived in \cite{cons} in the traditional case
  $\eta_\perp^2\gg 1$. It was shown in loc. cit. that certain models
  previously proposed in the literature satisfy only the weak SRRT
  condition. We consider only {\em strong} rapid turn in the present
  paper. } SRRT consistency condition}. When imposed strictly, i.e. as
an actual equality rather than an approximate equality, this condition
is a nonlinear geometric PDE (herein called the {\em strong SRRT
  equation}) which constrains the scalar field metric and potential in
terms of each other and hence selects {\em fiducial two-field models}
for strong SRRT inflation. The equation thus provides a model
selection criterion which is {\em apriori} in the sense that it does
not require one to compute trajectories in order to test the fitness
of a candidate model for existence of this inflationary regime. The
scalar field metric and potential for which the cosmological equation
is likely to admit such a dynamical regime must be ``close'' (in the
Whitney topology) to those of a fiducial model, which satisfy the
strong SRRT equation exactly.

In previous literature on the subject, the dominant viewpoint has been
to start with a given pair consisting of a scalar field metric and
potential and interpret the strong SRRT equation as a constraint which
must be satisfied approximately by these objects in those regions of
the field space where strong SRRT trajectories have a chance to
exist. This approach is inefficient since it is based on ``trial and
error'' and because the set of two-field models with fixed scalar
field space is uncountably infinite. The point of view taken in the
present paper is more constructive.  Namely, we will take only the
scalar potential of the model as given and view the strong SRRT
equation as a geometric PDE for the scalar field metric, which can be
used to determine the latter in terms of the scalar potential provided
that appropriate boundary or asymptotic conditions are imposed. In
principle, this partially solves the model selection problem by
allowing one to parameterize fiducial strong SRRT models solely by
their target manifold and scalar potential.  In string
compactifications with low or no supersymmetry (i.e. those
compactifications most likely to be relevant to cosmology), the scalar
potential tends to be easier to compute than the metric on the moduli
space.

When the scalar potential is fixed, the strong SRRT equation (viewed
as a PDE for the scalar field metric) turns out to be rather
formidable when written in an arbitrary coordinate system on the
scalar field space. To simplify it, we fix the conformal class of the
metric, which -- for an oriented scalar field space -- amounts to
fixing a complex structure $J$ which renders the latter into a
(generally non-compact) Riemann surface. This {\em complex
  parameterization} allows us to locally write the strong SRRT
equation with fixed scalar potential as a first order PDE for the
conformal factor of the scalar field metric, which turns out to be the
local form of a geometric PDE of contact Hamilton-Jacobi type for the
volume form of that metric. This establishes an unexpected connection
between consistency of ``strong'' rapid turn inflation and contact
mechanics (see \cite{GG, BCT, LB, LTW, LV, LLLR, VL,
  Guijarro, SLVD}) as well as contact geometry \cite{Arnold, Geiges,
  LM}.

When written in local isothermal coordinates for the complex structure
$J$, our geometric equation takes the standard contact Hamilton-Jacobi
form for a contact Hamiltonian which is a cubic polynomial in momenta
with coefficients dependent in a complicated manner on the scalar
potential. Analyzing this equation with the classical method of
characteristics allows one to extract some qualitative information
about its local solutions, while studying various asymptotic regimes
provides further insight. As typical for nonlinear PDEs, a local
smooth solution of the equation can develop ``shocks'' (singularities
or discontinuities in its partial derivatives) when one attempts to
extend it to a solution defined on the entire scalar manifold,
as illustrated by the phenomenon of crossing characteristics. However,
the equation turns out to be {\em proper} in the sense of the theory
of viscosity solutions (see \cite{CL, CEL, BD, EvansBook, Tran, CIL,
  C, Katzourakis}), so its appropriately posed Dirichlet problems
have unique global viscosity solutions when the boundary conditions
are interpreted in the viscosity sense. Properness of the equation
also allows one to approximate viscosity solutions by classical
solutions of its viscosity perturbation, when one takes the viscosity
parameter of the latter to zero. This allows for numerical computation
of viscosity solutions, as we illustrate in the paper. In general, our
contact Hamiltonian does {\em not} satisfy the Tonelli-type and other
conditions required by the contact version \cite{SWY, WWY1, WWY2,
  WWY3} of Aubry-Mather theory \cite{Evans, Fathi}, so the equation
does not seem amenable to weak KAM methods.

The paper is organized as follows. Section \ref{sec:models} recalls
the geometric formulation of two-field cosmological models in the
absence of cosmological perturbations and discusses the complex
parameterization of such models. In the same section, we introduce
certain geometric objects which will play an important role later
on. Section \ref{sec:SRRTeq} briefly recalls the strong SRRT
consistency condition derived in \cite{cons} and discusses the strong
SRRT equation in its frame-free form, using the jet bundle formalism
for geometric nonlinear differential operators (see
\cite{OlverEquivalence, Kruglikov, Vinogradov}). In Section
\ref{sec:HJ}, we explain the process of passing to the complex
parameterization of the scalar field metric and show that, when the
scalar potential and conformal class of this metric are fixed, the
resulting equation for the volume form is of geometric contact
Hamilton-Jacobi type. We also show that working in isothermal
Liouville coordinates on the first jet bundle of the positive
determinant half-line bundle of the target manifold allows us to write
the equation locally as a standard contact Hamilton-Jacobi equation
for the conformal factor of the scalar field metric. In Section
\ref{sec:nonlinear}, we analyze this equation using the method of
characteristics, establish its properness and present some numerically
computed characteristics and viscosity approximants for a natural
Dirichlet problem in the case of quadratic potentials which are
positive on a vicinity of the origin in $\R^2$. In Section
\ref{sec:quasilinear}, we study the quasilinear approximation of the
contact Hamilton-Jacobi equation in the vicinity of a non-degenerate
critical point of the scalar potential. After discussing the
characteristics of this approximation, we analyze its symmetries and
scale-invariant solutions. We also present some numerically-computed
characteristics and viscosity approximants for a Dirichlet problem of
the quasilinear equation near a non-degenerate critical point. Finally,
we extract natural asymptotic conditions for solutions of the fully
nonlinear eqation in the vicinity of such a point. Section
\ref{sec:torus} briefly discusses the case of models with 2-torus
target, which arise as toy examples in axion cosmology \cite{Axiverse,
  AxiverseM, AxiverseB, Marsh}. After explaining the complex
parameterization of such models and the ``axionic'' presentation of
general smooth scalar potentials, we make a few remarks on the
interpretation of the contact Hamilton-Jacobi equation as a nonlinear
PDE for doubly-periodic functions on the complex plane and find the
general symmetry-adapted solution of this equation for potentials with
a single charge vector. Section \ref{sec:conclusions} presents our
conclusions and some directions for further research.

\subsection{Notations and conventions}
\label{subsec:notations}

\noindent All manifolds considered in this paper are smooth and
paracompact. We denote by $\cO$ the trivial real line bundle on
$\cM$. For any smooth vector bundle $E$ over $\cM$ and any $k\in
\Z_{\geq 0}\sqcup \{\infty\}$, we denote by $\cC^k(E)$ the space of
those sections of $E$ defined over $\cM$ which are of class
$\cC^k$. For ease of notation, we let $\Gamma(E)\eqdef \cC^\infty(E)$
denote the space of smooth sections of $E$. We denote by
$\fX(\cM)\eqdef \Gamma(T\cM)$ the space of smooth vector fields
defined on $\cM$. To avoid cumbersome formulas, we will occasionally
not distinguish between the notations $(x^1,x^2)$ and $(x_1,x_2)$ for
local coordinates on a $2$-manifold. The norms used on natural vector
bundles associated to a Riemannian manifold are induced by a
Riemannian metric given on that manifold.  We denote by
$\varepsilon_{ij}$ the components of the Levi-Civita {\em density} in
local coordinates on a two-manifold $\cM$ and by $\epsilon_{ij}$ the
components of the Levi-Civita {\em tensor} relative to a Riemannian
metric $\cG$ defined on $\cM$. The two are related by the equation:
\be
\epsilon_{ij}=\sqrt{|\cG|}\epsilon_{ij}~~,
\ee
where $|\cG|$ is the Gram determinant of $\cG$ in those local coordinates. 

This paper uses the standard
notations and conventions of modern Riemannian geometry (see
\cite{Petersen}). We refer the reader to \cite{Saunders} for the basic
theory of jet bundles, which is indispensable for understanding nonlinear
differential equations on manifolds. 

\section{The geometric formulation of two-field cosmological models and their complex parameterization}
\label{sec:models}

\noindent By a {\em two-field cosmological model} we mean a classical
cosmological model with two real scalar fields derived from the
following Lagrangian on a spacetime with topology $\R^4$, which is
endowed with its standard differentiable structure and orientation:
\ben
\label{cL}
\cL[g,\varphi]=\left[\frac{M^2}{2} \mathrm{R}(g)-\frac{1}{2}\Tr_g \varphi^\ast(\cG)-V\circ \varphi\right]\vol_g~~.
\een
Here $M$ is the reduced Planck mass and $g$ is a Lorentzian metric on
$\R^4$ (taken to be of ``mostly plus'' signature), while $\vol_g$ and
$\mathrm{R}(g)$ are the volume form and scalar curvature of $g$. The
real scalar fields are described by a map $\varphi:\R^4\rightarrow
\cM$, where $\cM$ is a (generally non-compact) connected, smooth and
paracompact surface without boundary which we assume to be oriented
and endowed with a smooth Riemannian metric $\cG$, while
$V:\cM\rightarrow \R$ is a smooth function which plays the role of
potential for the scalar fields described by $\varphi$. The quantity
$V\circ\varphi$ is the composition of the map $V$ with $\varphi$,
which is a smooth map from $\R^4$ to $\R$. The quantity $\varphi^\ast(\cG)$
denotes the pull-back of $\cG$ by $\varphi$, which is a symmetric
covariant 2-tensor defined on $\R^4$. The operation $\Tr_g$ is the
trace induced by $g$ on the space of such tensors, which is obtained
by lifting either of the indices using $g$ and then taking the trace
of the resulting endomorphism of the tangent bundle of $\R^4$. In
local coordinates on $\cM$ and $\R^4$, we have:
\be
\Tr_g \varphi^\ast(\cG)=g^{\alpha\beta} \cG_{ij} \pd_\alpha \varphi^i\pd_\beta\varphi^j~~,
\ee
which is the usual local form of the Lagrange density for the
nonlinear sigma model with target space $(\cM,\cG)$. Notice that the
formulation given in \eqref{cL} is manifestly coordinate-free and
valid for target surfaces $\cM$ of nontrivial topology. The fact that
the topology of $\cM$ has highly nontrivial effects on the
cosmological dynamics derived from \eqref{cL} was pointed out in
\cite{genalpha} and illustrated in \cite{elem, modular} in the case
when the target space metric $\cG$ is hyperbolic\footnote{Such effects
are rarely discussed in the cosmology literature, most of which
tacitly assumes that the scalar manifold $\cM$ is contractible. There
exists no reason in string theory or supergravity to restrict to
models with a contractible target space.}. See also \cite{ren} and
\cite{grad}.

One usually requires that the target space metric $\cG$ is complete in
order to ensure conservation of energy. For simplicity, we will assume
throughout this paper that $V$ is strictly positive and non-constant
on $\cM$. The Riemannian 2-manifold $(\cM,\cG)$ is called the {\em
  scalar manifold} of the model while the triplet $(\cM,\cG,V)$ is
called its {\em scalar triple}. It is convenient to parameterize the
Lagrangian \eqref{cL} by the quadruplet $\fM\eqdef (M_0,\cM,\cG,V)$,
where:
\be
M_0\eqdef M\sqrt{\frac{2}{3}}
\ee
is the {\em rescaled Planck mass}. We denote by:
\ben
\label{CritDef}
\Crit V\eqdef \{m\in \cM~\vert~(\dd V)(m)=0\}~~\mathrm{and}~~\cM_0\eqdef \cM\setminus \Crit V
\een
the {\em critical} and {\em noncritical} subsets of $\cM$; the first
of these consists of all those points of $\cM$ which are critical
points of the scalar potential. Notice that the set $\Crit V$ is
closed in the manifold topology of $\cM$ and hence its complement
$\cM_0$ is an open subset of $\cM$.  Since $\cM$ is connected, the
latter set is non-empty iff $V$ is not constant on $\cM$.

The two-field cosmological model parameterized by the quadruplet $\fM$
is obtained by assuming that $g$ is an FLRW metric with flat spatial
section:
\ben
\label{FLRW}
\dd s^2_g=-\dd t^2+a(t)^2\sum_{i=1}^3 \dd x_i^2~,
\een
where $a$ is a smooth and positive function of $t$ and by taking
$\varphi$ to depend only on the {\em cosmological time} $t\eqdef
x^0$. Here $(x^0,x^1,x^2,x^3)$ are the Cartesian coordinates on
$\R^4$.  Define the {\em Hubble parameter} of the FLRW metric through:
\be
H(t)\eqdef \frac{\dot{a}(t)}{a(t)}~~,
\ee
where the dot indicates derivation with respect to $t$.

\subsection{The cosmological equation}

\noindent Let $\pi:T\cM\rightarrow \cM$ be bundle projection.
Define the {\em rescaled Hubble function} $\cH:T\cM\rightarrow
\R_{>0}$ of the scalar triple $(\cM,\cG,V)$ through:
\ben
\label{cHdef}
\cH(u)\eqdef \left[||u||^2+2V(\pi(u))\right]^{1/2}~~\forall u\in T\cM~~,
\een
where $||~||$ indicates the fiberwise norm induced on $T\cM$ by the
Riemannian metric $\cG$. Let $\nabla$ be the Levi-Civita connection of
the scalar manifold $(\cM,\cG)$. When $H(t)>0$ (which we assume for
the remainder of the paper), the Euler-Lagrange equations of
\eqref{cL} with the FLRW Ansatz \eqref{FLRW} reduce (see \cite{
  Noether1, Noether2, genalpha}) to the following second order
geometric ODE known as the {\em cosmological equation}:
\ben
\label{eomsingle}
\nabla_t \dot{\varphi}(t)+ \frac{1}{M_0}\cH(\dot{\varphi}(t))\dot{\varphi}(t)+ (\grad V)(\varphi(t))=0~~
\een
(where $\nabla_t\eqdef \nabla_{\dot{\varphi}(t)}$), together with the {\em Hubble condition}:
\ben
\label{Hvarphi}
H(t)= \frac{1}{3 M_0}\cH(\dot{\varphi}(t))~~,
\een
which determines the Hubble parameter for each solution of
\eqref{eomsingle}.  The smooth solutions $\varphi:I\rightarrow \cM$ of
\eqref{eomsingle} (where $I$ is a non-degenerate interval, i.e. an
interval which is not reduced to a point) are called {\em cosmological
  curves}, while their images in $\cM$ are called {\em cosmological
  orbits}.  Given a cosmological curve $\varphi$, relation
\eqref{Hvarphi} determines the conformal factor $a$ of the FLRW metric
up to a multiplicative constant. This formulation reduces classical
cosmological dynamics to the problem of solving the cosmological
equation (given appropriate initial conditions). Notice that the
cosmological equation is well-defined for scalar manifolds $\cM$ of
nontrivial topology, given that its derivation is valid in general.

A cosmological curve $\varphi:I\rightarrow \cM$ need not be an
immersion. Accordingly, we define the {\em singular} and {\em regular}
  parameter sets of $\varphi$ through:
\beqa
I_\sing\eqdef \{t\in I~\vert~\dot{\varphi}(t)=0\}~~,~~I_\reg\eqdef I\setminus I_\sing= \{t\in I~\vert~\dot{\varphi}(t)\neq 0\}~~.
\eeqa
When $\varphi$ is not constant, it can be shown that $I_\sing$ is a discrete subset of $I$.
The sets of {\em critical and noncritical times} of $\varphi$ are defined through:
\beqa
I_\crit\eqdef \{t\in I~\vert~\varphi(t)\in \Crit V \}~~,~~I_\noncrit \eqdef I\setminus I_\crit \eqdef \{t\in I~\vert~\varphi(t)\not\in \Crit V\}~~.
\eeqa
The cosmological curve $\varphi$ is called {\em non-degenerate} if
$I_\sing=\emptyset$ and {\em noncritical} if
$I_\crit=\emptyset$. Noncriticality means that the orbit of $\varphi$
does not meet the critical set of $V$, i.e. $\varphi(I)$ is contained
in the noncritical submanifold $\cM_0$ of $\cM$. It is easy to see
that a cosmological curve $\varphi:I\rightarrow \cM$ is constant iff
its orbit consists of a single point which coincides with a critical
point of $V$, which in turn happens iff there exists some $t\in
I_\sing\cap I_\crit$.  Hence for any non-constant cosmological curve
we have:
\be
I_\sing\cap I_\crit=\emptyset~~.
\ee

\subsection{The complex parameterization of two-field models}
\label{subsec:complexparam}

Let $E\!=\! Sym^2(T^\ast\! \cM)$ be the vector bundle of symmetric
covariant 2-tensors on $\cM$ and $E_+$ be its fiber sub-bundle
consisting of strictly positive-definite tensors. Then
$\Met(\cM)\eqdef \Gamma(E_+)$ is the set of smooth Riemannian metrics
defined on $\cM$. Cosmological models whose target space is a fixed
oriented surface $\cM$ and with rescaled Planck mass $M_0$ are
parameterized by the continuously-infinite set $\Met(\cM)\times
\cC^\infty(\cM,\R_{>0})$ of pairs $(\cG,V)$. This set can be endowed,
for example, with the Fr\'{e}chet topology.

A well-known and convenient way to parameterize Riemannian metrics
$\cG$ on an oriented surface $\cM$ is through the pair $(J,\omega)$
formed by the complex structure $J$ defined by the conformal class
of $\cG$ relative to the orientation of $\cM$ and the volume form
$\omega$ of $\cG$ relative to that orientation. As explained in the
next subsection, this identifies $\Met(M)$ with $\Cx(\cM)\times
\Gamma(L_+)$, where $\Cx(\cM)$ is the set of those complex structures
which induce the given orientation of $\cM$ and $L_+$ is the
sub-bundle of positive open half-lines in the real determinant line bundle
of $\cM$. In this parameterization, the scalar manifold $(\cM,\cG)$
identifies with the pair $(S,\omega)$ consisting of the smooth (and
generally non-compact) Riemann surface $S\eqdef (\cM,J)$ and the
2-form $\omega\in \Omega^2(\cM)=\Omega^{1,1}(S)$. In fact, the metric
$\cG$ is Hermitian with respect to the complex structure $J$ induced
by its conformal class and hence automatically K{\"a}hler since $\dd
\omega=0$ for dimension reasons. Moreover, $\omega$ coincides with its
K\"{a}hler form relative to $J$. Accordingly, the scalar triple
$(\cM,\cG,V)$ identifies with the ordered system $(S,\omega,V)$. If
$\Riemm$ denotes the proper class of all (not necessarily compact)
Riemann surfaces, then the class of all two-field models with
oriented target space fibers over $\Riemm$ with fiber above $S\in
\Riemm$ given by $\Gamma(L_+^S)\times \cC^\infty(S)$, where $L_+^S$ is
the bundle of open positive half-lines over $S$. We will use this {\em
  complex parameterization} in most of the paper. In this formulation,
the information carried by the conformal class of $\cG$ is encoded by
the complex structure of the Riemann surface $S$, while the remaining
information carried by $\cG$ is encoded by its volume form.

The complex parameterization can be made locally explicit by choosing
a local holomorphic coordinate $z$ on $S$ (relative to the complex
structure $J$). In this case, we have:
\be
\omega=\frac{\i}{2} e^{2\phi}\dd z\wedge \dd \bar{z}=\i \cG_{z\bar{z}}\dd z\wedge \dd \bar{z}~~,
\ee
where the {\em conformal factor} $\phi$ of $\cG$ is a locally-defined
smooth function of $z$ and $\bar{z}$ ($\phi$ should not be confused
with the cosmological curves $\varphi(t)$ of the previous
subsection!). Thus:
\be
\cG_{z\bar{z}}=\cG_{\bar{z} z}=\frac{1}{2}e^{2\phi}~~,~~\cG^{z\bar{z}}=\cG^{\bar{z} z}=2 e^{-2\phi}
\ee
and the only non-vanishing Christoffel symbols are:
\beqan
\label{ChristoffelComplex}
\Gamma^z_{zz}&=&\cG^{z\bar{z}}\pd_z \cG_{z\bar{z}}=2\pd_z \phi\nn\\
\Gamma^{\bar{z}}_{\bar{z}\bar{z}}&=&\overline{\Gamma^z_{zz}}=2\pd_{\bar{z}} \phi~~.
\eeqan
Using these relations and the formulas:
\beqa
&&\nabla_t \dot\varphi^i=\ddot\varphi^i+\Gamma^i_{jk}\dot\varphi^j\dot\varphi^k \\
&& \grad_\cG V=(\grad_\cG V)^i\partial_i =(\cG^{ij}\partial_j V)\partial_i~~
\eeqa
for $(\varphi^1,\varphi^2)=(z,\bar z)$, one finds that the cosmological equation
\eqref{eomsingle} is equivalent locally on $\cM$ with the following system of second
order ODEs:
\beqan
\label{CosmEqComplex}
&& \ddot{z}\!+2(\pd_z \phi)(z,\bar{z}) \dot{z}^2\!+\!\frac{1}{M_0}\! \left[e^{2\phi(z,\bar{z})} \dot{z}\dot{\bar{z}}+ 2V(z,\bar{z})\right]^{1/2}\! \dot{z}+2e^{-2\phi(z,\bar{z})} (\pd_{\bar z} V)(z,\bar{z})\!=0\nn\\
&& \ddot{\bar{z}}\!+2(\pd_{\bar z} \phi)(z,\bar{z}) \dot{\bar{z}}^2\!+\!\frac{1}{M_0} \!\left[e^{2\phi(z,\bar{z})}\! \dot{z}\dot{\bar{z}}+ 2V(z,\bar{z})\right]^{1/2} \dot{\bar{z}}+2e^{-2\phi(z,\bar{z})}
  (\pd_{z} V)(z,\bar{z})\!=0~,\nn
\eeqan
which is parameterized by $M_0$, $V$ and $\varphi$. The qualitative
behavior of the solutions of this system depends markedly on the
conformal factor $\phi$. In this paper, we will be interested in those
choices of $\phi$ which ensure that the system admits solutions with
sustained ``strong'' rapid turn and third order slow roll. As we
explain in latter sections, such solutions exist provided that $\phi$
satisfies a certain geometric contact Hamilton-Jacobi equation, which
takes classical contact Hamilton-Jacobi form in the isothermal
coordinates $x_1\eqdef \Re z$ and $x_2\eqdef \Im z$ defined by the
local complex coordinate $z$. This allows one to apply techniques from
contact Hamilton-Jacobi theory to address the strong SRRT selection
problem for such models.

\begin{remark}
The cosmological equation \eqref{eomsingle} admits an equivalent
description as a {\em geometric} dynamical system \cite{Palis,Katok}
defined on the tangent bundle $T\cM$ of the scalar manifold $\cM$. See
\cite{ren} for an account of this equivalence.
\end{remark}

\subsection{The parameterization of scalar manifold metrics through volume forms and complex structures}

Let $\Tw(T\cM)$ be the fiber sub-bundle of the bundle of endomorphisms
$End(T\cM)\simeq T^\ast \cM\otimes T\cM$ whose fiber at a point $m\in
\cM$ consists of those endomorphisms $J_m\in \End(T_m\cM)$ which
satisfy the condition $J_m^2=-\id_{T_m\cM}$. Moreover,
let\footnote{This should not be confused with the moduli space of
complex structures on $\cM$, which is (a certain compactification of)
the quotient of $\Cx(\cM)$ by the group of orientation-preserving
diffeomorphisms of $\cM$.}:
\be
\Cx(\cM)\eqdef \Gamma(\Tw(T\cM))=\{J\in \End(T\cM)~\vert~J^2=-\id_{T\cM}\}
\ee
be the set of complex structures of $T\cM$. Here $\End(T\cM)\eqdef
\Gamma(End(\cM))$.  Recall that any almost complex structure on $T\cM$
is integrable (since its Nijenhuis tensor vanishes for dimension
reasons) and hence is automatically a complex structure on
$\cM$. Every complex structure on $T\cM$ induces a reduction of
structure group from $\GL(2,\R)$ to $\GL(1,\C)=\C^\ast\simeq
\R_{>0}\times\U(1)=\R_{>0}\times \SO(2,\R)$ (which is
homotopy-equivalent with $\SO(2,\R)$) and hence induces an orientation
of $\cM$. Let $\Cx_+(\cM)\subset \Cx(\cM)$ be the subset of those
complex structures which induce the given orientation of $\cM$. Any
such complex structure determines a maximal holomorphic atlas which
makes $\cM$ into a smooth Riemann surface $S\eqdef (\cM,J)$ whose
natural orientation coincides with the original orientation of $\cM$.

Since $\cM$ is two-dimensional and oriented, the metric $\cG$ gives a
reduction of structure group to $\SO(2,\R)=\U(1)\subset
\C^\ast=\GL(1,\C)$ and hence determines a complex structure $J\in
\Cx_+(\cM)$ which makes $\cG$ into a Hermitian metric\footnote{The
last statement follows from Schur's lemma applied to the vector
representation of $\SO(2,\R)$ since by definition $J$ commutes with
all elements of $\SO(2,\R)$.}. Therefore, the complex structure $J$ is
$\cG$-orthogonal, i.e. it satisfies:
\ben
\label{Jorth}
\cG(JX,JY)=\cG(X,Y)\quad\forall X,Y\in \fX(\cM)~~,
\een
which is equivalent with the $\cG$-antisymmetry condition:
\ben
\label{Jasym}
\cG(JX,Y)=-\cG(X,JY)\quad\forall X,Y\in \fX(\cM)~~
\een
because $J^2=-\id_{T\cM}$. The complex structure $J$ depends only on
the conformal equivalence class of $\cG$ since the reduction of
structure group is unchanged if one multiplies $\cG$ by an
everywhere-positive smooth function. Since the K\"ahler form
$\omega_0$ of $\cG$ is closed for dimension reasons, it follows that
$\cG$ is a K\"{a}hler metric on the Riemann surface $S\eqdef (S,J)$. Thus
the volume form of $\cG$ equals $\frac{\omega_0^n}{n!}$ where $n=1$ is
the complex dimension of $S$, i.e. we have $\omega_0=\omega$.  It
follows that $J$ is the unique endomorphism of $T\cM$ which satisfies:
\ben
\label{omegaJ}
\omega(X,Y)=\cG(JX,Y)\quad\forall X,Y\in \fX(\cM)~~.
\een
Using \eqref{Jorth} and the identity $J^2=-\id_{T\cM}$, relation \eqref{omegaJ}
implies:
\be
\omega(X,JY)\!=\!\cG(JX,JY)\!=\!\cG(X,Y)~\Rightarrow ~\omega(JX,JY)\!=\!\omega(X,Y)\quad\forall X,Y\in \fX(\cM)~,
\ee
which together with \eqref{Jorth} shows that $J$ preserves both $\cG$ and $\omega$.

It is well-known that the correspondence which associates $J$ to $\cG$
is a bijecton between the set of conformal equivalence classes of
Riemannian metrics on $\cM$ and the set $\Cx_+(\cM)$ (see
\cite[Sec. 3.11]{Jost}). To describe this geometrically, notice that
the endomorphism $J_m\in \End(T_m\cM)$ at a point $m\in \cM$ is
determined algebraically by the value $\cG_m\in (E_+)_m$ of the metric
$\cG$ at that point through the relation (cf. \eqref{omegaJ}):
\ben
\label{omegaJm}
\omega_m(u,v)=\cG_m(J_mu,v)\quad\forall u,v\in T_m\cM~~.
\een
Recall that $\omega_m$ itself is determined algebraically by $\cG_m$
through the Gram determinant. Let: 
\be
L=\det(T^\ast\cM)\eqdef \wedge^2T^\ast \cM
\ee
be the real determinant line bundle of $\cM$. Since $\cM$ is oriented,
the fibers of $T\cM$ and hence also the fibers of $L$ are equipped
with a canonical orientation and the slit bundle $\dot{L}$ (defined as
the complement of the image of the zero section in $L$) decomposes as:
\be
\dot{L}=L_+\sqcup L_-~~,
\ee
where the fibers of $L_+$ and $L_-$ correspond respectively to the
strictly positive and strictly negative real axes relative to the
orientation of $L$. The volume form $\omega$ of any Riemannian metric
$\cG$ defined on $\cM$ is a smooth section of $L_+$, and any such
section can be realized as the volume form of some metric. Since
\eqref{omegaJ} allows one to recover $\cG$ from $J$ and $\omega$ as:
\ben
\label{cGomega}
\cG(X,Y)=\omega(X,JY)~~,
\een
it follows that \eqref{omegaJm} defines an isomorphism of fiber bundles:
\ben
\label{xi}
\eta:E_+\stackrel{\sim}{\rightarrow} \Tw_+(T\cM)\times L_+
\een
whose fiber $\eta_m$ at $m\in \cM$ sends $\cG_m$ into the pair
$(J_m,\omega_m)$. Here $\Tw_+(T\cM)$ is the fiber sub-bundle of $T\cM$
whose fiber at $m\in M$ consists of those complex structures on the
vector space $T_m\cM$ which induce the same orientation as that
inherited from the orientation of $\cM$. The fiber bundle isomorphism
\eqref{xi} induces a bijection on sections:
\be
\Met(\cM)\simeq \Cx_+(\cM)\times \Gamma(L_+)
\ee
which provides the complex parameterization of target space metrics
used in the previous subsection.

\begin{remark}
Let $(U, x^1,x^2)$ be any positively-oriented local
coordinate system on $\cM$ and denote by:
\be
\epsilon_{ij}\eqdef \omega(\pd_i,\pd_j)=\sqrt{\det \cG} \,\varepsilon_{ij}
\ee
the Levi-Civita tensor, where $\varepsilon_{ij}$ is the Levi-Civita
symbol. We have $J\pd_i=J_i^{\, j}\pd_j$ and relation \eqref{omegaJ}
gives:
\ben
\label{J}
J_{ij}=\epsilon_{ij}\Longrightarrow J\pd_i=\epsilon_i^{\,\, j} \pd_j~~.
\een
For any $X=X^j\pd_j\in \fX(U)$, we have:
\be
JX=(JX)^i\pd_i~~\mathrm{with}~~(JX)^i=J_j^{\,\, i} X^j=\epsilon_j^{\,\,i} X^j~~.
\ee
Notice that $\epsilon_i^{\,\, j}=-\epsilon^j_{\,\, i}$.  The identity
$\epsilon_{ij} \epsilon^{jk}=-\delta^k_i$ implies $\epsilon^i_{\,
  j}\epsilon^j_{\, k}=-\delta^i_k$, which recovers the relation
$J^2=-\id_{T\cM}$.
\end{remark}

\subsection{The modified gradient operator}

Using the complex structure $J$, we define the {\em modified gradient}
operator $\grad_J:\cC^\infty(\cM)\rightarrow \fX(\cM)$ through:
\be
\grad_J V\eqdef J\grad V~~,
\ee
where $\grad V\in \fX(\cM)$ is the gradient vector field of $V$
relative to the Riemannian metric $\cG$.  With this definition, we
have:
\be
(\dd V)(X)=\cG(X,\grad V)=\omega(X,\grad_JV)\quad\forall X\in \fX(\cM)~,
\ee
where we used relation \eqref{cGomega}. Moreover, we have:
\be
||\grad_JV||=||\grad V||=||\dd V||~,
\ee
since $J$ is a $\cG$-orthogonal endomorphism of $T\cM$. 

\begin{remark}
In positively-oriented local coordinates $(U,x^1,x^2)$, we have:
\be
\grad_JV=_U \pd^i V J\pd_i=\epsilon_j^{\,\, i} \pd^j V \pd_i=-\epsilon^{ij} \pd_j V \pd_i~~.
\ee
\end{remark}

\subsection{The adapted frame of the non-critical submanifold}
\label{subsec:adapted}

Following \cite{cons}, we recall the {\em adapted frame} of the scalar
triple $(\cM,\cG,V)$.  This is a frame of $T\cM_0$ defined by $\cG$
and $V$, where $\cM_0$ is the non-critical submanifold of $\cM$
defined in \eqref{CritDef}. Let:
\ben
n\eqdef \frac{\grad V}{||\grad V||}\in \fX(\cM_0)
\een
be the normalized gradient vector field of $V$, which is well-defined
on the non-critical submanifold $\cM_0$. Notice the relation:
\be
||\grad V||=||\dd V||~~,
\ee
where in the right hand side we denote by $||~||$ the fiberwise norm
of the metric induced by $\cG$ on the cotangent bundle of $\cM$. The
relation follows from the fact that the musical isomorphism
$\flat:T\cM\rightarrow T^\ast \cM$ defined by the metric $\cG$ is a
fiberwise isometry.

Let $\tau\in \fX(\cM_0)$ be the normalized vector field defined on
$\cM_0$ which is orthogonal to $n$ and chosen such that the frame
$(n,\tau)$ of $T\cM_0$ is positively oriented at every point of
$\cM_0$. We have:
\ben
\label{taudef}
\tau=(\iota_n\omega\vert_{\cM_0})^\sharp~~,
\een
where $\omega$ is the volume form of $(\cM,\cG)$ and $\sharp:T^\ast
\cM\rightarrow T\cM$ is the other musical isomorphism defined by
$\cG$. Notice that:
\be
\omega=n_\flat\wedge \tau_\flat~~.
\ee

\begin{prop}
\label{prop:tauJ}
We have:
\ben
\tau=J n=\frac{\grad_JV}{||\dd V||}~~,
\een
where $J$ is the complex structure determined by $\cG$ relative to the
orientation of $\cM$.
\end{prop}

\begin{proof}
In a positively-oriented local coordinate system $(U, x^1,x^2)$ of $\cM$, we have:
\be
\omega=\sqrt{\det \cG}\, \dd x^1\wedge \dd x^2=\frac{1}{2}\epsilon_{ij}\dd x^i\wedge \dd x^j~~.
\ee
Writing $n=n^i\pd_i$, we compute: 
\be
\iota_n \omega=\frac{1}{2} \left(\epsilon_{ij} n^i\dd x^j-\epsilon_{ij} n^j\dd x^i \right)=\epsilon_{ij} n^i\dd x^j~~.
\ee
This gives:
\be
\tau=(\iota_n\omega)^\sharp=\epsilon_{ij} n^i\cG^{jk}\pd_k=\epsilon_i^{\,\, k} n^i \pd_k=n^i J\pd_i=J n~~,
\ee
where we used the relation $(\dd x^j)^\sharp=\cG^{jk} \pd_k$.
\end{proof}

\begin{definition}
The oriented orthonormal frame $(n,\tau)$ of $T\cM_0$ is called the
{\em adapted frame} of the oriented scalar triple $(\cM,\cG,V)$.
\end{definition}

\noindent Notice that the gradient flow and level set curves of $V$
determine an orthogonal mesh on the open Riemannian submanifold
$(\cM_0,\cG\vert_{\cM_0})$ of the scalar manifold $(\cM,\cG)$.

\begin{definition}
\label{def:mesh}
The {\em potential mesh} of the oriented scalar triple $(\cM,\cG,V)$
is the oriented orthogonal mesh determined by the gradient flow and
level set curves of $V$ on the noncritical open submanifold $\cM_0$,
where we orient gradient flow lines towards increasing values of $V$
and level set curves such that their normalized tangent vectors
generate the vector field $\tau$.
\end{definition}

\noindent Of course, the gradient flow of $V$ is considered with
respect to the scalar field metric $\cG$. With Definition
\ref{def:mesh}, the frame $(n,\tau)$ is the oriented orthonormal frame
defined by the potential mesh.

\subsection{The Riemannian and modified Riemannian Hessian of $V$}

Recall that the {\em Riemannian Hessian} of $V$ relative to $\cG$ is
the symmetric covariant 2-tensor defined on $\cM$ through:
\be
\Hess(V)\eqdef \nabla \dd V\in \Gamma(Sym^2(T^\ast \cM))~~,
\ee
where $\nabla$ is the Levi-Civita connection of $(\cM,\cG)$. Using the
complex structure $J$, we define the {\em modified Riemannian Hessian}
of $V$ relative to $\cG$ to be the covariant 2-tensor $\Hess_J(V) \in
\Gamma(\otimes^2T^\ast \cM)$ given by:
\be
\Hess_J(V)(X,Y)\eqdef \Hess(V)(X,JY)\quad\forall X,Y\in \fX(\cM)~~.
\ee
For any vector fields $X,Y\in \fX(\cM)$, we set:
\ben
\label{HessXY}
V_{XY}\!\!\eqdef\!\Hess(V\!)(X,Y\!)\!=\!(\nabla_X \dd V\!)(Y\!)\!=
\!X(Y\!(V\!))-(\dd V\!)(\nabla_X Y\!)\!=\!X(Y\!(V\!))-(\nabla_XY\!)(V\!)
\een
and:
\ben
\label{ModHessXY}
{\tilde V}_{XY}\eqdef \Hess_J(V)(X,Y)~~.
\een
Notice that $V_{XY}=V_{YX}$ since $\Hess(V)$ is a symmetric tensor. 

\begin{remark}
In local coordinates $(U,x^1,x^2)$ on $\cM$, we define:
\ben
\label{HessV}
V_{ij}\eqdef V_{\pd_i\pd_j}=\Hess(V)(\pd_i,\pd_j)=\pd_i\pd_jV -\Gamma_{ij}^k \pd_kV
\een
and:
\ben
\label{HessJV}
{\tilde V}_{ij}\eqdef {\tilde V}_{\pd_i\pd_j}= \Hess(V)(\pd_i,J\pd_j)=\epsilon_j^{\,\,k} V_{ki}=-V_{ik} \epsilon^k_{\,\,j} ~~.
\een
In particular, the symmetric part of $\tilde{V}_{ij}$ is given by:
\be
{\tilde V}_{(ij)}=-\frac{1}{2}(V_{jk} \epsilon^k_{\,i}+V_{ik} \epsilon^k_{\,j})~~.
\ee
\end{remark}

When the metric $\cG$ is fixed, the Hessian and modified Hessian give
second-order linear differential operators $f\!\rightarrow \!\Hess(f)$ and
$f\!\rightarrow\! \Hess_J(f)$ acting on $f\in \cC^\infty(\cM)$. When
the scalar potential $V$ is fixed, the maps $\cG\!\rightarrow
\!\Hess(V)$ and $\cG \!\rightarrow \!  \Hess_J(V)$ are quasilinear
first order operators acting on $\cG\in \Met(\cM)$, because the
Christoffel symbols are given by:
\be
\Gamma^i_{jk}=\frac{1}{2} \cG^{im}\left(\pd_j \cG_{mk}+\pd_k\cG_{mj}-\pd_m\cG_{jk} \right)~~
\ee
and the complex structure $J$ is determined by $\cG$ through a purely
algebraic relation. When both $\cG$ and $V$ are allowed to vary, the
Hessian and modified Hessian give geometric differential operators
acting on the pair $(\cG,V)$. Formally, this means that there exist
morphisms of fiber bundles $~\hess:j^1(E_+)\times_\cM j^2(\cO)\rightarrow
\cO$ and $\hess_J:j^1(E_+)\times_\cM j^2(\cO)\rightarrow \cO$ such that:
\be
\Hess(V)=\hess(j^1(\cG),j^2(V))~~,~~\Hess_J(V)=\hess_J(j^1(\cG),j^2(V))~~.
\ee
Here $j^k(F)$ denotes the $k$-th jet bundle of a fiber bundle $F$
defined over $\cM$ while $\cO$ is the trivial real line bundle on
$\cM$. Notice that $\Gamma(\cO)=\cC^\infty(\cM)$, so in particular $V$
is a global section of $\cO$. Moreover, $j^k(s)\in \Gamma(j^k(F))$
denotes the $k$-th jet prolongation of a section $s\in \Gamma(F)$. We
use the same notation for a morphism of fiber bundles and for the map
which it induces on sections of those bundles. We refer the reader to
\cite{Saunders} for the elementary theory of jet bundles.

\begin{remark}
Let $\widehat{\Hess}(V)\in \End^s(T\cM)$ be the {\em Hesse
  endomorphism} of $\cM$, i.e. the $\cG$-symmetric endomorphism of
$T\cM$ defined through:
\be
\Hess(V)(X,Y)=\cG(X,\widehat{\Hess}(V)(Y))\quad\forall X,Y\in \fX(\cM)~~.
\ee
Define the {\em modified Hesse endomorphism} $\widehat{\Hess}_J(V)\in \End(T\cM)$ through:
\be
\Hess(V)(X,Y)=\omega(X,\widehat{\Hess}_J(V)(Y))\quad\forall X,Y\in \fX(\cM)~~.
\ee
Then $\widehat{\Hess}(V)$ is $\cG$-symmetric and we have:
\be
\widehat{\Hess}(V)=J \widehat{\Hess}_J(V)~\Longrightarrow \widehat{\Hess}_J(V)=-J \widehat{\Hess}(V)~~.
\ee
The $\cG$-transpose of $ \widehat{\Hess}_J(V)$ is given by:
\be
\widehat{\Hess}_J(V)^t=\widehat{\Hess}(V) J~~.
\ee
Moreover, we have: 
\be
\Hess_J(V)(X,Y)=\cG(\widehat{\Hess}(X), JY)=\cG(\widehat{\Hess}_J(X),Y)\quad\forall X,Y\in \fX(\cM)~~.
\ee
\end{remark}

\section{The strong SRRT equation}
\label{sec:SRRTeq}

A fundamental problem in the theory of scalar two-field cosmological
models is the {\em selection problem}, which asks for necessary and
sufficient conditions on the target space metric $\cG$ and scalar
potential $V$ defined on a fixed target manifold $\cM$ that suffice to
characterize those two-field models which admit cosmological curves
$\varphi$ with certain desirable properties. A solution to this
problem would allow one to rule out phenomenologically uninteresting
models without having to compute cosmological trajectories (a
computation which is unfeasible in general, in particular since such
trajectories depend on the choice of initial conditions for the
cosmological equation). In this section, we discuss the {\em strong SRRT
  equation}, which originated in the work of \cite{cons} and 
provides a selection criterion for models to admit trajectories
with sustained strong rapid turn and third order slow roll.

\subsection{The strong SRRT conditions}
\label{subsec:SRRTcond}

One phenomenologically-inspired condition that can be imposed on
cosmological curves $\varphi$ is that they have ``sustained {\em
  strong}\footnote{As pointed out in \cite{Achucarro}, one can impose
a weaker form of the SRRT conditions, which requires that the square
of the dimensionless turning rate is much larger than the first three
adiabatic slow roll parameters (see also \cite{LiliaCons}). As pointed
out in \cite{cons} and clarified in \cite{Achucarro}, some older
special models which were claimed to allow for rapid turn inflation
fail to satisfy the strong SRRT condition, but some of them do satisfy
the {\em weak} SRRT condition obtained in this manner.} rapid turn
with third order adiabatic slow roll'' \cite{cons}. This means that
$\varphi$ satisfies the first three adiabatic slow roll conditions,
the condition that the square of its dimensionless turning rate is
much greater than one and the condition that the relative rate of
change of its turning rate is much smaller than one in absolute
value. We call these the {\em strong SRRT conditions}. When the
cosmological equation admits such solutions, then it can realize
various scenarios that are important for early universe cosmology. To
simplify uninteresting factors, we shall take $M=1$ i.e.
$M_0=\sqrt{\frac{2}{3}}$ henceforth. This is justified by the fact
that the cosmological equation admits a similarity which allows one to
absorb $M$ into $\cG$ and $V$ or to set it to any desired non-zero
value (see \cite{ren}).

The precise mathematical formulation of the strong SRRT conditions is
as follow.  Given a noncritical cosmological curve
$\varphi:I\rightarrow \cM_0$ and a regular time $t\in I_\reg$, define
the first, second and third adiabatic slow-roll parameters:
\be
\varepsilon(t)\eqdef-\frac{\dot{H}(t)}{H(t)^2}~~,~~\eta_\parallel(t)\eqdef-\frac{\ddot{\sigma}(t)}{H(t)\dot{\sigma}(t)}~~,~~\xi(t)\eqdef \frac{\dddot{\sigma}(t)}{H(t)^2\dot{\sigma}(t)}~,
\ee
where we denote by $\sigma$ the increasing proper length parameter of $\varphi$
relative to an arbitrary but fixed origin of cosmological time. We
also define the first and second dimensionless turn parameters:
\be
\eta_\perp(t)\eqdef \frac{\Omega(t)}{H(t)}~~,~~\nu\eqdef \frac{\dot{\eta}_\perp(t)}{H(t)\eta_\perp(t)}~~,
\ee
where $\Omega$ is the turning rate of $\varphi$ (see \cite{cons}). We
say that $\varphi$ satisfies the {\em sustained slow roll with strong
  rapid turn} (strong SRRT) conditions at cosmological time $t$ if
  each of the five quantities:
\ben
\label{small}
\varepsilon(t),~|\eta_\parallel(t)|,~|\xi(t)|,~\frac{1}{\eta_\perp(t)^2},~|\nu(t)|~~
\een
is much smaller than one. 

\subsection{The strong SRRT equation}

It was shown in \cite{cons} that a noncritical cosmological curve
$\varphi:I\rightarrow \cM_0$ satisfies the strong SRRT conditions at
a regular time $t\in I_\reg$ only if the {\em strong SRRT equation}
\ben
\label{SRRT}
V_{n\tau}^2V_{\tau\tau}=3V V_{nn}^2
\een
is satisfied at the point $m=\varphi(t)\in \cM_0$ up to corrections of
order at least one in the small parameters listed in \eqref{small}.
If \eqref{SRRT} is satisfied at $m$ with strict
equality, then we say that $m$ is a {\em strong SRRT point} of
$\cM_0$. In this paper, we are interested in {\em fiducial models},
defined as those models for which all points of $\cM_0$ are strong
SRRT points, i.e. whose metric and scalar potential satisfy the strong
SRRT equation exactly everywhere on $\cM_0$. Such models can serve as
starting points for an approximation scheme, in the sense that a model
which allows for cosmological curves that satisfy the strong SRRT
condition should be ``locally close'' to a fiducial model -- namely,
its scalar field metric and potential should be well-approximated by
the scalar field metric and potential of a fiducial model in those
regions of the target manifold where strong SRRT curves are likely to
exist.

Notice that \eqref{SRRT} is a manifestly covariant equation, since the
adapted frame is defined intrinsically in terms of the scalar field
metric, scalar potential and orientation of $\cM$. In principle, this
equation solves the selection problem for strong SRRT inflation by
providing a {\em selection criterion} for fiducial models.  However,
equation \eqref{SRRT} is not immediately practical since it is written
using the adapted frame. To understand it better, we first rewrite it
geometrically in a frame-free form.

\subsection{The frame-free form of the strong SRRT equation}

To write the strong SRRT equation in a more useful
form, recall that:
\be
n\!=\!\frac{1}{||\dd V||_\cG}\grad V\!=_U\!\frac{1}{||\dd V||_\cG} \pd^i V \pd_i ~\quad~,~\quad~\tau\!=\!\frac{1}{||\dd V||_\cG}\grad_JV\!=\!-\frac{1}{||\dd V||_\cG}\epsilon^{ij}\pd_j V \pd_i~~,
\ee
where:
\be
||\dd V||=\sqrt{\pd^iV\pd_iV}~~.
\ee
Using these relations, we find: 
\ben
\label{VHessOp}
V_{n\tau}=\frac{1}{||\dd V||_\cG^2} \cD_1(\cG,V)~~,~~V_{n n}=\frac{1}{||\dd V||_\cG^2} \cD_2(\cG,V)~~,~~V_{\tau\tau}=\frac{1}{||\dd V||_\cG^2} \cD_3(\cG,V)~~,
\een
where: 
\beqan
\label{D1D2D3}
&& \cD_1(\cG,V\!)\!\eqdef\! \Hess_J(V)(\grad V,\grad V)\!=_U\!\!{\tilde V}_{(ij)} \pd^iV\pd^jV~~,\nn\\
&& \cD_2(\cG,V\!)\!\eqdef\! \Hess(V)(\grad V,\grad V)\!=_U\!\! V_{ij}\pd^iV\pd^jV~~,\\
&& \cD_3(\cG,V\!)\!\eqdef\! \Hess(V\!)(J\grad V, \!J\grad V\!)\!=_U\!\! V_{ij}\epsilon^{ik}\epsilon^{jl}\pd_kV\pd_lV\!\!=\!||\dd V||^2 \Delta V\!\!-\!\cD_2(\cG,\!V\!)~~.\nn
\eeqan
Here:
\be
\Delta V=\tr(\widehat{\Hess}(V))=V_i^i=\frac{1}{\sqrt{\det \cG}}\pd_i(\sqrt{\det \cG} \,\pd^i V)=\pd^i\pd_iV-\cG^{ij}\Gamma_{ij}^k\pd_kV
\ee
is the (negative-definite) Laplacian of $V$ with respect to $\cG$. To
derive the last relation of \eqref{D1D2D3}, we used the identity:
\be
\epsilon^{ij}\epsilon^{kl}=\cG^{ik} \cG^{jl}-\cG^{il}\cG^{jk}~~.
\ee
The objects $\cD_k$ are nonlinear scalar differential operators:
\be
\cD_k:\Gamma(E_+)\times \cC^\infty(\cM)\rightarrow \cC^\infty(\cM)
\ee
of order $(1,2)$. Mathematically, this means that there exist
morphisms of fiber bundles:
\be
\cF_k:j^1(E_+)\times_\cM j^2(\cO) \rightarrow \cO
\ee
(with $k=1,2,3$) such that:
\be
\cD_k(\cG,V)=\cF_k(j^1(\cG), j^2(V))~~.
\ee
Substituting \eqref{VHessOp} into \eqref{SRRT} gives {\em the
  frame-free form of the strong SRRT equation}:
\ben
\label{SS}
\cD_1(\cG,V)^2 \cD_3(\cG,V)=3V ||\dd V||^2 \cD_2(\cG,V)^2~~.
\een
Hence \eqref{SRRT} is equivalent with the following geometric PDE: 
\ben
\label{SRRTop}
\cD(\cG,V)=0~~,
\een
where we introduced the nonlinear differential operator
\be
\cD:\Gamma(E_0)\times \cC^\infty(\cM)\rightarrow \cC^\infty(\cM)
\ee
of order $(1,2)$ defined through:
\ben
\label{cD}
{\cD}(\cG,V)\!\eqdef \!3V ||\dd V||^2 \cD_2(\cG,V)^2-\cD_1(\cG,V)^2 \cD_3(\cG,V)~.
\een
When the metric $\cG$ is fixed, relation \eqref{SS} is a nonlinear
second order geometric PDE for the scalar potential $V$. When $V$ is
fixed, it becomes a first order PDE for the scalar field metric $\cG$.

To formally describe the {\em strong SRRT operator} $\cD$, notice that
the map $(\cG,V)\rightarrow V ||\dd V||^2$ is induced by a morphism of
fiber bundles $\cF_0:E_+\times_\cM j^1(\cO)\rightarrow \cO$ which in
turn induces a morphism of jet bundles $\bar{\cF}_0:j^1(E_+)\times_\cM
j^2(\cO)\rightarrow \cO$ upon composing with the fiber product of
the projection morphisms $j^1(E_+)\rightarrow E_+$ and
$j^2(\cO)\rightarrow j^1(\cO)$.

Consider the morphism of fiber bundles
$\cF:j^1(E_+)\times_\cM j^2(\cO)\rightarrow \cO$ defined through:
\ben
\label{cFdef}
\cF\eqdef 3 \bar{\cF}_0 \otimes (\cF_2)^2-(\cF_1)^2 \cF_3~~.
\een
With these notations, we have:
\ben
\label{cDdef}
\cD(\cG,V)=\cF(j^1(\cG),j^2(V))~~.
\een
The strong SRRT operator $\cD$ is well-defined on the entire scalar
manifold $\cM$ and sensitive to its topology.

\subsection{The strong SRRT equation for metrics in a fixed conformal class} 

When written in general local coordinates, the strong SRRT equation
\eqref{SS} is quite formidable. One way to simplify this equation is
to restrict the metric $\cG$ to lie within a fixed conformal class.
As explained in Subsection \ref{subsec:complexparam}, this amounts to
fixing the complex structure $J$ defined by the conformal class of
$\cG$ relative to the orientation of $\cM$.

Recall the isomorphism of fiber bundles \eqref{xi} defined through
relation \eqref{omegaJm} and let
$\eta_J:E_+^J\stackrel{\sim}{\rightarrow} L_+$ be its $J$-section,
i.e. the isomorphism of fiber bundles obtained by restricting $\eta$
to the conformal class which corresponds to a complex structure
$J$. Here $E_+^J\subset E_+$ is the sub-bundle of $E_+$ whose sections
correspond to metrics within the conformal class corresponding to $J$;
its fibers are open half-lines sitting inside the fibers of $E_+$
(which are full open cones of dimension 3). On global sections, this
induces the map which sends a metric $\cG$ which belongs to the
conformal class specified by $J$ to its volume form $\omega$. The
inverse of $\eta_J$ allows us to transport $\cF_k$ to $J$-dependent
fiber bundle morphisms:
\be
F_k:=F_k^J\eqdef \cF_k\circ [j^1(\eta_J^{-1})\times_\cM \id_{j^2(\cO)}]:j^1(L_+)\times j^2(\cO)\rightarrow \cO
\ee
and hence to transport $\cD_k$ to $J$-dependent differential operators:
\be
D_k:=D_k^J \eqdef F_k\circ (j^1\times j^2):\Gamma(L_+)\times \cC^\infty(\cM)\rightarrow \cC^\infty(\cM)
\ee
acting on pairs $(\omega,V)$. They act as first order operators on
$\omega$ and as second order nonlinear operators on $V$. Also consider
the fiber bundle morphism:
\be
F_0:=F_0^J\eqdef \cF_0\circ (\eta_J^{-1}\times \id_{j^1(\cO)}):L_+\times j^1(\cO)\rightarrow
\cO
\ee
and the bundle morphism $\bar{F}_0:j^1(L_+)\times j^2(\cO)\rightarrow
\cO$ obtained by composing it with the appropriate bundle
projections. Using these objects, the strong SRRT equation
\eqref{SRRTop} becomes:
\ben
\label{SRRTconf}
D(\omega, V)=0~~,
\een
where:
\be
D(\omega,V):=D^J(\omega,V)\eqdef 3 F_0(\omega, j^1(V)) D_2(\omega,V)^2-D_1(\omega,V)^2 D_3(\omega,V)
\ee
is the $J$-dependent operator obtained by transporting $\cD$ to an
operator defined on $\Gamma(L_+)\times \Gamma(\cO)$ using the bundle
isomorphism $\eta_J$. We have:
\be
D(\omega,V)=F(j^1(\omega),j^2(V))~~,
\ee
where:
\ben
\label{Fdef}
F:=F^J\eqdef \cF\circ [j^1(\eta_J^{-1})\times \id_{j^2(\cO)}]:j^1(L_+)\times j^2(\cO)\rightarrow \cO~~.
\een
Once again, notice that $D$ is a geometric differential operator
which is globally well-defined and sensitive to the topology of
the scalar manifold $\cM$.

\subsection{The strong SRRT equation with fixed $V$ for metrics in a fixed conformal class}
Let
\ben
\label{FVdef}
F_V:=F^J_V:j^1(L_+)\rightarrow \R
\een
be the $j^2(V)$-section of the function \eqref{Fdef}. When $V$ is
fixed, the strong SRRT equation \eqref{SRRTconf} for volume forms of
metrics $\cG$ belonging to the conformal class determined by $J$
reads:
\ben
\label{HJcov}
F_V(j^1(\omega))=0~~.
\een
One can immediately recognize this as {\em the stationary contact
  Hamilton-Jacobi equation} of the geometric contact Hamiltonian
system $(j^1(L_+), F_V)$, where the jet bundle $j^1(L_+)$ is endowed
with its Cartan contact structure (see \cite{Saunders}). We refer the
reader to \cite{GG, BCT, LB, LTW, LV, LLLR, VL, Guijarro, SLVD} for
background on the geometric formulation of contact Hamilton-Jacobi
mechanics and to the classic references \cite{Arnold, Geiges, LM} for
background on contact geometry and contact topology. The formulation
of contact mechanics given in \cite{GG} applies most directly to our
situation. This observation allows one to apply results and techniques
from that subject to the study of the strong SRRT equation. As we will
see below, equation \eqref{HJcov} is a {\em proper} first order PDE,
so one can approach it using the theory of viscosity solutions (see
\cite{CL, CEL, BD, EvansBook, Tran, CIL, C, Katzourakis}).

\section{The contact Hamiltonian in isothermal Liouville coordinates}
\label{sec:HJ}

We can make the formulation of the previous section less abstract (but
also less geometric) by extracting the form of the contact Hamiltonian
$F_V$ in certain natural coordinates on $j^1(L_+)$ induced by local
isothermal coordinates\footnote{ Recall that any oriented Riemannian
2-manifold $(\cM,\cG)$ (not necessarily compact) admits an atlas
consisting of positively-oriented isothermal coordinate charts, which
corresponds to a holomorphic atlas for the complex structure $J$. See
for example \cite{Jost}.} $(U, x^1,x^2)$ relative to the complex
structure $J$ (such coordinates are the real and imaginary parts of a
local complex coordinate $z=x^1+\i x^2$ defined on $U$).

\subsection{Isothermal coordinates on the scalar manifold}

In local isothermal coordinates $(U,x^1,x^2)$ for the Riemann surface
$S\eqdef (\cM,J)$, the metric $\cG$ and its volume form are given by:
\beqan
\label{isothermal}
&&\cG=e^{2\phi} \cG_0~,~~~\dd s^2= e^{2\phi}\dd s_0^2=e^{2\phi}(\dd x_1^2+\dd x_2^2)~~,\nn\\
&&~~\omega=e^{2\phi} \omega_0=e^{2\phi}\dd x^1 \wedge \dd x^2=\frac{\i}{2} e^{2\phi} \dd z\wedge \dd \bar{z}~~,
\eeqan
where $\phi\in \cC^\infty(U)$ and $\cG_0$ is the flat metric on $U$
with squared line element and volume form given by:
\ben
\label{cG0}
\dd s_0^2=\dd x_1^2+\dd x_2^2=|\dd z|^2~~,~~\omega_0=\dd x^1\wedge \dd x^2\in \Gamma(L_+\vert_U)~~.
\een
Thus $\cG$ is Weyl-equivalent on $U$ to the flat metric $\cG_0$.  In
such coordinates, all remaining information about the scalar field
metric $\cG$ is locally encoded by the {\em conformal factor}
$\phi\in\cC^\infty(U)$, which carries the same information as the
restriction of the volume form $\omega$ to the open set $U$. The
information carried by the conformal class of $\cG$ is encoded in the
choice of isothermal coordinate system, since the latter is determined
(up to a local biholomorphism) by the local complex coordinate
$z=x_1+\i x_2$, which in turn depends on the complex structure $J$.

The Christoffel symbols of $\cG$ are given by:
\ben
\label{ChristoffelIsothermal}
\Gamma_{ij}^k=\delta_i^k \pd_j \phi +\delta_j^k \pd_i \phi-\delta_{ij}\pd_k\phi~~,
\een
while its Gaussian curvature (which equals half of the scalar curvature) takes the form:  
\ben
\label{curvature}
K=- e^{-2\phi} \Delta \phi~~.
\een
The Levi-Civita tensor of $\cG$ is:
\ben
\label{epsilon}
\epsilon_{ij}=\omega(\pd_i,\pd_j)=f\varepsilon_{ij}\Longrightarrow \epsilon^i_{\, j}=\varepsilon_{ij}~~,
\een
where $f\eqdef e^{2\phi}$. The complex structure $J$ is given on $U$ by the conditions: 
\be
J\pd_1=\pd_2~~,~~J\pd_2=-\pd_1~~,~~\mathrm{i.e.}~~J\pd_i=\varepsilon_{ij}\pd_j~~.
\ee
Since $J\pd_i=J_i^{\,j} \pd_j$, this gives $J_i^{\,j}=\varepsilon_{ij}=\epsilon^i_{\, j}=-\epsilon^j_{\, i}$ and:
\be
J_{ij}=\epsilon_{ij}=f\varepsilon_{ij}~~.
\ee
In particular, we have:
\be
\omega_{ij}=\epsilon_{ij}=J_{ij}~~,
\ee
i.e. $\omega(X,Y)=\cG(JX,Y)$ for all $X, Y\in \fX(U)$, which recovers
relation \eqref{omegaJ}.

\subsection{Isothermal Liouville coordinates}

The total space of the restricted jet bundle $j^1(L_+)\vert_U$ admits
coordinates $(x^1,x^2, h, h_1,h_2)$ associated to the isothermal
coordinates $(U,x^1,x^2)$ of $\cM$, where $h$ is the fiber
coordinate on $L_+\vert_U$ defined through the relation (here $\pi_L:L\rightarrow \cM$
is the bundle projection of $L$):
\be
v= h(v) \omega_0\vert{\pi_L(v)}\quad\forall v\in L_+\vert_U
\ee
and $h_1$, $h_2$ are the fiber coordinates of the bundle $j^1(L_+)\vert_U\rightarrow L_+\vert_U$
corresponding to the formal partial derivatives of
$h$ with respect to $x_1$ and $x_2$. Since the isothermal coordinate
system is positively-oriented, we have $\dd x_1\wedge \dd
x_2\vert_U\in \Gamma(L_+\vert_U)$, which implies $h(v)>0$ for all
$v\in L_+\vert_U$.  The volume form $\omega_0=\dd x^1\wedge \dd x^2$
of $\cG_0$ gives a local trivialization of $L$ above $U$ which allows
us to identify the restriction of $\omega\in \Gamma(L_+)$ to the open
set $U$ with the function $f=h\circ\omega=e^{2\phi}\in
\cC^\infty(U)$. In these coordinates on $j^1(L_+)\vert_U$, the basic
contact form $\Theta\in \Omega^1(j^1(L_+)\vert_U)$ is given by:
\be
\Theta=\dd h-h_1 \dd x^1-h_2\dd x^2~~.
\ee
By general theory, this form is a basis for the contact ideal of
$j^1(L_+)\vert_U$. In particular, any contact $1$-form $\theta$ defined on
$j^1(L_+)\vert_U$ has the form $\theta=\psi \Theta$, where $\psi\in
\cC^\infty(j^1(L_+)\vert_U)$ is a smooth function defined on the total space
of the bundle $j^1(L_+)\vert_U$.

\begin{definition}
The {\em isothermal Liouville coordinates} $(x^1,x^2,u,p_1,p_2)$ on
the total space of $j^1(L_+)\vert_U$ determined by the local
isothermal coordinates $(U,x^1,x^2)$ of $S=(\cM,J)$ are defined
through:
\be
u\eqdef \frac{1}{2}\log h~~,~~p_1\eqdef \fd_1 u= \frac{1}{2}\frac{h_1}{h}~~,~~p_2\eqdef \fd_2 u=\frac{1}{2}\frac{h_2}{h}~~,
\ee
where $\fd_1 $ and $\fd_2 $ stand for the formal partial derivatives with respect to $x^1$ and $x^2$. 
\end{definition}

\noindent We refer to $x^1,x^2$ as {\em position coordinates} and to
$p_1,p_2$ as {\em momentum coordinates}. The coordinate $u$ is called
the {\em action coordinate}. Notice the relations:
\be
u\circ \omega=\frac{1}{2}\log f=\varphi~~,~~p_1\circ \omega=\frac{1}{2}\frac{\pd_1 f}{f}=\pd_1\varphi~~,~~p_2\circ\omega=\frac{1}{2}\frac{\pd_2f}{f}=\pd_2\varphi~~.
\ee
It is convenient to define complex-valued fiber coordinates for the
bundle $j^1(L_+)\vert_U\rightarrow L_+\vert_U$ through:
\be
p_z\eqdef \frac{1}{2}(p_1-\i p_2)~~,~~p_{\bar z}\eqdef \frac{1}{2}(p_1+\i p_2)~~,~~
\ee
which correspond to the formal partial derivatives of $u$ with respect
to $z$ and $\bar{z}$. Then $(u,p_z,p_{\bar z})$ are complex-valued
coordinates on the fibers of $j^1(L_+)\vert_U\rightarrow U$ such that
${\bar u}=u$ and $p_{\bar z}=\overline{p_z}$, while $(z,u,p_z,p_{\bar
  z})$ are complex-valued coordinates on the total space of
$j^1(L_+)\vert_U$, where $u$ happens to be valued in $\R$. In
isothermal Liouville coordinates, the basic contact form $\Theta$ is
given by:
\be
\Theta=2 e^{2u} \theta~~, 
\ee
where the contact form:
\be
\theta\eqdef \dd u -p_1 \dd x^1-p_2\dd x^2=\dd u-p_z \dd z-p_{\bar z} \dd \bar{z}
\ee
is the basic contact form in the coordinates $(x^1,x^2, u,p_1,p_2)$
and $(z,\bar{z},u,p_z,p_{\bar z})$ on $j^1(L_+)\vert_U$. In
particular, $(x^1,x^2, u,p_1,p_2)$ are Liouville coordinates for
$\theta$. The Cartan distribution on $j^1(L_+)\vert_U$ coincides with
the kernel distribution of $\theta$ and is given by the completely
non-integrable Pfaffian equation:
\be
\dd u=p_1\dd x^1+p_2 \dd x^2~,~\mathrm{i.e.}~~\dd u=p_z\dd z+p_{\bar z}\dd \bar{z}~~.
\ee
This distribution is globally well-defined and makes $j^1(L_+)$ into a
five-dimensional contact manifold.

In isothermal Liouville coordinates, the geometric contact Hamiltonian $F_V$
becomes a function of the variables $x^1,x^2, u,p_1, p_2$ and the
strong SRRT equation takes the local form:
\be
F(x^1,x^2,\phi(x^1,x^2),\pd_1 \phi(x^1,x^2), \pd_2\phi(x^1,x^2))=0~~,
\ee
being a nonlinear first order PDE of the conformal factor $\phi$. Here
$F$ is the local form of the map $F_V$ of \eqref{HJcov} in isothermal
Liouville coordinates. The equation has classical contact
Hamilton-Jacobi form, where $F_V$ plays the role of Hamiltonian and
$\phi$ plays the role of Hamilton-Jacobi action.  In practice, this
local form can be extracted by substituting expression
\eqref{isothermal} for the scalar field metric in isothermal
coordinates into the frame-free form \eqref{SS} of the strong SRRT
equation.

In the next subsection, we extract the explicit form of this equation
for the conformal factor $\phi$.  Before proceeding with that
derivation, let us discuss the change of isothermal Liouville
coordinates induced by a change of isothermal coordinates of
$S=(\cM,J)$.

Let $(U', y^{1},y^{2})$ be another isothermal coordinate chart of
$(\cM,\cG)$ with $U\cap U'\neq \emptyset$ and associated complex
coordinate $w=y^{1}+\i y^{2}$. Then:
\be
\dd s^2\vert_{U'}=e^{2\phi'}(\dd y_1^2+\dd y_2^2)=e^{2\phi'}|\dd w|^2
\ee
for some $\phi'\in \cC^\infty(U')$. Moreover, we have $w=\chi\circ z$
for some biholomorphism $\chi:x(U\cap V)\rightarrow y(U\cap U')$ and
the following relation holds on $U\cap U'$:
\be
\phi=\phi' +\frac{1}{2}\log \det \left(\frac{\pd y^{i}}{\pd x^j}\right)_{i=1,2}=\phi'+\log |\lambda(z)|~~,
\ee
where:
\be
\lambda(z)\eqdef \frac{\dd w(z)}{\dd z}=\chi'(z)~~.
\ee
The holomorphicity condition $\frac{\pd w}{\pd \bar{z}}=0$ amounts to the
Cauchy-Riemann equations:
\be
\frac{\pd y^{1}}{\pd x^1}=\frac{\pd y^{2}}{\pd x^2}~~,~~\frac{\pd y^{1}}{\pd x^2}=-\frac{\pd y^{2}}{\pd x^1}~~.
\ee
Hence the change of complex coordinates $z\rightarrow w$ induces the transformation:
\ben
\label{varphitf}
\phi(z)\rightarrow \phi'(w)=\phi(z)-\log |\lambda(z)|~~,
\een
where $\lambda$ is a holomorphic function of $z$ which does not vanish
on the open set $z(U\cap U')\subset \C$. We have:
\be
\pd_z\phi'= \pd_z \phi-\frac{\lambda'}{2\lambda}~~,~~\pd_{\bar z}\phi'= \pd_{\bar z} \phi-\frac{\bar{\lambda}'}{2\bar{\lambda}}
\ee
and:
\be
\pd_{w} \phi'=\frac{1}{\lambda}\pd_z\phi'~~,~~\pd_{\bar w}\phi'=\frac{1}{\bar{\lambda}}\pd_{\bar z} \phi'~~.
\ee
This corresponds to the following change of complex coordinates on the jet bundle $j^1(L_+)\vert_{U\cap U'}$:
\be
z\rightarrow w~~,~~u\rightarrow u'\eqdef u -\log |\lambda|~~,~~p_z\rightarrow p_{w}=\frac{1}{\lambda} \left(p_z-\frac{\lambda'}{2\lambda}\right)~~,~~p_{\bar z}\rightarrow p_{\bar w}=\frac{1}{\bar{\lambda}}
\left(p_{\bar z}-\frac{\overline{\lambda'}}{2\overline{\lambda}}\right)~~.
\ee
Notice that this change of coordinates preserves the form of $\theta$:
\be
\theta=\dd u-p_z \dd z-p_{\bar{z}}\dd \bar{z}=\dd u'-p_w \dd w-p_{\bar{w}}\dd \bar{w}
\ee
and hence $(y^{1},y^{2},u, p'_1\eqdef 2\Re(p_w), p'_2\eqdef -2 \Im
(p_w))$ are isothermal Liouville coordinates on $j^1(L_+)\vert_{U\cap U'}$. Thus
a change of isothermal coordinates on the Riemann surface $S=(\cM,J)$
induces a change of local Liouville coordinates on $j^1(L_+)$, i.e. a
(passive) contact transformation.

\subsection{Explicit form of the contact Hamiltonian}

Let $(U,x^1,x^2)$ be an isothermal coordinate chart of $S=(\cM,J)$
such that $U$ is contained in the noncritical submanifold $\cM_0$.
Let $\grad_0, \grad_0^J:\cC^\infty(U)\rightarrow \cC^\infty(U)$ be the
gradient and modified gradient operators defined by the local flat
metric $\cG_0\in \Met(U)$:
\ben
\label{grad0}
\grad_0 V=_U e^{2\phi}\pd^iV \pd_i= \pd_iV \pd_i ~~,~~\grad^J_0 V=_U e^{2\phi} J^{\,\,i}_j\pd^j V \pd_i =-\varepsilon_{ij} \pd_j V \pd_i
\een
and let $D_0, {\tilde D}_0:\cC^\infty(U)\rightarrow \cC^\infty(U)$ be the operators defined through:
\beqan
\label{DDbarDef}
D_0 V &\eqdef_U& \Hess(V)(\grad_0V, \grad_0V)= _U V_{ij}\pd_iV\pd_jV~~\nn\\
{\tilde D}_0V&\eqdef_U& \Hess_J(V)(\grad_0V,\grad_0V)=_U{\tilde V}_{ij}\pd_iV\pd_jV~~.
\eeqan
Moreover, let $\Delta_0:\cC^\infty(U)\rightarrow \cC^\infty(U)$ be the
(negative-definite) Laplacian defined by the local flat metric
$\cG_0$:
\ben
\Delta_0 V =_U (\pd_1^2 +\pd_2^2) V~~.
\een
Finally, let:
\ben
||\dd V||_0 =_U \sqrt{(\pd_1V)^2+(\pd_2 V)^2}
\een
be the norm of $\dd V$ with respect to $\cG_0$.

\begin{lemma}
\label{lemma:Liso}
The strong SRRT equation \eqref{SS} is equivalent on $U$ with the following equation:
\ben
\label{SSThermal}
({\tilde D}_0V)^2 [||\dd V||_0^2\Delta_0V -D_0V]=3 e^{2\phi} V ||\dd V||_0^2 (D_0V)^2~~.
\een
\end{lemma}

\begin{proof}
We have: 
\be
D_2 (\cG,V)=e^{-4\phi} D_0V~~,~~D_1(\cG,V)=e^{-4\phi}{\tilde D}_0V~~,~~||\dd V||^2=e^{-2\phi} ||\dd V||^2_0
\ee
and the formula for change of the Laplacian under a conformal transformation gives:
\be
\Delta V=e^{-2\phi} \Delta_0V~~.
\ee
Using these relations, the strong SRRT equation \eqref{SS} reduces on $U$ to equation \eqref{SSThermal}.
\end{proof}

Let $\cdot$ be the scalar product defined by $\cG_0$ (which satisfies
$\pd_i\cdot \pd_j=\delta_{ij}$) and $\Hess_0(V)\in \cC^\infty(U)$ be
the Riemannian Hessian of $V$ with respect to the local flat metric $\cG_0$:
\be
\Hess_0(V)\eqdef_U\nabla_0\dd V= \pd_i\pd_jV \dd x^i\otimes \dd x^j\in \Gamma(\otimes^2 T^\ast U)~~,
\ee
where $\nabla_0$ is the Levi-Civita connection of $\cG_0$. 

\begin{lemma}
\label{lemma:DDbar}
We have:
\beqan
\label{DDbar}
&& D_0V=\Hess_0(V)(\grad_0V,\grad_0V)-||\dd V||_0^2 \grad_0 \phi \cdot \grad_0V~~\nn\\
&& {\tilde D}_0V =\Hess_0(V)(\grad_0V, \grad_0^JV)-||\dd V||_0^2 \grad_0 \phi \cdot \grad^J_0V~~.
\eeqan
\end{lemma}

\begin{proof}
In isothermal coordinates, the Hessian tensor of $V$ is given by:
\ben
V_{ij}=\pd_i\pd_jV-\Gamma_{ij}^k \pd_kV=\pd_i\pd_jV+\delta_{ij}\pd_k\phi\pd_kV-\pd_i \phi\pd_jV-\pd_j \phi\pd_i V~~,  
\een
where we used relations \eqref{HessV} and
\eqref{ChristoffelIsothermal}. Using this in \eqref{DDbarDef} gives
\eqref{DDbar}. 
\end{proof}

\begin{remark}
Using \eqref{grad0}, we can write \eqref{DDbar} as: 
\beqan
\label{D0rels}
D_0V&=&\pd_i\pd_jV\pd_iV\pd_jV- ||\dd V||_0^2 \pd_i \phi \pd_iV~~\nn\\
{\tilde D}_0V&=&\!\!-\left[\pd_i\pd_jV\pd_iV(\varepsilon_{jk}\pd_kV)-||\dd V||_0^2 \pd_i \phi (\varepsilon_{ij} \pd_jV)\right]~~.
\eeqan
\end{remark}

\begin{prop}
The strong SRRT equation \eqref{SSThermal} is equivalent on $U$ with the
following equation:
\ben
\label{SSThermalRed}
[\grad_0^JV \cdot \grad_0\phi-\tilde{H}_0]^2 [\grad_0V \cdot \grad_0\phi+\Delta_0V - H_0]=3 e^{2\phi} V  [\grad_0V \cdot \grad_0\phi-H_0]^2~~
\een
where we defined:
\beqa
&& H_0\eqdef \frac{\Hess_0(V)(\grad_0 V,\grad_0 V)}{||\dd V||_0^2}=
\frac{(\partial_i\partial_j V)(\partial_i V)(\partial_j V)}{||\dd V||_0^2}~, \\
&& \tilde H_0\eqdef \frac{\Hess_0(V)(\grad_0 V,\grad_0^J V)}{||\dd V||_0^2}=
\frac{-(\partial_i\partial_j V)(\partial_i V)\varepsilon_{jk}(\partial_j V)}{||\dd V||_0^2}~.
\eeqa
\end{prop}

\begin{proof}
Follows immediately from Lemmas \ref{lemma:Liso} and \ref{lemma:DDbar}.
\end{proof}

Let $U_0\subset \R^2$ be the image of $U$ in the isothermal coordinate
chart $(U,x^1,x^2)$. The isothermal Liouville coordinates $(x^1,x^2,
u,p_1,p_2)$ induce an isomorphism of fiber bundles
$j^1(L_+)\vert_U\simeq U_0\times \R\times \R^2$, which, by a common
abuse of notation, allows us to take $x=(x^1,x^2)\in U_0$, $u\in \R$
and $p\eqdef (p_1,p_2)\in \R^2$. Consider the smooth functions
$A,B:U_0\times \R^2\rightarrow \R$ defined through:
\ben
\label{ABdef}
A(x,p)\eqdef (\grad_0 V)(x)\cdot p=(\pd_iV)(x) p_i~~,~~B(x,p)\eqdef \grad_0^JV \cdot p=-\epsilon_{ij} (\pd_jV)(x) p_i~~.
\een
Notice that $A$ and $B$ are linear functions of degree one in $p_1$
and $p_2$ whose coefficients depend smoothly on $x\in U_0$. The linear
transformation $\R^2\ni (p_1,p_2)\rightarrow (A,B)\in \R^2$ has the
matrix form:
\be
\left[\begin{array}{c} A\\ B \end{array}\right]=\left[\begin{array}{cc} \pd_1V & \pd_2 V \\ -\pd_2V & \pd_1 V \end{array}\right] \left[\begin{array}{c}p_1\\p_2\end{array}\right]
\ee
and its determinant $||\dd V||_0^2=(\pd_1V)^2+(\pd_2V)^2$ cannot
vanish on $U_0$ since $U$ is a subset of the noncritical submanifold
$\cM_0$.  Hence this linear transformation is bijective everywhere on
$U_0$, with inverse given by:
\ben
\label{pAB}
p_1=\frac{\pd_1V A-\pd_2V B}{(\pd_1V)^2+(\pd_2V)^2}~~,~~p_2=\frac{\pd_2V A+\pd_1V B}{(\pd_1V)^2+(\pd_2V)^2}~~.
\een

\begin{thm}
In isothermal Liouville coordinates $(x^1,x^2,u,p_1,p_2)$ on
$j^1(L_+)\vert_U$, the contact Hamiltonian $F_V$ of \eqref{FVdef}
corresponds to the smooth function $F:U_0\times \R^3\rightarrow \R$
given by:
\beqan
\label{Floc}
F(x,u,p)\!\!\!\!\!&\eqdef&\!\!\!\!\! -[B(x,p)\!-\!\tilde{H}_0(x)]^2 [A(x,p)\!+\!(\Delta_0V)(x)\! -\! H_0(x)]\!+\!3 e^{2 u} V  [A(x,p)\!-\!H_0(x)]^2\nn\\
\eeqan
In particular, the contact Hamilton-Jacobi equation \eqref{HJcov} has
the following local form in these coordinates:
\ben
\label{Feq}
F(x_1,x_2,\phi, \pd_1\phi,\pd_2\phi)=0~~.
\een
\end{thm}

\begin{proof}
Since: 
\ben
\label{AB}
A=\grad_0 V\cdot \grad_0\phi~~,~~B=\grad_0^J V\cdot \grad_0\phi~~,
\een
an easy computation shows that equation \eqref{SSThermalRed} is equivalent with \eqref{Feq}.
\end{proof}

We call \eqref{Feq} the {\em contact Hamilton-Jacobi equation for the
  conformal factor}. Notice that $F$ is a polynomial function of
degree $3$ in $p_1$ and $p_2$ whose coefficients depend smoothly on
the point $(x,u)\in U_0\!\times \!\R$.  In the next sections, we fix
the scalar potential $V$ and analyze \eqref{Feq} as a nonlinear first
order PDE for the conformal factor $\phi$.

\section{Analysis of the contact Hamilton-Jacobi equation for the conformal factor}
\label{sec:nonlinear}

In this section, we discuss some aspects of the contact Hamilton-Jacobi
equation for the conformal factor, including local existence and
unicity of solutions for its Dirichlet problem using the method of
characteristics. We also show that the contact Hamilton-Jacobi
operator is {\em proper} in the sense of the theory of viscosity
solutions and discuss a few general features of such solutions, giving
some numerically computed examples. 

\subsection{The momentum curve}
\label{subsec:curve}

The condition $F(x,u,p)=0$ defines a cubic curve $C_{x,u}$ in the
$p$-plane which depends on the parameters $(x,u)\in U_0\times \R$; we
will call this the {\em momentum curve} at the point $(x,u)$. In the
linear coordinates $A,B$ on the $p$-plane defined in \eqref{ABdef},
the momentum curve has equation:
\beqa
&& A B^2 -3e^{2 u} V A^2 -2 {\tilde H}_0 A B+({\tilde H}_0^2+ 6 e^{2 u} H_0 V) A-2{\tilde H}_0(\Delta_0 V -H_0) B\nn\\
&&+B^2(\Delta_0 V -H_0)+{\tilde H}_0^2 (\Delta_0V-H_0)-3e^{2 u} H_0^2V=0~~,
\eeqa
whose only cubic term is $AB^2$. We have:
\be
\label{F}
F\!=\! -A B^2 \!+3V e^{2u}\! A^2\!-(\Delta_0V\!-\!H_0) B^2\!+\!2{\tilde H}_0 A B -(6 V e^{2u} H_0+{\tilde H}_0^2) A\!+\! 2{\tilde H}_0(\Delta_0 V\!-\!H_0) B+F_0~,
\ee
where: 
\be
F_0={\tilde H}_0^2 (\Delta_0V\!-\!H_0)-3V e^{2 u} H_0^2
\ee
is the momentum-independent term. When $F_0\neq 0$, the momentum curve
does not contain the origin of the $p$-plane. 

\subsection{Properness of the contact Hamilton-Jacobi equation for $\phi$. Viscosity solutions}

Notice that $F$ depends
linearly on $e^{2u}$, which enters through the term:
\be
3V(H_0- A)^2e^{2u}~~.
\ee
Since $V>0$, it follows that $F$ is a monotonously increasing
function of $u$ when $x$ and $p$ are fixed. Hence our contact
Hamilton-Jacobi equation \eqref{Feq} is a {\em proper} first order PDE
in the sense of \cite{CIL} and we can apply to it the theory of
viscosity solutions. Viscosity solutions of weakly-elliptic nonlinear
PDEs of first and second order were introduced by Crandall, Evans and
Lyons in \cite{CL,CEL} motivated by the method of vanishing
viscosity. We refer the reader to \cite{BD}, \cite[Chap.10]{EvansBook}
and \cite{Tran} for the first order theory and to
\cite{CIL,C,Katzourakis} for the second-order theory.

Viscosity solutions are a special kind of weak solutions which are
particularly useful in the study of proper contact Hamilton-Jacobi
equations. By definition, they are continuous (but generally not
differentiable) functions which satisfy the equation in a generalized
sense involving certain inequalities related to a weak form of the
maximum principle. Under reasonably mild conditions, properness of a
contact Hamilton-Jacobi equation implies that a consistent viscosity
version of the Dirichlet problem admits a unique viscosity
solution. Moreover, approximations of that viscosity solution can be
obtained through the {\em method of vanishing viscosity}, which in our
case proceeds by adding a small {\em viscosity term} $-\fc\Delta_0
\phi$ to the opposite of \eqref{Feq}, thereby replacing the contact
Hamilton-Jacobi equation for $\phi$ with its {\em viscosity
  perturbation}:
\ben
\label{viscFeq}
F(x_1,x_2,\phi,\pd_1\phi,\pd_2\phi)-\fc\Delta_0\phi =0~~.
\een
The quantity $\fc>0$ is a small {\em positive} real number called the
{\em viscosity parameter}. One then computes the classical solutions
$\phi_\fc$ of the given Dirichlet problem for equation \eqref{viscFeq}
and considers its limit as $\fc\rightarrow 0$. Under reasonably mild
conditions, one can show that this recovers the viscosity solution of
the viscosity version of the initial Dirichlet problem of \eqref{Feq}.

Since a viscosity solution need not be of class $\cC^1$, its partial
derivatives may be discontinuous on a locus of measure zero --- in
which case one says that the solution has {\em shocks} along that
locus. In general, a locally-defined classical solution of \eqref{Feq}
cannot be extended to a globally-defined classical solution, an
observation which is well-known in the subject of nonlinear PDEs and
arises in the local theory due to the phenomenon of ``crossing
characteristics''. If one attempts to extend such a solution while
preserving continuous differentiability, then one generally encounters
singularities where the classical solution becomes non-differentiable
or tends to infinity.

In practice, the viscosity solution of a Dirichlet problem for the
contact Hamilton-Jacobi equation \eqref{Feq} can be approximated numerically
by choosing a very small viscosity parameter $\fc$ and computing the
numerical solution of \eqref{viscFeq} through an appropriate version
of the method of finite elements. The computation becomes increasingly
resource intensive as the viscosity parameter decreases. In this
paper, we performed a few such computations using
\texttt{Mathematica}, without attempting to address the mathematical
questions of convergence to a viscosity solution of \eqref{Feq}, numerical
stability and so forth.

Figure \ref{fig:ViscSolutions} illustrates how numerical solutions of
a Dirichlet problem for the viscosity perturbation \eqref{viscFeq}
approach a viscosity solution of \eqref{Feq} as one takes
$\fc\rightarrow 0$. In this example, we considered the viscosity
perturbation of the contact Hamilton-Jacobi equation on a domain
around the origin in $\R^2$ where the quadratic potential
$V(x_1,x_2)=\frac{1}{18}+\frac{1}{2}(x_2^2-x_1^2)$ is positive and
imposed the Dirichlet condition
$\phi=-\log\left[\frac{\log(20)}{20}\right]$ along the circle of
radius $R=\frac{1}{20}$ centered at the origin. The figure shows
classical solutions of this boundary value problem for the viscosity
perturbation \eqref{viscFeq}, for three values of the viscosity
parameter $\fc$.

\begin{figure}[H]
\centering ~~
\includegraphics[width=.75\linewidth]{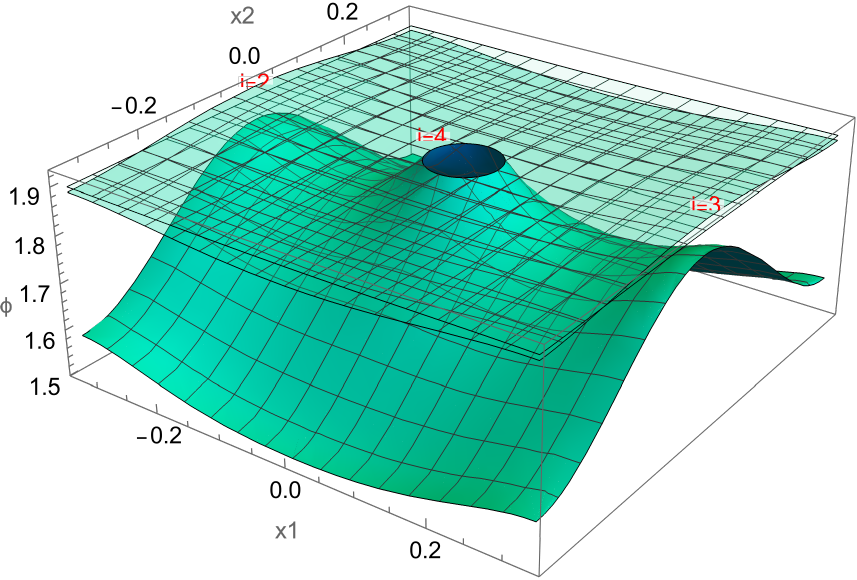}
\caption{The potential
  $V(x_1,x_2)=\frac{1}{18}+\frac{1}{2}(x_2^2-x_1^2)$ and viscosity
  approximants of the solution of the corresponding contact
  Hamilton-Jacobi equation with Dirichlet boundary condition
  $\phi=-\log[R \log(1/R)]$ imposed on the circle of radius
  $R=\frac{1}{20}$ centered at the origin. We show three viscosity
  approximants of the viscosity solution of this boundary value
  problem in the exterior of the disk which bounds this circle, for
  viscosity parameters $\fc=e^{-2i}$ with $i=2,3,4$. The solutions
  with larger viscosity coefficient are shown with higher
  transparency.}
\label{fig:ViscSolutions}
\end{figure}

\subsection{Study of the momentum curve}

For fixed $x$ and $u$, the momentum curve is a real planar cubic. To
study it, define:
\ben
P_1\eqdef A-H_0~~,~~P_2=B-{\tilde H}_0~~,
\een
which are related to $p_1$ and $p_2$ by an $x$-dependent affine
transformation. Then $-F$ can be written as:
\ben
\label{FP}
-F= P_1 P_2^2 - 3 V e^{2 u} P_1^2 +  (\Delta_0 V) P_2^2 ~~.
\een
Notice that the origin of the $(P_1,P_2)$-plane always lies on the
momentum curve. This corresponds to the point $(A,B)=(H_0,{\tilde
  H}_0)$ in the $(A,B)$-plane and to the point with coordinates:
\beqan
\label{pchar}
\!\!\!\!\!\!&&\!\!\!\!\!\!\!p_1^s(x)\eqdef -\frac{\grad V(x)\cdot (-H_0(x),{\tilde H}_0(x))}{||\dd V(x)||^2}=\frac{\pd_1V(x) H_0(x)-\pd_2V(x) {\tilde H}_0(x)}{(\pd_1V(x))^2+(\pd_2V(x))^2}~~\nn\\
\!\!\!\!\!\!&&\!\!\!\!\!\!\!p_2^s(x)\eqdef \frac{\grad_J V(x)\cdot (-H_0(x),{\tilde H}_0(x))}{||\dd V(x)||^2}=\frac{\pd_2V(x) H_0(x)+\pd_1V(x) {\tilde H}_0(x)}{(\pd_1V(x))^2+(\pd_2V(x))^2}
\eeqan
in the $(p_1,p_2)$-plane, which we call {\em the special momentum} at
$x\in U_0$. Notice that the action coordinate $u$ does not appear in
the right hand side of \eqref{pchar}. When $x$ varies in $U_0$, the
equations above define the {\em special surface} $\Sigma_F\subset
U_0\times \R^2$.

The singular points of the momentum curve are the solutions of
the equations:
\be
\pd_{p_1}F=\pd_{p_2}F=F=0~\Longleftrightarrow~ \pd_{P_1}F=\pd_{P_2}F=F=0~~.
\ee
It is easy to check that the only solution is the origin of the
$(P_1,P_2)$ plane, which therefore is the only singular point of the
momentum curve. This corresponds in the $(p_1,p_2)$ plane to the
special momentum \eqref{pchar}. When $(\Delta_0 V)(x)=0$, the momentum
curve is reducible. In this case, $F$ factorizes as:
\be
F=-P_1 (P_2^2 -3 V e^{2u} P_1)
\ee
and hence the irreducible components are the line $P_1=0$ and the
parabola with equation $P_2^2 =3 V e^{2u} P_1$.  The two components are
tangent to each other at the origin of the $(P_1,P_2)$-plane.

The equation $F=0$ reads:
\be
P_2^2(P_1+\Delta_0 V) =3 V e^{2u} P_1^2~~.
\ee
When $P_2=0$ this equation requires $P_1=0$ and hence the origin of
the $(P_1,P_2)$-plane is the only point of the momentum curve which
lies on the $P_1$ axis. When $P_1\neq 0$, the equation requires
$P_1> -\Delta_0 V$, in which case it gives:
\ben
\label{P2sol}
P_2=\pm (3 V)^{1/2} e^{u} \frac{P_1}{\sqrt{P_1+\Delta_0V}}~~.
\een
Hence the momentum curve is symmetric under reflection in the
$P_1$-axis. When $\Delta_0 V\geq 0$, the curve is connected and
contained in the half-plane defined by the inequality $P_1\geq
-\Delta_0 V$. When $\Delta_0 V=0$, it is the union of the $P_2$-axis
with a parabola which is tangent to that axis at the origin of the
$(P_1,P_2)$-plane. When $\Delta_0 V>0$, it is the union of two
embedded curves which intersect each other at the origin of that
plane. When $\Delta_0 V<0$, the curve has three connected components,
namely the origin of the $(P_1,P_2)$-plane (which is its only singular
point) and two connected components which are nonsingular and
contained in the half-space $P_1> -\Delta_0 V$.  These three cases are
illustrated in Figure \ref{fig:CubicCurve} at a point $x\in U_0$ where
$V(x) e^{2u(x)}=1$ and $(\Delta_0 V)(x)\in \{-1,0,1\}$.

\begin{figure}[H]
\centering
\begin{minipage}{.32\textwidth}
\centering  \includegraphics[width=\linewidth]{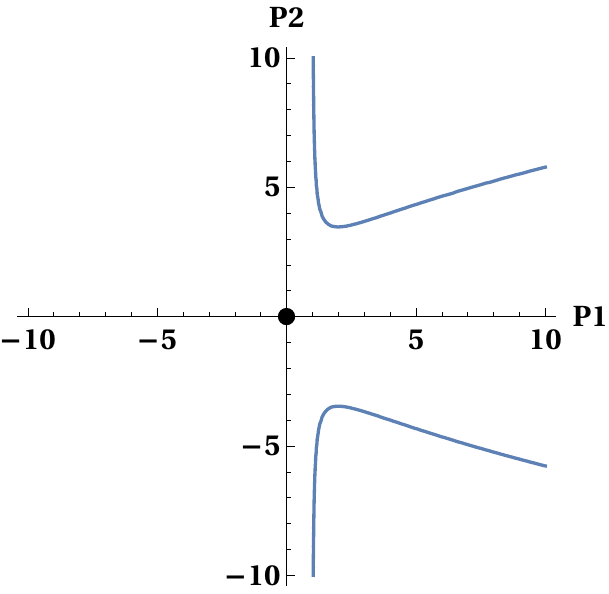}
\subcaption{$(\Delta_0 V)(x)=-1$.}
\end{minipage}
\hfill
\begin{minipage}{.32\textwidth}
\centering ~~ \includegraphics[width=\linewidth]{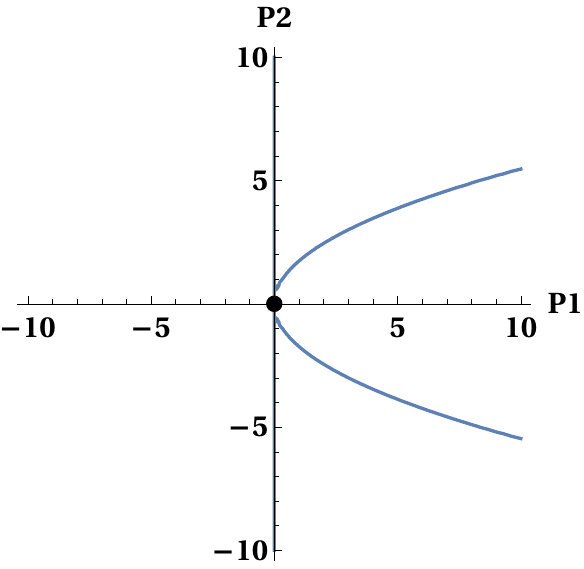}
\subcaption{$(\Delta_0 V)(x)=0$.}
\end{minipage}
\begin{minipage}{.32\textwidth}
\centering ~~ \includegraphics[width=\linewidth]{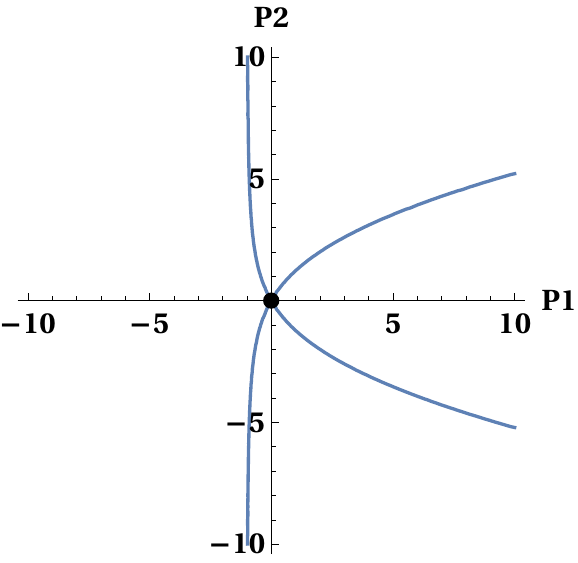}
\subcaption{$(\Delta_0 V)(x)=1$.}
\end{minipage}
\caption{The momentum curve for $V(x) e^{2u(x)}=1$ for the cases
  $(\Delta_0 V)(x)=-1,0,1$. The singular point of the curve is shown
  as a black dot.}
\label{fig:CubicCurve}
\end{figure}

\pagebreak

When $x$ varies in $U_0$, the momentum curve transitions between being
connected or disconnected when $x$ crosses the curve defined inside
$U_0$ by the equation $(\Delta_0 V)(x)=0$. This indicates a
qualitative change in the local character of the contact
Hamilton-Jacobi equation \eqref{Feq} and hence also in the qualitative
behavior of its local solutions. When $(\Delta_0V)(x)\neq 0$, equation
\eqref{P2sol} shows that the momentum curve is a double cover of an
interval on the $P_1$-axis, which is branched when $(\Delta_0V)(x)
>0$. When $\phi$ depends continuously on $x$, the equation
$F(x,\phi(x),p)=0$ determines $p$ as a {\em continuous} function of
$x$ only if $p$ is constrained to lie on one of the two branches. In
particular, this must be the case for the momentum
$p=(\pd_1\phi,\pd_2\phi)$ of a classical solution $\phi$ of the
contact Hamilton-Jacobi equation \eqref{Feq}. When $x$ varies in
$U_0$ such that it crosses the locus $\Delta_0 V=0$ (which
generically is one-dimensional), the partial derivatives of a
generalized solution $\phi$ (for example, a Lipschitz continuous
solution which satisfies the equation almost everywhere) may vary
discontinuously. In the regions where $\Delta_0 V\neq 0$ there are two
types of classical solutions $\phi$ of \eqref{Feq}, which are
distinguished by the branch of the momentum curve on which the
momentum $p=(\pd_1\phi,\pd_2\phi)$ lies.

\begin{remark}
When $P_1\neq 0$, equation \eqref{P2sol} gives:
\be
B={\tilde H}_0\pm (3 V)^{1/2} e^{u} \frac{A-H_0}{\sqrt{A-H_0+\Delta_0V}}~~.
\ee
One can use this relation to express $p_2$ in terms of $p_1$
locally and hence to write \eqref{Feq} as a time-dependent contact
Hamilton-Jacobi equation for a 1-dimensional system, where the
coordinate $x_1$ is treated as ``time''.  This can form the basis for
an alternative numerical approach to approximating viscosity solutions of
\eqref{Feq}, but we will not pursue this point of view in the present
paper.
\end{remark}

\subsection{Characteristic points}

By general theory, a {\em characteristic point} of \eqref{Feq} is a point
$(x,u,p)\in U_0\times \R^3$ such that:
\be
F(x,u,p)=F_{p_1}(x,u,p)=F_{p_2}(x,u,p)=0~~.
\ee
These conditions mean that $p=(p_1,p_2)$ coincides with the singular
point of the momentum curve $C_{x,u}$, i.e. with the special momentum
\eqref{pchar} at $x$. A point $(x,u,p)$ is characteristic iff $(x,p)$
lies on the special surface $\Sigma_F$.

\subsection{The method of characteristics}

We recall the classical method of characteristics for {\em
  locally} solving Dirichlet boundary value problems of first order
PDEs (see, for example, \cite{Melikyan}).

The Dirichlet problem for the PDE \eqref{Feq} asks for a solution
$\phi$ of this equation which satisfies the boundary condition:
\ben
\label{Dirichlet}
\phi\circ \gamma=\phi_0~~,
\een
where $\gamma:I\rightarrow U_0$ is a non-degenerate\footnote{I.e.  we
have $\gamma'(q)\neq 0$ for all $q\in I$, where $\gamma'\eqdef
\frac{\dd\gamma}{\dd q}$.} smooth curve and $\phi_0:I\rightarrow \R$
is a smooth function. Here $I$ is an interval not reduced to a
point. When this interval is compact and the boundary data
$(\gamma,\phi_0)$ is regular (see below), this problem can be solved
locally through the method of characteristics.

The characteristic system of \eqref{Feq} reads\footnote{The parameter
$t$ used here should not be confused with the cosmological time.}:
\beqan
\label{charsys}
&& \frac{\dd x^i}{\dd t}=F_{p_i}(x,u,p)\nn\\
&& \frac{\dd u}{\dd t}=p_i F_{p_i}(x,u,p)\\
&& \frac{\dd p_i}{\dd t}=-F_{x_i}(x,u,p)-p_i F_u(x,u,p)~~,\nn
\eeqan
where $F_{x_i}$, $F_u$ and $F_{p_i}$ are the partial derivatives
of $F$ with respect to $x_i$, $u$ and $p_i$. We denote the
defining vector field of this system by:
\ben
\label{xiF}
\Upsilon_F\eqdef (F_p,p F_p, -F_x-p F_u)~~,
\een
where $pF_p\eqdef p_1 F_{p_1}+p_2 F_{p_2}$. The restriction of this
vector field from $U_0\times \R\times \R^2$ to the {\em defining
  hypersurface} of the equation:
\ben
\label{W}
W_F\eqdef \{(x,u,p)\in U_0\times \R\times \R^2~~\vert~~F(x,u,p)=0\}
\een
is called the {\em characteristic vector field}. In our case, the
defining hypersurface is fibered over the set $U_0\times \R$ with fiber
at $(x,u)$ given by the momentum curve $C_{x,u}$. Notice that
$\Upsilon_F$ is tangent to $W_F$ and hence the pair $(W_F,\Upsilon_F)$
defines a four-dimensional geometric dynamical system, which we call
the {\em characteristic dynamical system of $F$}. The integral curves
of this dynamical system are the {\em characteristic curves} of the
contact Hamilton-Jacobi equation \eqref{Feq}. Any integral curve of
\eqref{charsys} which meets $W_F$ is contained in $W_F$ and hence is a
characteristic curve. To obtain a solution $\phi$ of the Dirichlet
problem \eqref{Dirichlet} for \eqref{Feq}, we must find a family of
solutions $(x(t,q),u(t,q), p(t,q))$ of \eqref{charsys} which satisfies
the initial conditions:
\ben
\label{incond}
x(0,q)=\gamma(q)~~,~~u(0,q)=\phi_0(q)~~,~~p(0,q)=p_0(q)~~(q\in I)~~,
\een
where the {\em curve of initial momenta} $p_0:I\rightarrow \R^2$ is a
smooth map which satisfies the following {\em admissibility
  conditions} for all $q\in I$:
\beqan
\label{admissibility}
&& F(\gamma(q),\phi_0(q),p_0(q))=0\nn\\
&& \gamma_1'(q)p_{01}(q)+\gamma_2'(q) p_{02}(q)=\phi_0'(q)~~.
\eeqan
The first admissibility condition means that the point $(x(0,q),
u(0,q), p(0,q))$ lies on the defining hypersurface $W_F$ for all $q\in
I$, which insures that the integral curves $t\rightarrow (x(t,q),
u(t,q),p(t,q))$ are contained in $W_F$ for all $q$. Equivalently, this
condition requires that $p_0(q)$ lies on the momentum curve
$C_{\gamma(q),\phi_0(q)}$ for all $q\in I$. The second admissibility
condition constrains the scalar product between $p_0(q)$ and the
tangent vector $\gamma'(q)$ to the curve $\gamma$ to equal the
derivative $\phi'_0(q)=\frac{\dd \phi_0(q)}{\dd q}$. Any smooth
solution $p_0$ to the admissibility conditions provides a lift:
\ben
\label{Gamma}
\Gamma=\{(\gamma(q), \phi_0(q), p_0(q))~~\vert~~q\in I\}\subset W_F
\een
to the defining hypersurface of the image $\gamma(I)\subset U_0$
of the curve $\gamma$.

The admissibility conditions can be used to locally determine $p_0$
given $\phi_0$ (though the solution might have a finite ambiguity)
provided that the Jacobian matrix of the system \eqref{admissibility}
with respect to $p_0$:
\be
\mathrm{Jac}_{p_0}(q)=\left[\begin{array}{cc} F_{p_1}(\gamma(q),\phi_0(q),p_0(q)) & F_{p_2}(\gamma(q),\phi_0(q),p_0(q))
\\ \gamma'_1(q) & \gamma'_2(q)\end{array}\right]
\ee
has maximal rank. Since the curve $\gamma$ is non-degenerate, this
amounts to the condition that $F_p$ is not proportional to $\gamma'$,
which means that the vector $F_p$ is not tangent to $\gamma$. In turn,
this is equivalent with the condition:
\ben
\label{regularity}
F_{p_1}(\gamma(q),\phi_0(q),p_0(q))\gamma_2'(q)-F_{p_2}(\gamma(q),\phi_0(q),p_0(q)) 
\gamma_1'(q)\neq 0\quad\forall q\in I~~.
\een
The triplet $(\gamma, \phi_0, p_0)$ is called {\em regular} (or {\em
  noncharacteristic}) if this condition is satisfied. The regularity
condition means that the projected characteristic curves $t\rightarrow
x(t)$ are nowhere tangent to the curve $\gamma$. In particular,
regularity requires that $F_{p_1}(\gamma(q),\phi_0(q),p_0(q))$ and
$F_{p_2}(\gamma(q),\phi_0(q),p_0(q))$ should not vanish simultaneously
for any $q\in I$. This means that $(\gamma(q),\phi_0(q),p_0(q))$ must
be a noncharacteristic point of \eqref{Feq} for all $q\in I$ (recall
that $(\gamma(q),\phi_0(q),p_0(q))\in W_F$ by the first condition in
\eqref{admissibility}). Since $p_0(q)$ lies on the momentum curve
$C_{\gamma(q),\phi_0(q)}$ by the first condition in
\eqref{admissibility}, it follows that a necessary (but not
sufficient) condition for regularity is that $p_0(q)$ should not
coincide with the singular point of this curve for any $q\in I$.  In
our case, this means that $p_0(q)$ should not coincide with the
special momentum $p^s(\gamma(q))$ computed at the point $\gamma(q)\in
U_0$. The boundary data $(\gamma,\phi_0)$ is called {\em regular} if
there exists a smooth function $p_0:I\rightarrow \R^2$ which satisfies
both \eqref{admissibility} and \eqref{regularity}. In this case, one
also says that the Dirichlet data is {\em noncharacteristic}.

Assume that the interval $I$ is compact. For regular boundary data and
an admissible choice of $p_0$, condition \eqref{regularity} insures
that there exists a small enough positive number $a$ such that the map
$(-a,a)\times I \ni (t,q)\rightarrow (x_1(t,q), x_2(t,q))\in U_0$
provided by a solution of \eqref{charsys} which satisfies the initial
conditions \eqref{incond} is a diffeomorphism on the interior $\cV$ of
its image. One can invert this map in order to extract $t$ and $q$ as
functions of $(x_1,x_2)\in \cV$. Then the function
$\phi:\cV\rightarrow \R$ defined through:
\ben
\label{phichar}
\phi(x_1,x_2)\eqdef u(t(x_1,x_2),q(x_1,x_2))
\een
is a (smooth) solution of \eqref{Feq} defined on $\cV\subset U_0$,
which satisfies the Dirichlet boundary condition:
\be
\phi(\gamma(q))=\phi_0(q)\quad\forall q\in I~~.
\ee
Notice that the image of $\gamma$ is contained in $\cV$. Also notice
that $\phi$ can be finitely non-unique, since there can exist a finite
ambiguity in the choice of the admissible initial momentum
$p_0:I\rightarrow \R^2$. One can relax the smoothness assumption on
the choice of admissible momentum $p_0$, but in this case the solution
\eqref{phichar} may fail to be smooth. This situation arises in some of
the examples considered later in the present paper.

\begin{remark}
Suppose the boundary data $(\gamma,\phi_0)$ is nonregular, i.e. there
exists no admissible $p_0$ which satisfies the regularity condition
\eqref{regularity}. The simplest such situation is that of {\em
  completely irregular Dirichlet problems}, when there exists a choice
of initial momenta $p_0:I\rightarrow \R^2$ such that
$F_p(\gamma(q),\phi_0(q),p_0(q))$ is tangent to $\gamma$ for all $q\in
I$. It can be shown (see, for example, \cite{Melikyan})
that such a Dirichlet problem admits a locally-defined classical solution
iff the lifted curve \eqref{Gamma} is a {\em characteristic curve},
which means that the characteristic vector field $\Upsilon_F$ is
tangent to $\Gamma$ at each point of $\Gamma$. In this case, there
exists an infinity of locally-defined smooth solutions to the
Dirichlet problem. In this paper, we consider only Dirichlet problems
which are regular, though we allow $p_0$ to be a piece-wise smooth
function of $q$ in some cases.
\end{remark}

\subsection{Constant Dirichlet conditions on a circle}
\label{subsec:DirichletCircle}

The following Dirichlet condition will be used later on.  Consider the
case when $\gamma$ is a circle of radius $R<1$ (with respect to the
local Euclidean metric $\cG_0$) which is contained in $U_0$ and
centered at some point $x_0\in U_0$. We can take $I=[0,2\pi]$ and
$q=\theta$, where $\theta$ is the polar angle in the isothermal
coordinates $(x^1,x^2)$. Then $\gamma$ has equations:
\be
x_1=R\cos \theta~~,~~x_2=R\sin \theta~~.
\ee
Moreover, we take the Dirichlet datum $\phi_0$ to be constant along $\gamma$,
i.e. $\phi_0'(\theta)=0$. In this case, the second admissibility
condition in \eqref{admissibility} requires that $p_0$ be normal to
$\gamma$, i.e. we have:
\be
p_0(\theta)=(\nu(\theta)\cos\theta,\nu(\theta)\sin\theta)\quad\forall \theta\in [0,2\pi]~~,
\ee
where $\nu\eqdef||p_0||_0$. For every $\theta\in [0,2\pi]$, the line
of angle $\theta$ in the $p$-plane intersects the momentum curve
$C_{\gamma(\theta),\phi_0}$ in at most three points. For the boundary
data $(\gamma,\phi_0)$ to be admissible, this intersection must be
non-empty for all $\theta\in [0,2\pi]$, in which case any point of
intersection determines a possible choice of $p_0(\theta)$. As
$\theta$ varies in the interval $[0,2\pi]$, this gives up to three
continuously-varying points on $C_{\gamma(\theta),\phi_0}$. A
necessary condition for the boundary data to be regular is that the
momentum $p_0(\theta)$ never becomes special, i.e. that at least one
of these choices never meets the singular point of
$C_{\gamma(\theta),\phi_0}$ as $\theta$ varies from $0$ to
$2\pi$. 

When the point $x$ varies on the circle $\gamma$, the special
momenta \eqref{pchar} describe a curve in the $(p_1,p_2)$ plane, which
we call the {\em curve of special initial momenta}. An admissible
initial momentum cannot lie on this curve.

Let us illustrate this in the case
$V=10^{-10}+\frac{1}{2}\left(\frac{1}{5}x_1^2+x_2^2\right)$, for the
constant boundary condition:
\ben
\label{BC}
\phi\vert_{C_R}=\phi_0=-\log[R\log(1/R)]
\een
imposed on the circle $C_R$ of radius $R=\frac{1}{20}$ centered at the
origin of the $(x_1,x_2)$ plane. This condition is natural given the
asymptotic forms discussed in Section \ref{sec:quasilinear}. 

Figure \ref{fig:InitialMomenta} shows the curve of special initial momenta
and the norms $\nu$ of admissible initial momenta (as functions of the
polar angle $\theta$) for this choice of $R$ and $V$. The latter are
obtained by solving the following cubic equation for $\nu$:
\ben
\label{nueq}
F(R\cos\theta,R\sin\theta, -\log[R\log(1/R)], \nu\cos\theta,\nu\sin\theta)=0
\een
whose real solutions $\nu(\theta)$ we plot as functions of
$\theta$. The plot has limited range, hence some real solutions are
not visible in the figure.

\begin{figure}[H]
\centering
\vspace{-1em}
\begin{minipage}{0.46\textwidth}
\centering ~~ \includegraphics[width=.5\linewidth]{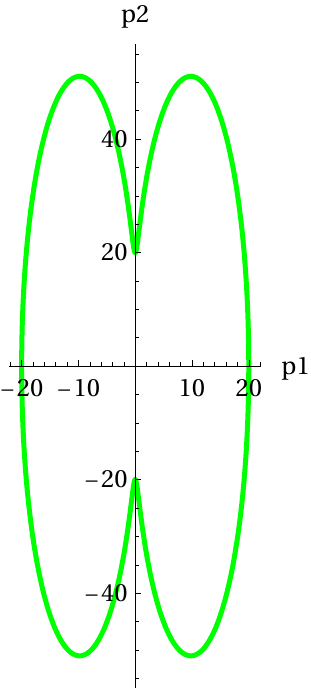}
\subcaption{The curve of special initial momenta.}
\end{minipage}
\hfill
\centering 
\begin{minipage}{0.5\textwidth}
\centering ~~\includegraphics[width=.9\linewidth]{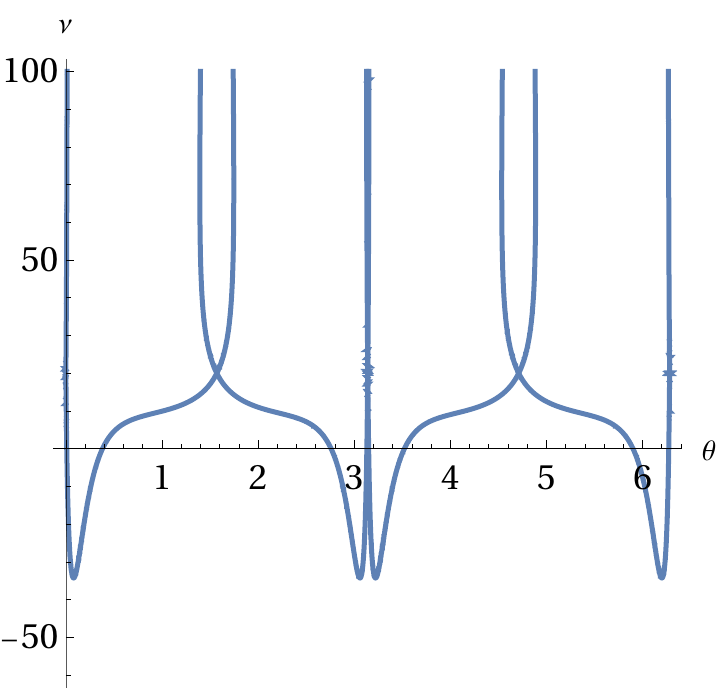}
\vspace{2em}
\subcaption{Norms of admissible initial momenta.}
\end{minipage}
\hfill
\caption{The curve of special initial momenta and the norms of
  admissible momenta for
  $V=10^{-10}+\frac{1}{2}\left(\frac{1}{5}x_1^2+x_2^2\right)$ and the
  Dirichlet boundary condition \eqref{BC} with $R\!=\!\frac{1}{20}$.}
\label{fig:InitialMomenta}
\end{figure}

Figure \ref{fig:Characteristics} shows the
choice of a set of regular initial momenta together with the
projections of the corresponding characteristic curves on
$(x_1,x_2,u)$-space.

\begin{figure}[H]
\centering
\begin{minipage}{0.45\textwidth}
\centering ~~ \includegraphics[width=.8\linewidth]{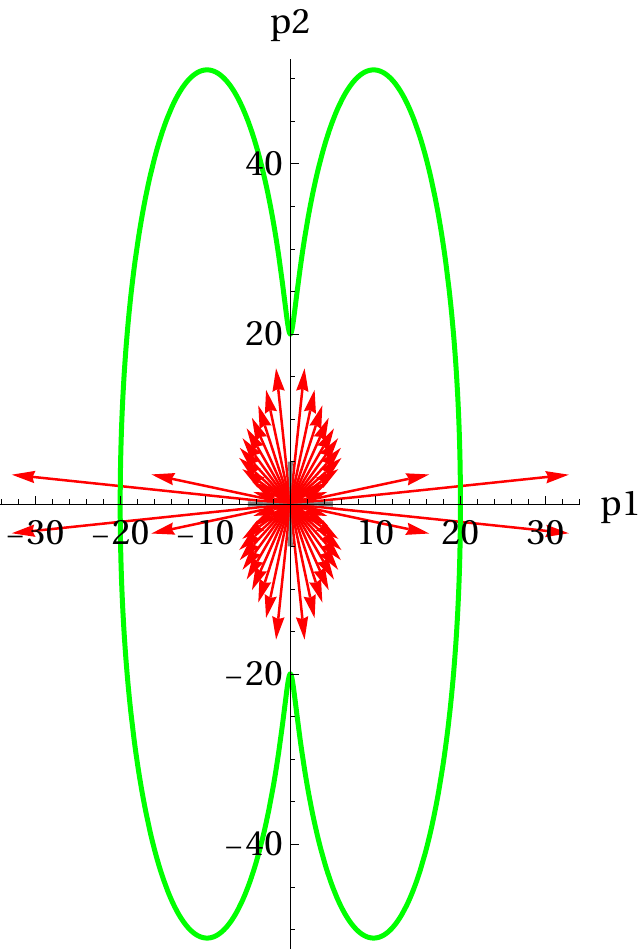}
\subcaption{A set of regular initial momenta \\(shown as red arrows).}
\end{minipage}
\hfill
\centering 
\begin{minipage}{0.49\textwidth}
\centering ~~\includegraphics[width=\linewidth]{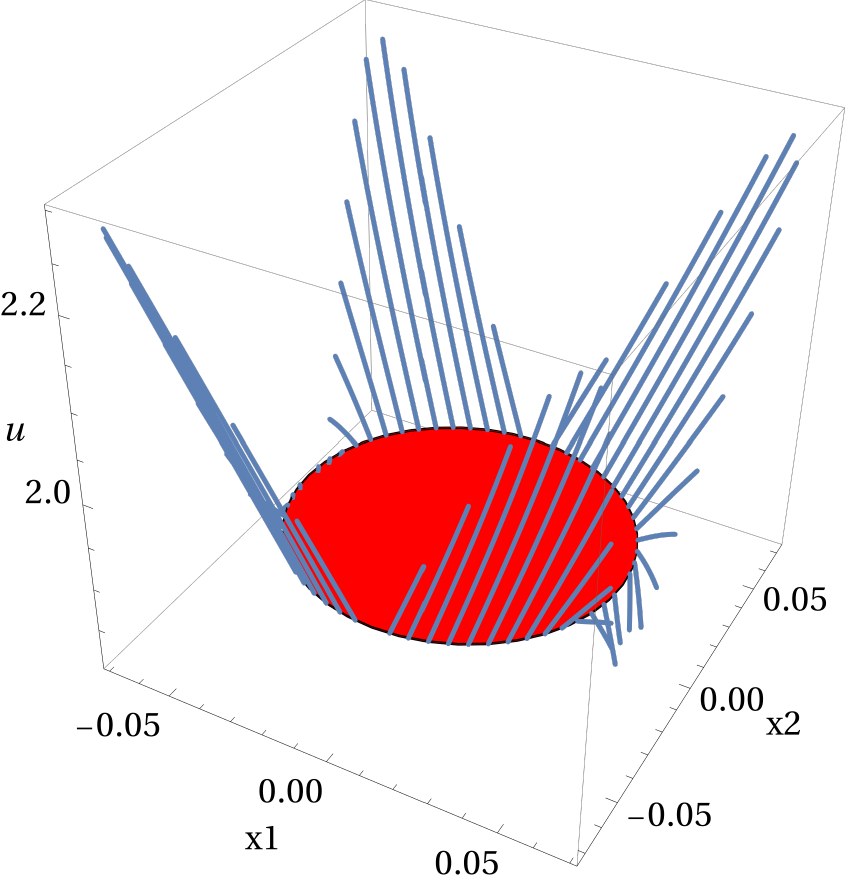}
\subcaption{Projection of characteristic curves on
  $(x_1,x_2,u)$-space for the choice of initial momenta shown on the
  left.}
\end{minipage}
\hfill
\caption{Choices of regular initial momenta and projection of the
  corresponding characteristic curves on $(x_1,x_2,u)$-space for
  $V=10^{-10}+\frac{1}{2}\left(\frac{1}{5}x_1^2\!+\!x_2^2\right)$ and
  the Dirichlet boundary condition \eqref{BC} with
  $R\!=\!\frac{1}{20}$.}
\label{fig:Characteristics}
\end{figure}

\subsection{The stationary points of the characteristic dynamical system}

The characteristic curves of an admissible Dirichlet problem are
integral curves of the dynamical system $(W_F,\Upsilon_F)$. The
stationary points of this system correspond to the zeroes of
$\Upsilon_F$ on $W_F$, i.e. to the points $(x,u,p)\in U\times \R^3$
which satisfy the conditions:
\ben
\label{singpoint}
F(x,u,p)=F_{p_i}(x,u,p)=F_{x_i}(x,u,p)+p_i F_u(x,u,p)=0 \quad \forall i=1,2~~.
\een
It is clear that any such point is a characteristic point. The following
result shows that the converse is also true. 

\begin{prop}
The stationary points of the characteristic dynamical system
$(W_F,\Upsilon_F)$ coincide with the characteristic points of $F$.
\end{prop}

\begin{proof}
It suffices to show that the last equalities in \eqref{singpoint} are
satisfied at any characteristic point. If $(x,u,p)$ is such a point,
then we have $F(x,u,p)=F_{p_i}(x,u,p)=0$ and hence $p_1$ and $p_2$ are
given by \eqref{pchar}, which correspond to $A=H_0$ and $B={\tilde
  H_0}$. Thus:
\be
F_u=6V e^{2u} (A-H_0)^2
\ee
vanishes and the last conditions in \eqref{singpoint} reduce to:
\beqa
F_{x_1}(x,u,p)=F_{x_2}(x,u,p)=0~~,
\eeqa
where $p_1$ and $p_2$ are given by \eqref{pchar}. Using expression
\eqref{FP}, we compute:
\be
F_{x_i}=\frac{\pd F}{\pd P_k} \pd_{x_i} P_k+\frac{\pd F}{\pd V} \pd_{x_i} V+\frac{\pd F}{\pd (\Delta_0 V)} \pd_{x_i} (\Delta_0 V)~~,
\ee
where $u$ and $p$ are held fixed. 

Since \eqref{pchar} corresponds to
$P=(0,0)$ (which is the singular point of the momentum curve), we have
$\frac{\pd F}{\pd P_k}\Big{|}_{P=0}=0$. On the other hand, the quantities:
\be
\frac{\pd F}{\pd V}=-3 e^{2u} P_1^2~\quad~,~\quad~\frac{\pd F}{\pd (\Delta_0 V)}=-P_2^2~~
\ee
vanish when $P_1=P_2=0$. It follows that $F_{x_i}(x,u,p)$ vanish and
hence any characteristic point is a stationary point of the
characteristic dynamical system.
\end{proof}

\subsection{Numerical examples}

Let us illustrate the solutions of the contact Hamilton-Jacobi
equation for the conformal factor and their characteristics. Recall
that $U_0$ is an open subset of $\R^2$, whose Cartesian coordinates
$x^1,x^2$ are isothermal on $U$. Thus $J$ is the canonical complex
structure which identifies $\R^2$ with the complex plane $\C$ of
complex coordinate $z=x^1+\i x^2$. For simplicity, we take:
\ben
\label{Vquad}
V(x_1,x_2)=V_c+\frac{1}{2}(\lambda_1 x_1^2+\lambda_2 x_2^2)~~,
\een
where $V_c>0$ and $\lambda_1,\lambda_2$ are nonzero real numbers. We
assume that $U_0$ is a domain contained in the region of $\R^2$ where
$V>0$. Such potentials arise as local approximations of an everywhere
positive potential of a cosmological model close to a non-degenerate
critical point. The quadratic potential \eqref{Vquad} has a single
critical point located at the origin of the $(x_1,x_2)$-plane, which
is an extremum if $\lambda_1\lambda_2>0$ and a saddle point if
$\lambda_1\lambda_2<0$.

We consider the Dirichlet problem of Subsection
\ref{subsec:DirichletCircle} with boundary condition \eqref{BC} on a
circle of radius $R<1$ centered at the origin. This problem is regular
except possibly for a finite set of points of this circle when
$\lambda_1\neq \lambda_2$ (we say that the problem is {\em essentially
  regular} in this case). If $\lambda_1=\lambda_2$, then the only
admissible choice of $p_0(\theta)$ is special for all $\theta$ and the
Dirichlet problem is irregular.

We show some numerically computed projected characteristic curves in a
few essentially regular cases in the bottom left corner of Figures
\ref{fig:Nonlinear1}, \ref{fig:Nonlinear2}, \ref{fig:Nonlinear4}.  For
the indicated values of $V_c$, $\lambda_1$, $\lambda_2$ and $R$, these
figures also show the contour and 3d plots of the corresponding
potential and the numerical solutions of the Dirichlet problem for the
viscosity perturbation \eqref{viscFeq} of the contact Hamilton-Jacobi
equation \eqref{Feq} for the listed values of the viscosity parameter
$\fc$. 

Figure \ref{fig:Nonlinear3} shows the irregular case
$\lambda_1=\lambda_2=1$, when the classical method of characteristics
does not apply. The qualitative behavior of the viscosity approximants
of the solutions differs little between these four cases on the region
within which they were computed. The values considered for the
viscosity parameter $\fc$ were dictated by the computational limits of
our system. Computing with smaller values of $\fc$ might reveal more
detail.

\begin{figure}[H]
\centering
\begin{minipage}{.49\textwidth}
\centering \!\!\!\!\includegraphics[width=0.85\linewidth]{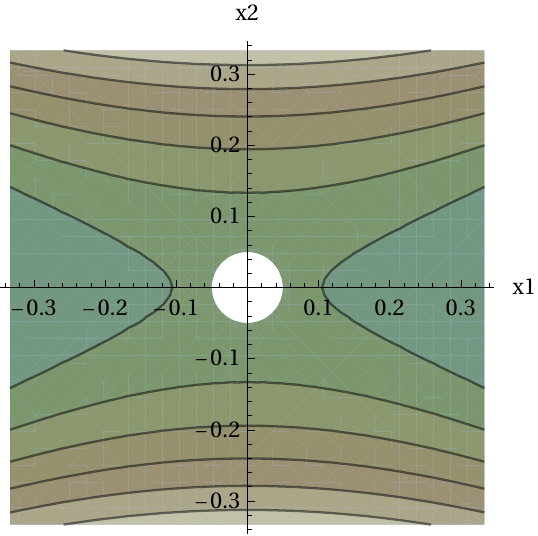}
\subcaption{Contour plot of the potential.}
\end{minipage}
\hfill
\begin{minipage}{.49\textwidth}
\centering \includegraphics[width=1.07\linewidth]{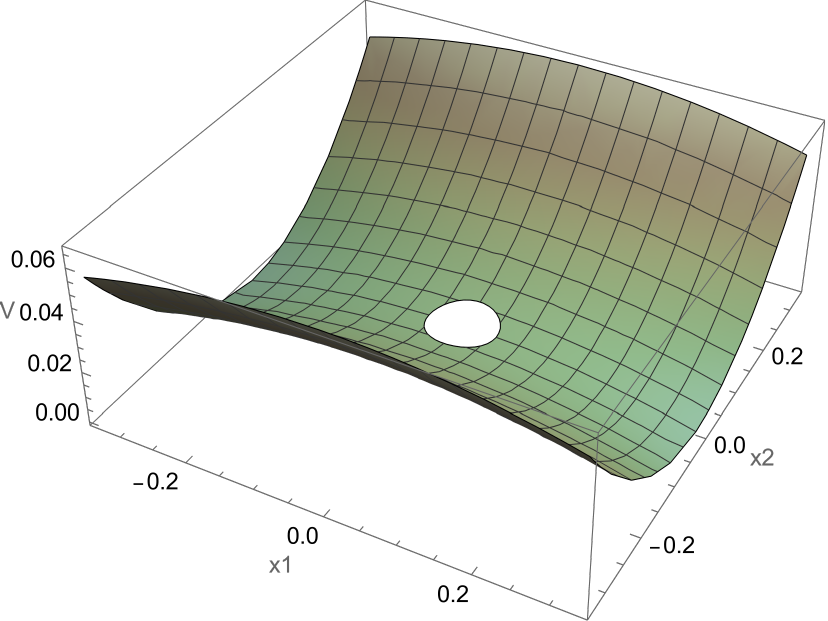}
\subcaption{3D plot of the potential.}
\end{minipage}
\hfill
\vspace{1em}
\centering
\begin{minipage}{.46\textwidth}
\centering ~~ \includegraphics[width=.3\linewidth]{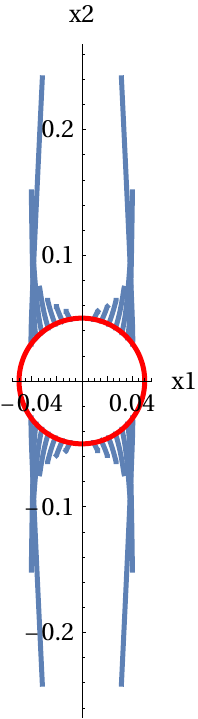}
\subcaption{Some characteristic curves projected \\on the
  $(x_1,x_2)$-plane. Some curves intersect if followed long enough,
  which signals the breakdown of the method of characteristics.}
\end{minipage}
\hfill
\begin{minipage}{.49\textwidth}
\centering \!\!\!\!\!\!\includegraphics[width=1.1\linewidth]{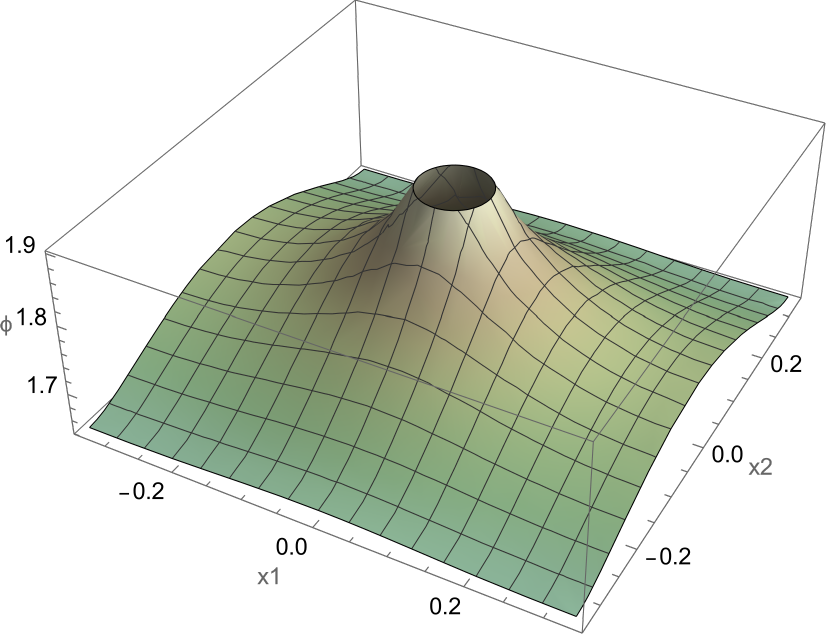}
\subcaption{Solutions of the Dirichlet problem for the viscosity
  perturbation with $\fc=e^{-7}$.}
\end{minipage}
\hfill
\caption{The potential, projected characteristics and a viscosity
  approximant of the solution of the Dirichlet problem for the contact
  Hamilton-Jacobi equation for $V_c=1/90$ and $\lambda_1=-1/5$,
  $\lambda_2=1$ with $R=1/20$. We show the approximant of the solution
  $\phi$ in a region where the potential is positive.}
\label{fig:Nonlinear1}
\end{figure}

\begin{figure}[H]
\vspace{3em}
\centering
\!\!\begin{minipage}{.49\textwidth}
\centering  \includegraphics[width=\linewidth]{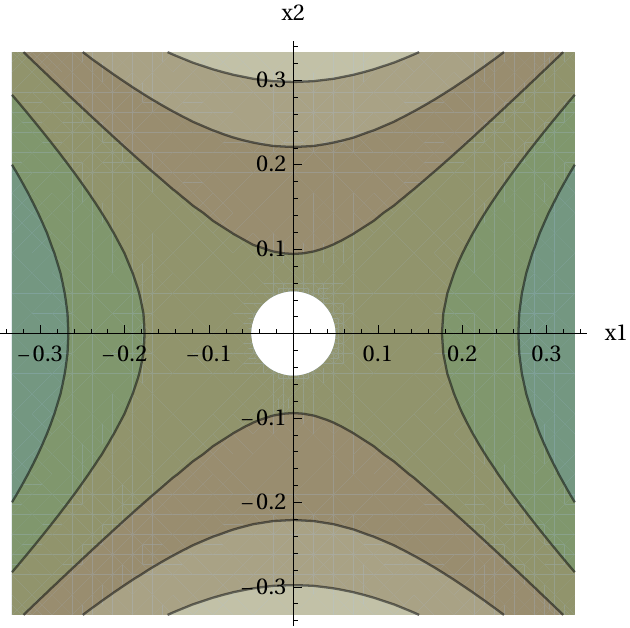}
\subcaption{Contour plot of the potential.}
\end{minipage}
\hfill
\begin{minipage}{.49\textwidth}
\vspace{1em}
\centering ~~~~\includegraphics[width=\linewidth]{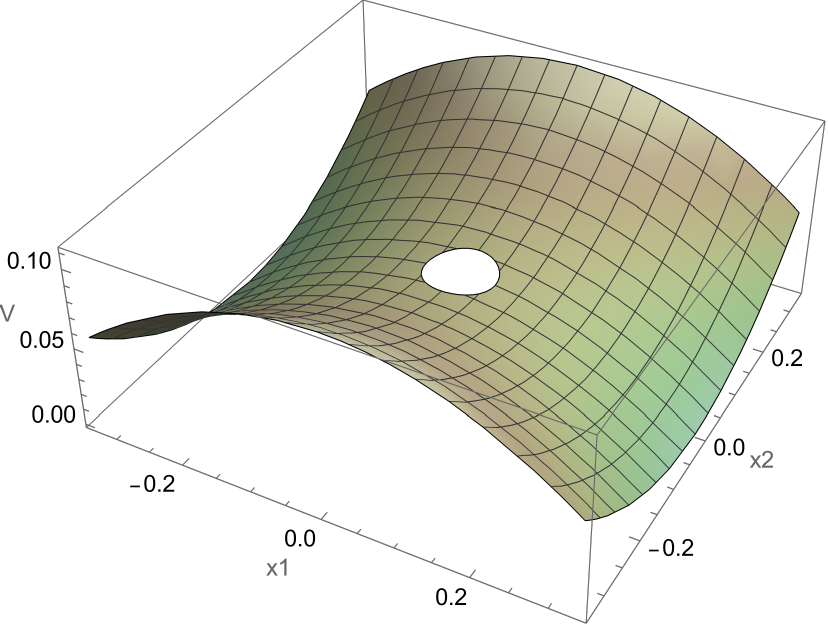}
\vspace{2em}
\subcaption{3D plot of the potential.}
\end{minipage}
\hfill
\vspace{2em}
\centering
\begin{minipage}{.49\textwidth}
\centering \!\!\!\!\!\!\!\!\!\! \includegraphics[width=.9\linewidth]{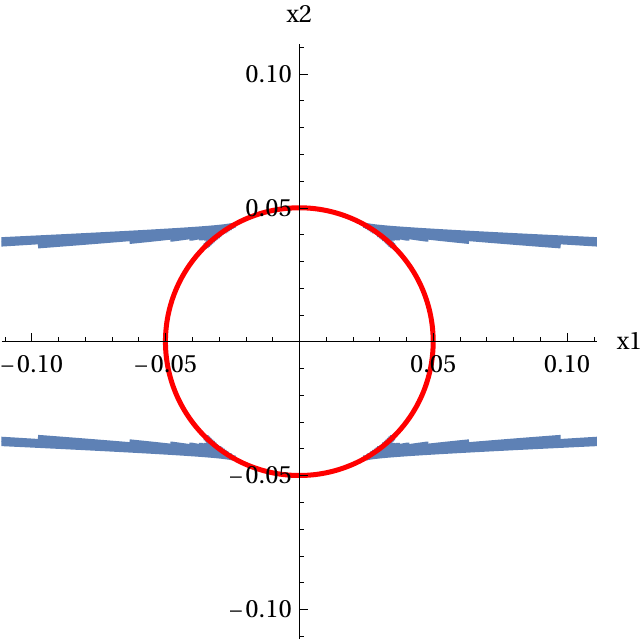}
\subcaption{Some characteristic curves projected \\on the $(x_1,x_2)$-plane.}
\end{minipage}
\hfill
\begin{minipage}{.49\textwidth}
\centering \includegraphics[width=\linewidth]{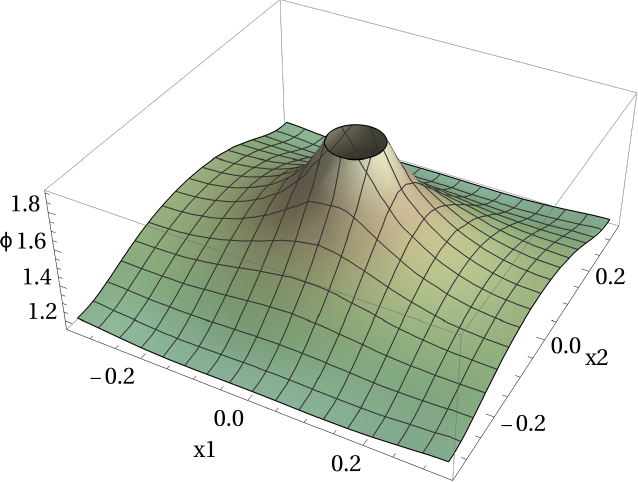}
\vspace{1em}
\subcaption{Solution of the Dirichlet problem for the viscosity
  perturbation with $\fc=e^{-8}$.}
\end{minipage}
\hfill
\vspace{2em}
\caption{The potential, projected characteristics and a viscosity
  approximant of the solution of the Dirichlet problem for the contact
  Hamilton-Jacobi equation for $V_c=1/18$ and $\lambda_1=-1$,
  $\lambda_2=1$ with $R=1/20$. We show the approximant of the solution
  $\phi$ in a region where the potential is positive.}
\label{fig:Nonlinear2}
\end{figure}

\vspace{4em}

\pagebreak

\begin{figure}[H]
\vspace{3em}
\centering
\begin{minipage}{.49\textwidth}
\centering \!\!\includegraphics[width=\linewidth]{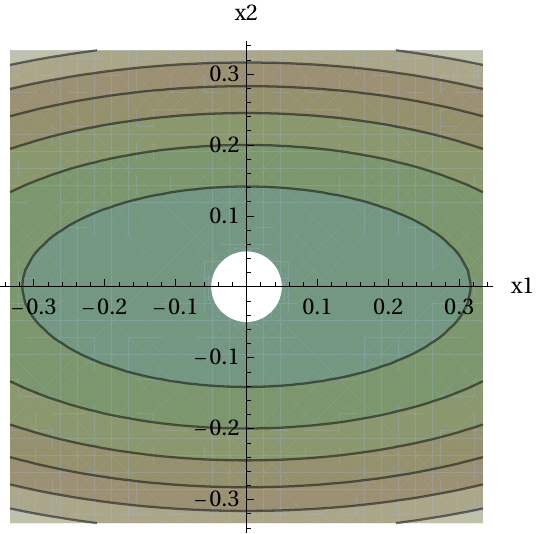}
\vspace{1em}
\subcaption{Contour plot of the potential.}
\end{minipage}
\hfill
\begin{minipage}{.49\textwidth}
\centering \!\!\!\!\!\!\!\includegraphics[width=1.05\linewidth]{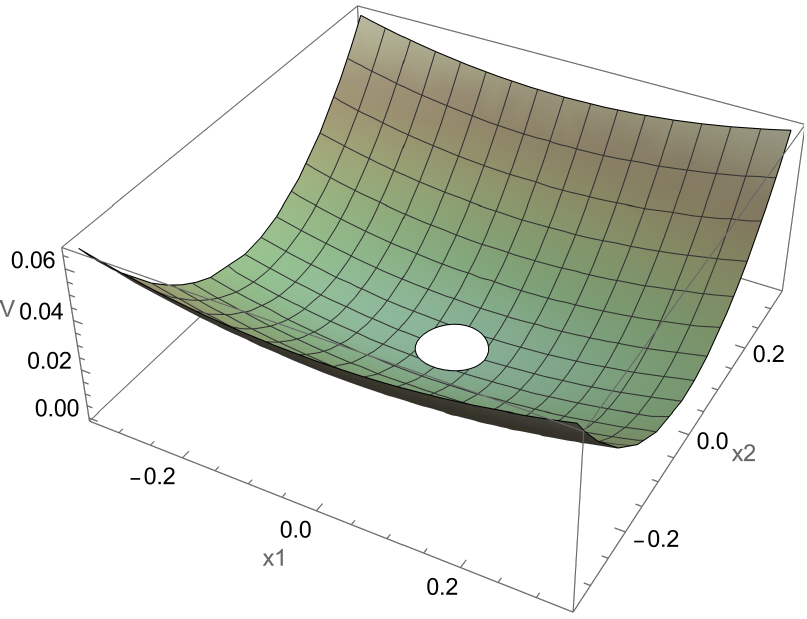}
\vspace{2em}
\subcaption{3D plot of the potential.}
\end{minipage}
\hfill
\\
\vspace{2em}
\centering
\begin{minipage}{.49\textwidth}
\centering \!\!\!\!\!\!\!\!\!\!\! \includegraphics[width=.9\linewidth]{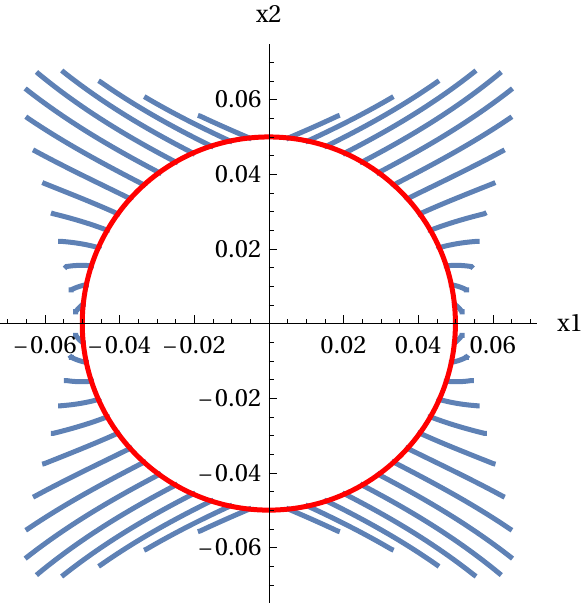}
\vspace{1em}
\subcaption{Some characteristic curves projected \\on the $(x_1,x_2)$-plane.}
\end{minipage}
\hfill
\begin{minipage}{.49\textwidth}
\centering  \includegraphics[width=1.05\linewidth]{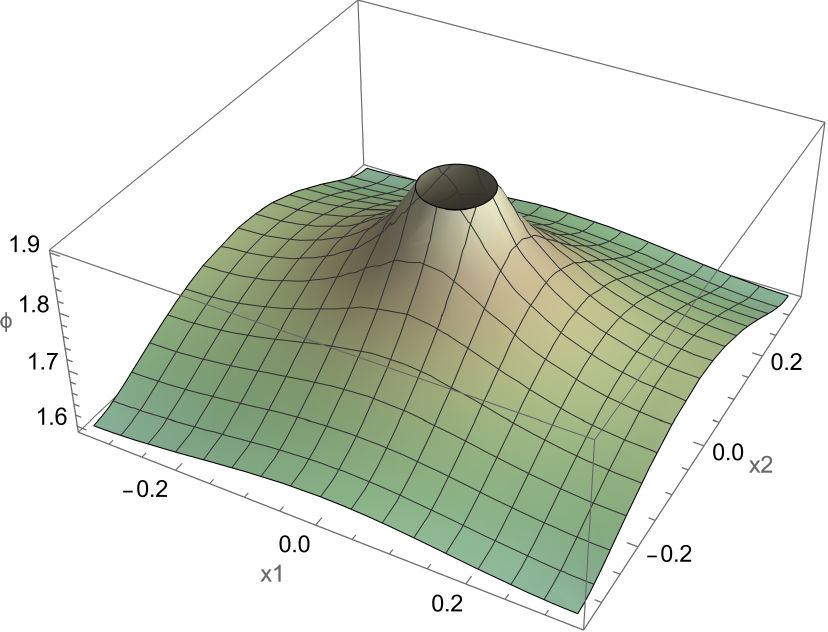}
\vspace{1em}
\subcaption{Solution of the Dirichlet problem for the viscosity perturbation with $\fc=e^{-7}$.}
\end{minipage}
\hfill
\vspace{1em}
\caption{The potential, projected characteristics and a viscosity
  approximant of the solution of the Dirichlet problem for the contact
  Hamilton-Jacobi equation for $V_c\!=\!10^{-10}$ and
  $\lambda_1\!=\!1/5$, $\lambda_2\!=\!1$ with $R\!=\!1/20$.}
\label{fig:Nonlinear4}
\end{figure}

\begin{figure}[H]
\centering
\begin{minipage}{.47\textwidth}
\centering \!\! \includegraphics[width=\linewidth]{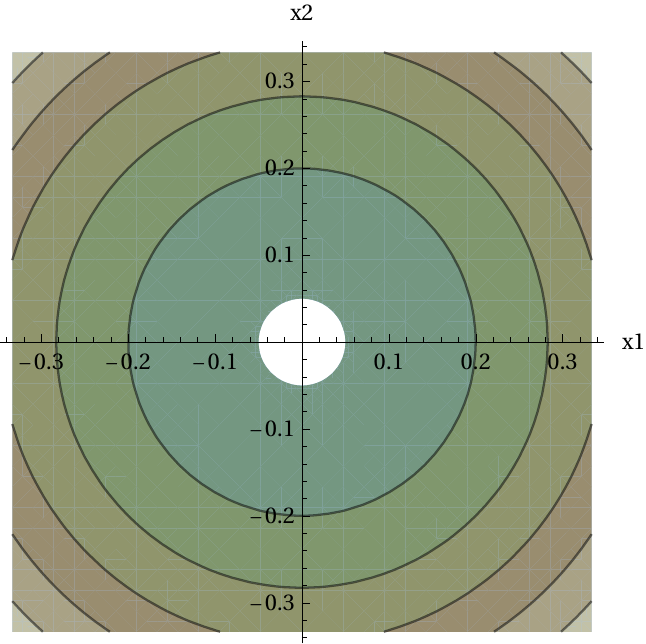}
\subcaption{Contour plot of the potential.}
\end{minipage}
\hfill
\begin{minipage}{.5\textwidth}
\centering \includegraphics[width=\linewidth]{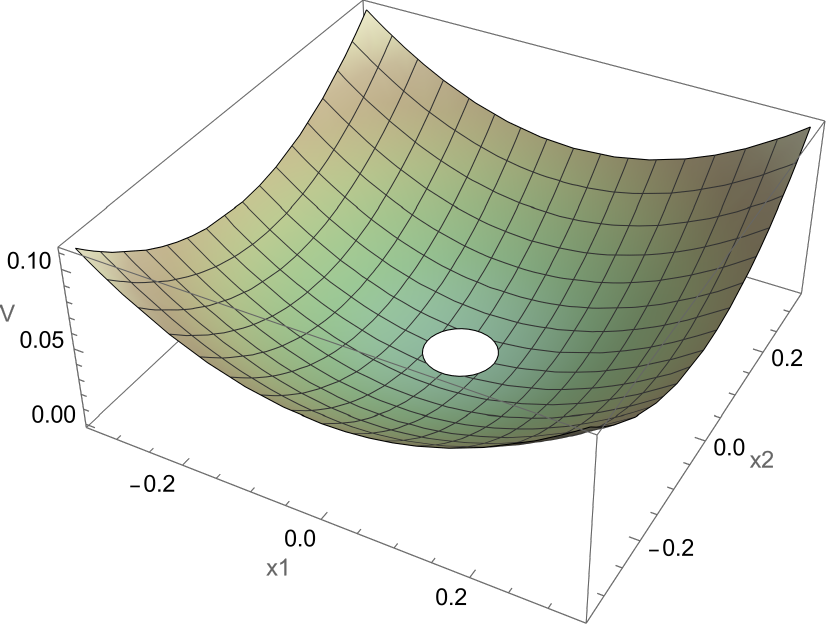}
\subcaption{3D plot of the potential.}
\end{minipage}
\hfill
\\
\centering
\begin{minipage}{.47\textwidth}
\centering \!\!\!\!\!\!\!\!\!\! \includegraphics[width=.85\linewidth]{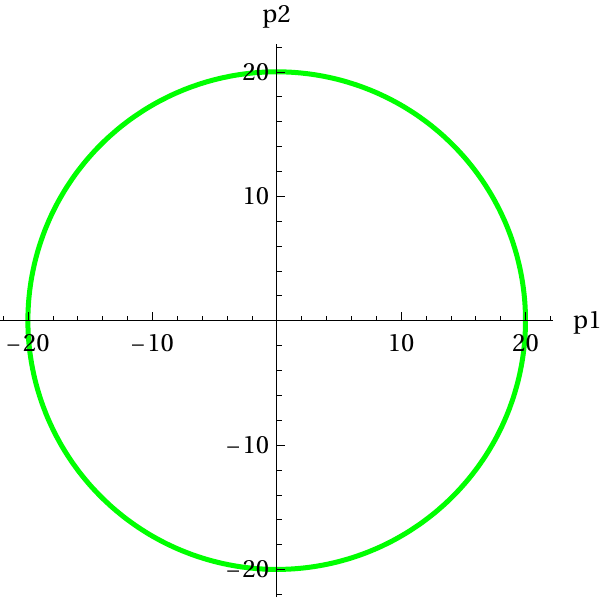}
\subcaption{The curve of special initial momenta.}
\end{minipage}
\hfill
\begin{minipage}{.5\textwidth}
\centering \includegraphics[width=1.07\linewidth]{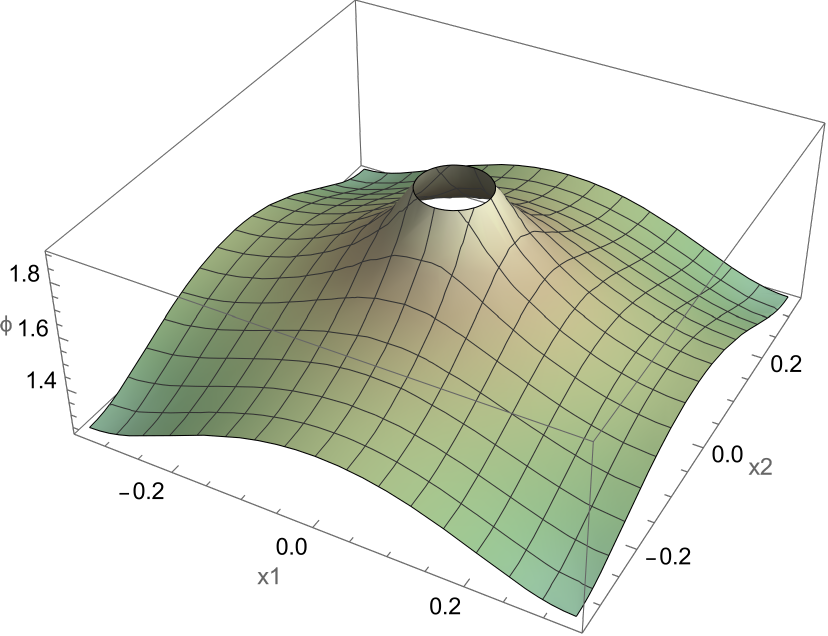}
\subcaption{Solution of the Dirichlet problem for the viscosity perturbation with $\fc=e^{-6}$.}
\end{minipage}
\hfill
\caption{The potential and a viscosity approximant of the solution of
  the Dirichlet problem for the contact Hamilton-Jacobi equation for
  $V_c=10^{-10}$ and $\lambda_1=\lambda_2=1$ with $R=1/20$. In this case,
  all admissible initial momenta lie on the curve of special initial
  momenta, which we show in the lower left corner. Thus the Dirichlet problem is
  irregular and the classical method of characteristics does not apply.}
\label{fig:Nonlinear3}
\end{figure}

\section{The quasilinear approximation near a non-degenerate critical point of the scalar 
potential}
\label{sec:quasilinear}

The critical points of the scalar potential $V$ play a crucial role in
determining important features of cosmological dynamics. In this
section, we study the behavior of the contact Hamilton-Jacobi equation
\eqref{Feq} and its solutions in the vicinity of {\em non-degenerate} critical
points of $V$. We show that this nonlinear equation can be
approximated by a quasilinear ODE near such a point and study the
solutions of the latter, which provide asymptotic approximants for certain
solutions of the contact Hamilton-Jacobi equation.

Let $c\in \Crit V$ be a non-degenerate (hence isolated) critical point
of $V$ and fix the complex structure $J$ on $\cM$. Let $(U,z)$ be a
complex coordinate chart for $(\cM,J)$ centered at $c$ (i.e. $z(c)=0$)
such that the punctured neighborhood $\dot{U}\eqdef U\setminus \{c\}$
is contained in $\cM_0$ and let $x^1=\Re z$ and $x^2=\Im z$ be the
corresponding isothermal coordinates defined on $U$.

Since $c$ is a non-degenerate critical point, the Riemannian Hessian
of $V$ at $c$ is a non-degenerate bilinear symmetric form defined on
$T_c\cM$ which is independent of the choice of the scalar field metric
$\cG$. In particular, we have:
\be
\Hess(V)(c)=\Hess_0(V)(c)=\frac{1}{2}(\pd_i\pd_j V)(c)\pd_i\otimes \pd_j\vert_{c}\in \Sym^2(T^\ast_c\cM)~~.
\ee
This covariant symmetric tensor can be diagonalized by a transformation:
\be
x'=Rx~,~\mathrm{with}~~R\in \SO(2,\R)\simeq \U(1)~~,
\ee
where $(U,x')$ is once again an isothermal coordinate system of the
same form as $(U,x)$, namely the isothermal coordinate system defined
by the complex coordinate $z'=e^{\i \theta} z$, where $\theta$ is the
rotation angle of the matrix $R$. Thus, we can assume without loss of
generality that the isothermal coordinate system $(U,x)$ was chosen
such that:
\be
\pd_i\pd_jV=\lambda_1 \delta_{i1}\delta_{j1}+\lambda_2\delta_{i2}\delta_{j2}~~,
\ee
where $\lambda_1,\lambda_2\!\in \!\R$ are the real eigenvalues of the Euclidean
Hessian operator $\widehat{\Hess}_0(c)\!\in \!\End(T_c\cM)$ at $c$ (recall that this
operator is $\cG_0(c)$-symmetric and hence has real eigenvalues). Non-degeneracy of the 
Hessian at $c$ implies that $\lambda_1$ and $\lambda_2$ are both nonzero. With this choice 
of isothermal coordinates, the Taylor expansion of $V$ at $c$ takes the form:
\be
V(x)=V(c)+\frac{1}{2}(\pd_i\pd_jV)(c)x^i x^j +\cO(||x||_0^3)=V_c+\frac{1}{2}(\lambda_1 x_1^2+\lambda_2 x_2^2) +\cO(||x||_0^3)~~,
\ee
which gives (without summation over $i$ !):
\ben
\label{gradVlin}
(\pd_iV)(x)=\lambda_i x^i+\O(||x||_0^2)
\een
and:
\be
||\dd V||_0=\sqrt{\lambda_1^2 x_1^2+\lambda_2^2 x_2^2}+\cO(||x||_0^2)~~.
\ee
Using these expressions, we find:
\ben
\label{Hleading}
H_0=\frac{\lambda_i^3 x_i^2}{\lambda_1^2 x_1^2+\lambda_2^2 x_2^2}+\cO(||x||_0^2)~~\mathrm{and}
~~{\tilde H}_0=\lambda_1\lambda_2(\lambda_2-\lambda_1) \frac{x_1x_2}{\lambda_1^2 x_1^2+\lambda_2^2 x_2^2}+\cO(||x||_0^2)~~.
\een
Consider the following homogeneous polynomial functions of degree two
in the variables $x_1$ and $x_2$, where $k\in \Z_{>0}$:
\ben
\label{sdef}
s_k(x)\eqdef \lambda_1^k x_1^2+\lambda_2^k x_2^2~~.
\een
Notice that these functions are invariant under the permutation which
exchanges $x_1$ with $x_2$ and $\lambda_1$ with $\lambda_2$.

\begin{prop}
We have:
\be
F(x,u,p)=-\frac{a_1(x,u) x^1 p_1+a_2(x,u) x^2 p_2-b(x,u)}{s_2(x)^3}+\cO(||x||_0^2)~~,
\ee
where $a_i$ and $b$ are homogeneous polynomial functions of degree
six in $x_1$ and $x_2$ (whose coefficients depend on $u$) given by:
\be
a_i(x,u)=\lambda_i s_2(x)\left[t_i(x)+6 V_c e^{2u} s_2(x) s_3(x)\right]~~,
\ee
with:
\beqa
&& t_1(x)\eqdef \lambda_1 \lambda_2^2(\lambda_1-\lambda_2) x_2^2 [s_2(x)-3 \lambda_2 s_1(x)]\nn\\
&& t_2(x)\eqdef \lambda_2 \lambda_1^2(\lambda_2-\lambda_1) x_1^2 [s_2(x)-3 \lambda_1 s_1(x)]~~
\eeqa
and:
\be
b(x,u)=-\lambda_1^3\lambda_2^3(\lambda_1-\lambda_2)^2 x_1^2 x_2^2 s_1(x)+3 e^{2u} V_c  s_2(x) s_3(x)^2~~.
\ee
\end{prop}

\begin{remark}
Notice that $a_1,a_2$ and $t_1,t_2$ are related by the permutation
which exchanges $\lambda_1$ with $\lambda_2$ and $x_1$ with $x_2$,
while $b$ is invariant under this permutation.
\end{remark}

\begin{proof}
Relations \eqref{gradVlin} show that the quantities $A$ and $B$ approximate as:
\ben
A=\lambda_1 x^1 p_1+\lambda_2 x^2 p_2 +\cO(||x||_0^2)~~,~~B=-\lambda_2 x^2 p_1 +\lambda_1 x^1 p_2 +\cO(||x||_0^2)
\een
and hence are of order one in $||x||$. Accordingly, the order one
approximation of $F$ is obtained by replacing $(\Delta_0 V)(x)$ with its
order zero approximation (which is the constant $\lambda_1+\lambda_2$)
and keeping only the terms linear in $A$ and $B$. Thus:
\be
-F=\big(\tilde{H}^2+6 e^{2 u} V_c  H\big) A+ 2\tilde{H}(H\!-\!\lambda_1 \!-\! \lambda_2) B+
(\lambda_1 +\lambda_2 -H) \tilde{H}^2-3e^{2 u} V_c  H^2 +\cO(||x||_0^2)
\ee
i.e.:
\beqa
&&F=3 e^{2u}V_c  H^2-{\tilde H}^2(\lambda_1 +\lambda_2 -H)\nn\\
&&+\left[2{\tilde H} (H - \lambda_1 - \lambda_2)\lambda_2 x_2 -
  \big(6  e^{2u} V_c H +\tilde{H}^2\big)\lambda_1 x_1\right] p_1 \nn\\
&&- \left[\,\,2{\tilde H} (H - \lambda_1 -\lambda_2 ) \lambda_1 x_1+ 
\big(6 e^{2u} V_c  H+ {\tilde H}^2\big)\lambda_2 x_2 \right]p_2
+\cO(||x||_0^2)~~.
\eeqa
Substituting \eqref{Hleading} in this expression gives the conclusion.
\end{proof}

\begin{cor}
The contact Hamilton-Jacobi equation \eqref{Feq} is approximated to
first order in $||x||_0$ by the following quasilinear first order PDE:
\ben
\label{quasilinear}
a_1(x,\phi) x^1 \pd_1 \phi+a_2(x,\phi) x^2 \pd_2 \phi=b(x,\phi)~~.
\een
\end{cor}

\begin{remark}
One can check that equation \eqref{quasilinear} is {\em not} proper in
the sense of the theory of viscosity solutions. For reasons of
practical computability, we will nonetheless use a viscosity
perturbation of this equation when numerically approximating its
solutions in Subsection \ref{subsec:numquasilinear}.
\end{remark}

\begin{remark}
Suppose that $\lambda_1=\lambda_2=\lambda\neq 0$. Then the quasilinear
PDE \eqref{quasilinear} reduces to:
\be
(x_1\pd_1+x_2\pd_2)\phi=\frac{1}{2}~~,
\ee
which has general solution (see Proposition \ref{prop:lineqequal} below): 
\be
\phi(r,\theta)=\frac{1}{2}\log r +\phi_0(\theta)~~,
\ee
where $\phi_0$ is a piecewise-smooth real-valued function defined on
the unit circle and $(r,\theta)$ are the polar coordinates defined by
the isothermal coordinate system $(x^1,x^2)$. In this case, the
quasilinear equation \eqref{quasilinear} is equivalent with the linear
equation in Proposition \ref{prop:linear} which we discuss later
on. We will exclude this special case in most of the following.
\end{remark}

\subsection{Symmetries}

Since $a_i$ and $b$ are homogeneous polynomials of the same degree in
$x_1$ and $x_2$, the quasilinear equation \eqref{quasilinear} is
invariant under the homotheties $x^i\rightarrow \kappa x^i$ with
$\kappa>0$.  If $\phi$ is a solution, it follows that the function
$\phi_\kappa$ defined by $\phi_\kappa(x)=\phi(\kappa x)$ is also a
solution. Moreover, equation \eqref{quasilinear} is invariant under
the action of the Klein four-group $\Z_2\times \Z_2$ generated by
$x_1\rightarrow -x_1$ and $x_2\rightarrow -x_2$. Thus given a solution
$\phi$, the functions
$\phi_{\epsilon_1,\epsilon_2}$ defined through:
\be
\phi_{\epsilon_1,\epsilon_2}(x_1,x_2)=\phi(\epsilon_1 x_1,\epsilon_2x_2)~~,~~\mathrm{where}~~\epsilon_1,\epsilon_2\in \{-1,1\}
\ee
are also solutions. Combining these observations, we conclude that the
group $\R_{>0}\times \Z_2\times \Z_2$ acts naturally on the set of
solutions to \eqref{quasilinear}.

\subsection{Scale-invariant solutions of the quasilinear equation}
\label{subsec:scale}

It is natural to look for scale-invariant solutions of the quasilinear
equation \eqref{quasilinear}, i.e. solutions which are fixed points of
the scaling action of $\R_{>0}$. Such solutions satisfy:
\be
\phi(\kappa x)=\phi(x)\quad\forall \kappa>0
\ee
and hence depend only on the angular coordinate $\theta$ of the polar coordinate system. Since:
\be
\pd_1 \phi=-\frac{\sin\theta}{r}\phi'(\theta)~~,~~\pd_2\phi=\frac{\cos\theta}{r}\phi'(\theta)~~,
\ee
equation \eqref{quasilinear} reduces to the ODE:
\be
[a_2(\cos\theta,\sin\theta,\phi)-a_1(\cos\theta,\sin\theta,\phi)]\cos\theta\sin\theta \phi'(\theta)=b(\cos\theta,\sin\theta)~~.
\ee
The change of variables $t=\cos(2\theta)\in [-1,1]$ brings the last equation to the form:
\ben
\label{Req}
\frac{\dd\phi}{\dd t}=R(t,\phi)~~,
\een
where:
{\scriptsize \beqa
  &&\!\!\!\!\!\!\!\!\!\!\!\!\!\!R(t,\phi)\!=\frac{1}{(\!t^2\!-\!1\!) \!\big[\lambda_1^2\!(t\!+\!1\!)
      \!\!-\!\lambda_2^2\!(t\!-\!1\!)\big](\!\lambda_1\!\!-\!\lambda_2\!) }\times\nn\\ &&\!\!\!\!\!\!\!\!\!\!\!\!\!\!\!\!\!\!\frac{3V_c e^{2\phi }\big[\lambda_1^2(t\!+1)\!-\lambda_2^2(t\!-1)\big]\big[\lambda_1^3(t\!+1)\!-\!\lambda_2^3(t\!-1)\big]^2\!+\lambda_1^3(\lambda_1\!-\!\lambda_2)^2
      \lambda_2^3(t^2\!-1)\big[\lambda_1(t\!+1)\!-\!\lambda_2(t\!-1)\big]}
     { \lambda_1^2\!\lambda_2^2\big[\!-\!6\lambda_2\!\lambda_1
      \!(t^2\!-\!1\!)\!\!+\!\lambda_1^2(t\!+\!1\!)(3t\!+\!1\!)\!+\!\lambda_2^2\!(t\!-\!1\!)
      \!(3t\!-\!1\!)\big]\!\!-\! 6V_c e^{2 \phi}\!\big[\lambda_1^2\!(t\!+\!1\!)\!-\!\!\lambda_2^2
         \!(t\!-\!1\!)\big]\big[\lambda_1^3\!(t\!+\!1\!)\!\!-\!\!\lambda_2^3\!(t\!-\!1\!)\big]}~.\\
\eeqa}
\noindent\!\!Since the right hand side of \eqref{Req} can become
singular in the interval $(-1,1)$, a locally-defined classical
solution of this equation need not extend to a differentiable function
defined on the entire interval. Below, we consider generalized
solutions which are of class $\cC^1$ on the interval $(-1,1)$ except
for a finite set of points of this interval, where they can tend to
infinity or have discontinuous derivative. However, we require that the
restriction of a generalized solution to the open interval $(-1,1)$ is
continuous as a map valued in $\bar{\R}=\R\cup \{-\infty,\infty\}$.
  
Since a generalized solution depends only on $t=\cos(2\theta)$, it
follows that $\phi$ as a function of $\theta$ has periodicity $\pi$
and is invariant under the transformation $\theta\rightarrow -\theta$:
\be
\phi(\theta+\pi)=\phi(\theta)~~,~~\phi(-\theta)=\phi(\theta)~~.
\ee
Hence as a function defined on the unit circle, $\phi$ is invariant
under the Klein four-group $\Z_2\times \Z_2$ generated by reflections
in the two coordinate axes\footnote{Notice that we allow
$\phi(\theta)$ to have different directional limits from the two sides
at $\theta\in \{0,\frac{\pi}{2},\pi, \frac{3\pi}{2}\}$.}. Thus
$\phi$ is also a fixed point of the discrete symmetries of equation
\eqref{quasilinear} and we can restrict its domain of definition to
$\theta\in [0,\pi/2]$ without losing information. The range $t\in
      [-1,1]$ corresponds to $\theta\in [0,\pi/2]$, where $\theta=0$
      corresponds to $t=1$ and $\theta=\pi/2$ corresponds to $t=-1$.

The denominator of $R(t,\phi)$ factors as $T_1(t,\phi)T_2(t,\phi)$, where
{\footnotesize \beqa
&&\!\!\!\!\!\!T_1= (t^2-1)[(\lambda_1^2-\lambda_2^2)t+\lambda_1^2+\lambda_2^2](\lambda_2-\lambda_1)\nn\\
&&\!\!\!\!\!\!T_2=3 (\lambda_1\!-\!\lambda_2)^2 t^2 \big[2 (\lambda_1\!+\!\lambda_2) (\lambda_1^2\!+\!\lambda_2 \lambda_1\!+\!\lambda_2^2) e^{2 \phi } V_c\!-\!\lambda_1^2 \lambda_2^2\big]
   +4 t \big[3 (\lambda_1^5\!-\!\lambda_2^5) e^{2 \phi } V_c\!-\!\lambda_2^2 \lambda_1^4
   \!+\!\lambda_2^4 \lambda_1^2\big] \nn\\
 &&\quad\quad  +6 (\lambda_1^5+\lambda_2^2 \lambda_1^3+\lambda_2^3 \lambda_1^2+\lambda_2^5) e^{2 \phi } V_c-\lambda_1^2 \lambda_2^2 (\lambda
  _1^2+6 \lambda_2 \lambda_1+\lambda_2^2)~~.
\eeqa}
\noindent\!\!\!When $\lambda_1=\lambda_2$, this denominator vanishes
identically while the numerator of $R$ equals $24 \lambda^8 V_c
e^{2\phi}$, so in this case equation \eqref{Req} has no
solutions. When $\lambda_1\neq \lambda_2$, the denominator of
$R(t,\phi)$ vanishes for $t\in
\left\{-1,1,-\frac{\lambda_1^2+\lambda_2^2}{\lambda_1^2-\lambda_2^2}\right\}$
(which are the roots of $T_1(t,\phi)$ as a polynomial of $t$) and for
$t\in \{t_-(\phi),t_+(\phi)\}$, where:
{\footnotesize \be
  t_\pm(\phi)=\frac{-6 (\lambda_1^4+\lambda_2 \lambda_1^3+\lambda
    _2^2 \lambda_1^2+\lambda_2^3 \lambda_1+\lambda_2^4) e^{2
      \phi} V_c +\lambda_1^2\lambda_2^2\left(2\lambda_1+2\lambda_2\pm \sign(\lambda_1-\lambda_2) \sqrt{d}\right)}
{3(\lambda_1-\lambda_2)\left[ 2 (\lambda_1+\lambda_2)(\lambda_1^2+\lambda_1\lambda_2\lambda_2^2)e^{2\phi} V_c -\lambda_1^2\lambda_2^2\right]}
\ee}
\!\!are the roots of $T_2(t,\phi)$ as a polynomial of $\phi$. These
roots are real iff the quantity:
\be
d\eqdef 36 e^{4\phi} V_c^2+60 (\lambda_1+\lambda_2) e^{2 \phi} V_c+\lambda_1^2+\lambda_2^2-10 \lambda_1 \lambda_2
\ee
(which is proportional to the discriminant of $T_2$) is non-negative.

Since
$\Big{|}\frac{\lambda_1^2+\lambda_2^2}{\lambda_1^2-\lambda_2^2}\Big{|}>1$,
the only potentially singular points of $R$ which lie in the interval
$[-1,1]$ are $-1,1$ and the points of the set $[-1,1]\cap
\{t_-(\phi),t_+(\phi)\}$. If any of the points $t_\pm(\phi)$ is contained
in the interval $[-1,1]$, then the derivative of a solution of
\eqref{Req} may tend to infinity at that point, depending on the
behavior of the numerator of $R$. Let us consider the ``worst'' kind
of singularity, which arises when both $|\phi(t)|$ and
$|\frac{\dd\phi(t)}{\dd t}|=|R(t,\Phi(t))|$ tend to infinity as we
approach the singular point. To understand where this kind of
singularity may occur, it suffices to consider the cases
$\phi=+\infty$ and $\phi=-\infty$ in the formulas for
$t_\pm(\phi)$. We have:
\beqa
&&t_+(+\infty)=\twopartdef{-\frac{\lambda_1^3+\lambda_2^3}{\lambda_1^3-\lambda_2^3}}{\lambda_1>\lambda_2}{-\frac{\lambda_1^2+\lambda_2^2}{\lambda_1^2-\lambda_2^2}}
    {\lambda_1<\lambda_2}\nn\\ &&t_-(+\infty)=\twopartdef{-\frac{\lambda_1^2+\lambda_2^2}{\lambda_1^2-\lambda_2^2}}
    {\lambda_1>\lambda_2}{-\frac{\lambda_1^3+\lambda_2^3}{\lambda_1^3-\lambda_2^3}}{\lambda_1<\lambda_2}
    \eeqa
and
\beqa
&&t_+(-\infty)=\twopartdef{-\frac{2(\lambda_1+\lambda_2)+\sqrt{\lambda_1^2-10\lambda_1\lambda_2+\lambda_2^2}}{3(\lambda_1-\lambda_2)}}
      {\lambda_1>\lambda_2}{-\frac{2(\lambda_1+\lambda_2)-\sqrt{\lambda_1^2-10\lambda_1\lambda_2+\lambda_2^2}}{3(\lambda_1-\lambda_2)}}{\lambda_1<\lambda_2}
      \nn\\ &&t_-(-\infty)=\twopartdef{-\frac{2(\lambda_1+\lambda_2)-\sqrt{\lambda_1^2-10\lambda_1\lambda_2+\lambda_2^2}}{3(\lambda_1-\lambda_2)}}
            {\lambda_1>\lambda_2}{-\frac{2(\lambda_1+\lambda_2)+\sqrt{\lambda_1^2-10\lambda_1\lambda_2+\lambda_2^2}}{3(\lambda_1-\lambda_2)}}{\lambda_1<\lambda_2}
\eeqa
The roots $t_\pm(-\infty)$ are real iff the quantity:
\ben
\label{delta}
\delta\eqdef \lambda_1^2-10\lambda_1\lambda_2+\lambda_2^2
\een
satisfies:
\ben
\label{detcond}
\delta\geq 0 \Longleftrightarrow \frac{\lambda_1}{\lambda_2} \not\in (5-2\sqrt{6}, 5+2\sqrt{6})~~.
\een
Notice that $t_\pm(+\infty)$ and $t_\pm(-\infty)$ do not depend on
$V_c$. Some aspects of the qualitative behavior of $\phi$ depend
on how many of these values lie in the interval $[-1,1]$.

\begin{lemma}
\label{lemma:roots}
Suppose that $\lambda_1,\lambda_2\in \R\setminus \{0\}$ satisfy $\lambda_1\neq \lambda_2$. Then the following statements hold:
\begin{enumerate}
\item  If $\lambda_1<\lambda_2$, then we have $t_+(+\infty)\not \in [-1,1]$ and $t_-(+\infty)\in \R\setminus \{-1,1\}$. Moreover, we have:
\be
t_-(+\infty)\in (-1,1)~~\mathrm{iff}~~\lambda_1\lambda_2<0~~.
\ee
\item  If $\lambda_1>\lambda_2$, then we have $t_-(+\infty)\not \in [-1,1]$ and $t_+(+\infty)\in \R\setminus \{-1,1\}$. Moreover, we have:
\be
t_+(+\infty)\in (-1,1)~~\mathrm{iff}~~\lambda_1\lambda_2<0~~.
\ee
\item When $\lambda_1^2-10\lambda_1\lambda_2+\lambda_2^2\geq 0$, we always have: 
\be
t_\pm(-\infty)\in (-1,1)~~.
\ee
\end{enumerate}
\end{lemma}

\begin{proof}
We have:
\ben
t_+(+\infty)^2-1=\twopartdef{\frac{4(\lambda_1\lambda_2)^3}{(\lambda_1^3-\lambda_2^3)^2}}{\lambda_1>\lambda_2}{\frac{4(\lambda_1\lambda_2)^2}{(\lambda_1^2-\lambda_2^2)^2}}{\lambda_1<\lambda_2}
\een
and
\ben
t_-(+\infty)^2-1=\twopartdef{\frac{4(\lambda_1\lambda_2)^2}{(\lambda_1^2\lambda_2^2)^2}}{\lambda_1>\lambda_2}{\frac{4(\lambda_1\lambda_2)^3}{(\lambda_1^3-\lambda_2^3)^2}}{\lambda_1<\lambda_2}~~.
\een
This implies the first two points of the lemma. On the other hand, we have:
\ben
t_-(-\infty)^2-1=\twopartdef{\cA+\cB}{\lambda_1>\lambda_2}{\cA-\cB}{\lambda_1<\lambda_2}
\een
and
\ben
t_-(-\infty)^2-1=\twopartdef{\cA-\cB}{\lambda_1>\lambda_2}{\cA+\cB}{\lambda_1<\lambda_2}~~,
\een
where:
\beqan
&&\cA=-\frac{4(\lambda_1^2-4\lambda_1\lambda_2+\lambda_2^2)}{9(\lambda_1-\lambda_2)^2}\nn\\
&&\cB=\frac{4(\lambda_1+\lambda_2)\sqrt{\lambda_1^2-10\lambda_1\lambda_2+\lambda_2^2}}{9(\lambda_1-\lambda_2)^2} ~~.
\eeqan
We have:
\be
\cA^2-\cB^2=\frac{64\lambda_1^2\lambda_2^2}{9(\lambda_1-\lambda_2)^4}>0~~,
\ee
which shows that $|\cB|<|\cA|$ and hence the sign of the quantity
$t_\pm(-\infty)^2-1$ always coincides with that of $\cA$. It is easy
to see that the inequality $\delta\geq 0$ implies
$\lambda_1^2-4\lambda_1\lambda_2+\lambda_2^2>0$, which ensures that
$\cA$ is always negative or zero. It follows that we have
$t_\pm(-\infty)^2-1\leq 0$ and we conclude.
\end{proof}

\noindent Lemma \ref{lemma:roots} implies the following: 

\begin{prop}
\label{prop:sings}
Suppose that $\lambda_1,\lambda_2\in \R\setminus \{0\}$ satisfy
$\lambda_1\neq \lambda_2$. Then the following statements hold:
\begin{enumerate}
\item If $\delta<0$, then the absolute value of a generalized solution of
  the ODE \eqref{Req} can tend to infinity together with that of its
  derivative in the interior of the interval $(-1,1)$ iff
  $\lambda_1\lambda_2<0$. In this case, there is at most one such
  singular point, namely $t_-(+\infty)$ if $\lambda_1<\lambda_2$,
  respectively $t_+(+\infty)$ if $\lambda_1>\lambda_2$.
\item If $\delta\geq 0$, then a generalized solution of \eqref{Req}
  has at most three points in the interval $(-1,1)$ where its absolute
  value can tend to infinity together with that of its
  derivative. These can be any or both of the points $t_\pm(-\infty)$
  (which coincide when $\delta=0$) and at most one of the following
  points, which belongs to the interval $(-1,1)$ only when
  $\lambda_1\lambda_2<0$:
\begin{enumerate}[(a)]
\item $t_-(+\infty)$, if $\lambda_1<\lambda_2$,
\item $t_+(+\infty)$, if $\lambda_1>\lambda_2$.
\end{enumerate}
\end{enumerate}
\end{prop}

\begin{remark}
Whether the derivative of $\phi$ becomes infinite at any point listed
in Proposition \ref{prop:sings} depends on whether $R(t,\phi(t))$
actually blows up at that point. Even in that case, it is not assured
that $|\phi|$ tends to infinity there. The proposition merely
constrains the position of {\em potential} singularities of the kind
considered, but does {\em not} guarantee that any of the points listed
in the statement arises as a singularity in any given solution.
\end{remark}

\noindent Figure \ref{fig:roots} shows $t_\pm(+\infty)$ and
$t_\pm(-\infty)$ as functions of $\lambda_1$ and $\lambda_2$, while
Figure \ref{fig:ScaleSol} shows the particular solutions $\phi(t)$ of
equation \eqref{Req} for a few choices of $V_c$, $\lambda_1$ and
$\lambda_2$ with the initial conditions listed. The black and red
  dots in Figure \ref{fig:ScaleSol} represent those points among $t_\pm(+\infty)$, respectively
  $t_\pm(-\infty)$ which lie inside the interval $(-1,1)$. We show solutions
whose restrictions to the open interval $(-1,1)$ are continuous as
maps valued in $\bar{\R}=\R\cup \{-\infty,\infty\}$.

\begin{figure}[H]
\centering
\begin{minipage}{.9\textwidth}
\centering ~~ \includegraphics[width=.7\linewidth]{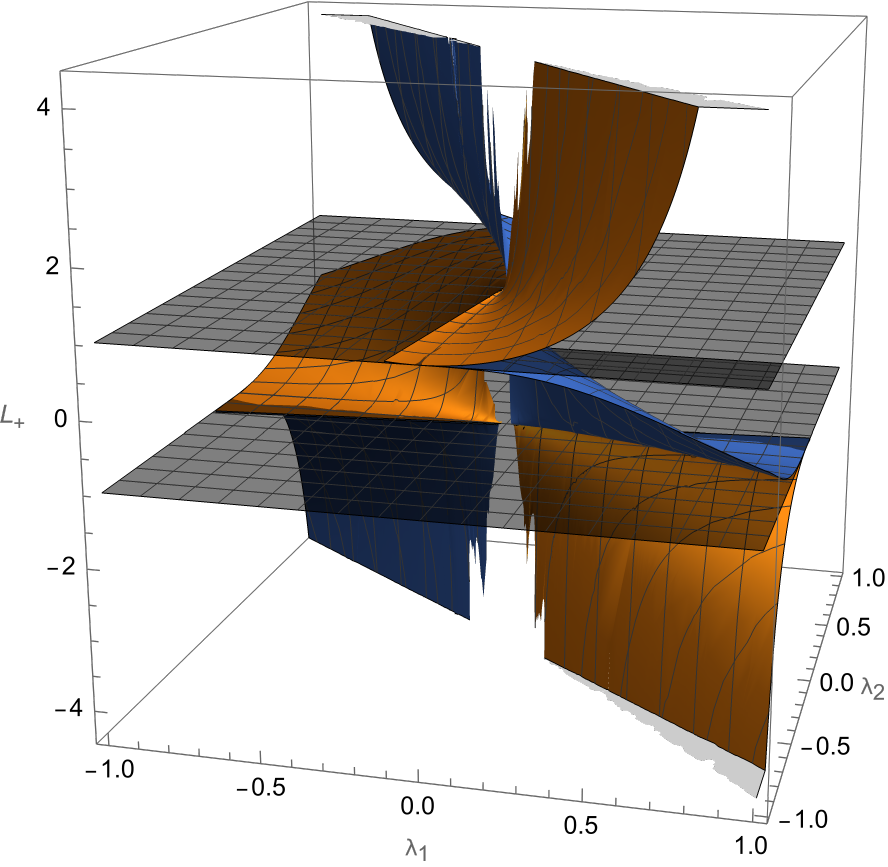}
\vspace{1em}
\subcaption{$t_\pm(+\infty)$ as functions of $\lambda_1$ and $\lambda_2$. The rigged regions in the plot
are numerical artifacts.}
\end{minipage}
\hfill
\centering
\begin{minipage}{.95\textwidth}
\centering \includegraphics[width=.7\linewidth]{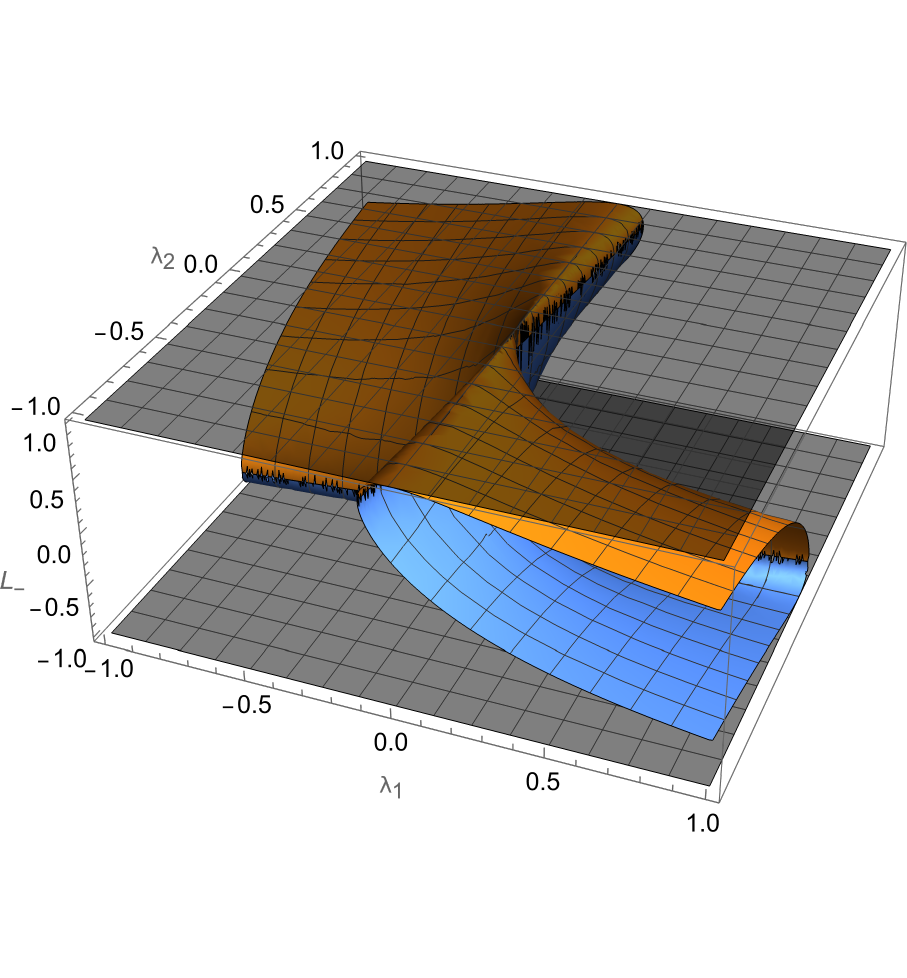}
\vspace{-2em}
\subcaption{$t_\pm(-\infty)$ as functions of $\lambda_1$ and $\lambda_2$.}
\end{minipage}
\hfill
\caption{Plots of $t_\pm(+\infty)$ and $t_\pm(-\infty)$ as functions of 
$\lambda_1$, $\lambda_2$.}
\label{fig:roots}
\end{figure}

\begin{figure}[H]
\centering
\begin{minipage}{0.53\textwidth}
\centering ~~ \includegraphics[width=.64\linewidth]{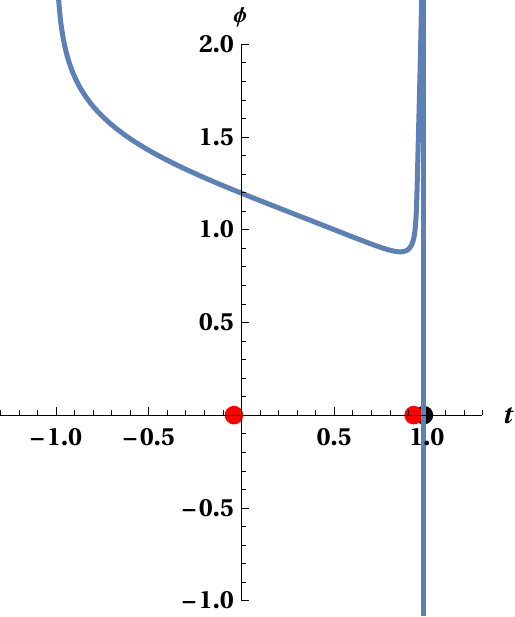}
\subcaption{The solution $\phi$ for $V_c\!=\!1/90$,
  $\lambda_1\!=\!-1/5$, $\lambda_2\!=\!1$ with initial condition
  $\phi(1/2)=1$. We have $\delta\approx 3.04$, $t_-(+\infty)\approx 0.98 $,
  $t_+(+\infty)\approx 1.08$, $t_-(-\infty)\approx  -0.04$,
  $t_+(-\infty)\approx  0.93$.}
\end{minipage}
\centering
\hspace{1em}
\begin{minipage}{0.4\textwidth}
\centering \includegraphics[width=\linewidth]{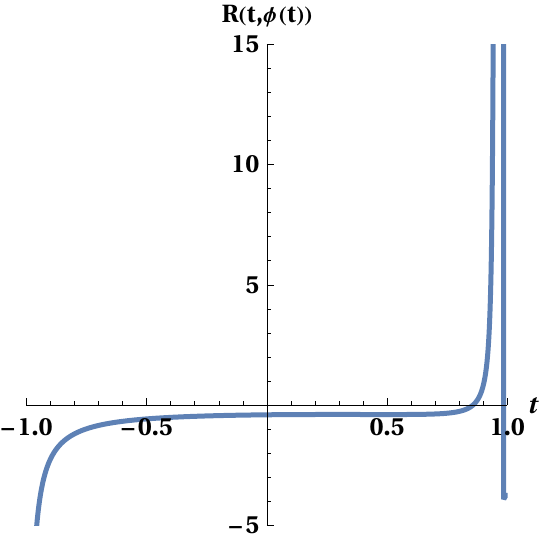}
\subcaption{Plot of $\phi'(t)=R(t,\phi(t))$ where $\phi$ is the
  solution shown on the left.}
\end{minipage}
\hfill
 \centering
\begin{minipage}{0.53\textwidth}
\centering \includegraphics[width=.46\linewidth]{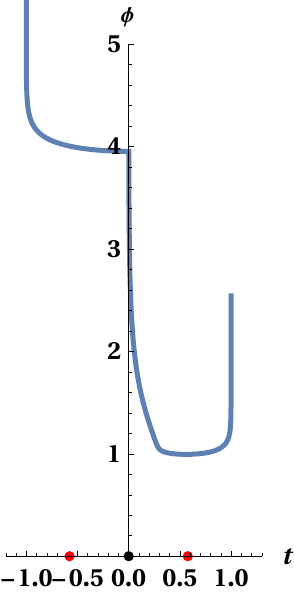}
\subcaption{The solution $\phi$ for $V_c\!=\!1/18$,
  $\lambda_1\!=\!-1$, $\lambda_2\!=\!1$ with initial condition
  $\phi(0.379)=1.01$. We have $\delta=12$, $t_-(+\infty)=0$,
  $t_+(+\infty)=-\infty $, $t_-(-\infty)\approx -0.57$,
  $t_+(-\infty)\approx 0.57$.}
\end{minipage}
\centering
\hspace{1em}
\begin{minipage}{0.4\textwidth}
\centering \includegraphics[width=\linewidth]{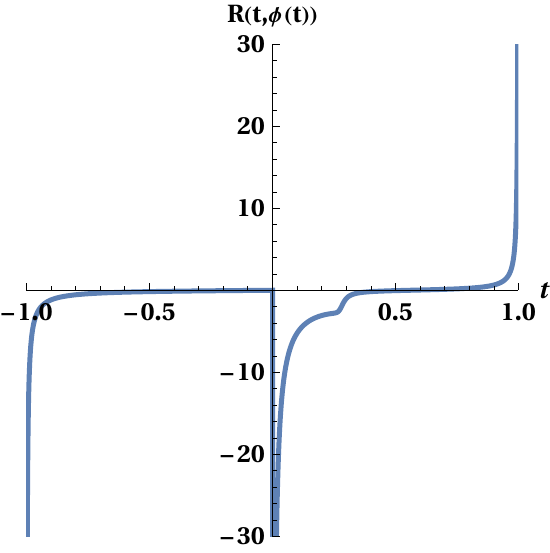}
\subcaption{Plot of $\phi'(t)=R(t,\phi(t))$ where $\phi$ is the
  solution shown on the left.}
\end{minipage}
\hfill
\\
\centering
\begin{minipage}{.53\textwidth}
  \centering \includegraphics[width=.7\linewidth]{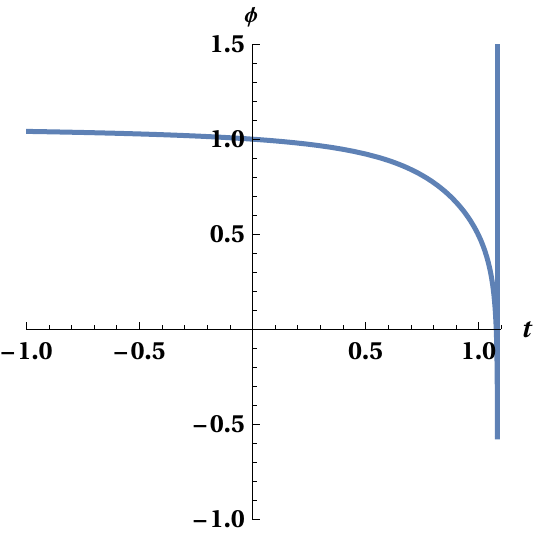}
  \subcaption{The solution $\phi$ for $V_c\!=\!10^{-10}$,
    $\lambda_1\!=\!1/5$, $\lambda_2\!=\!1$ with initial condition
    $\phi(0)=1$. We have $\delta=-0.95$, $t_-(+\infty)\approx 1.01$,
    $t_+(+\infty)\approx 1.08$, $t_-(-\infty)\approx 1 -0.4 \i$,
    $t_+(-\infty)\approx 1 +0.4\i$.}
\end{minipage}
\centering
\hspace{1em}
\begin{minipage}{0.4\textwidth}
\centering  \includegraphics[width=\linewidth]{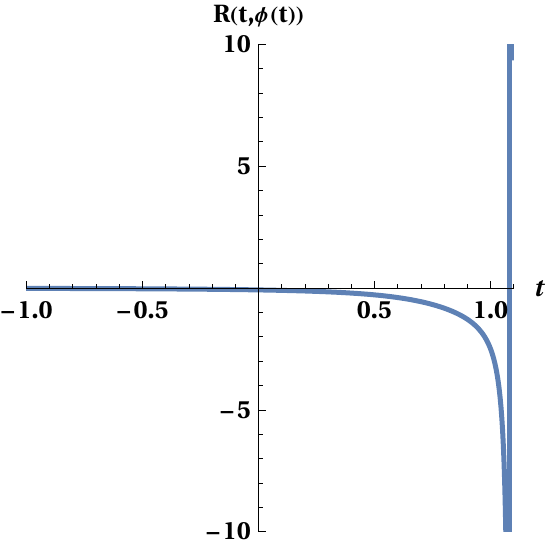}
\subcaption{Plot of $\phi'(t)=R(t,\phi(t))$ where $\phi$ is the solution shown on the left.}
\end{minipage}
\caption{Plots of the solutions $\phi$ of equation \eqref{Req}  and of their derivatives with  the initial conditions listed for various values of $V_c$ and
  $\lambda_1$, $\lambda_2$.}
\label{fig:ScaleSol}
\end{figure}

\subsection{Linearization of the contact Hamilton-Jacobi equation when $\phi$ blows up at a non-degenerate critical point}

The quasilinear equation \eqref{quasilinear} can be approximated by a
linear equation in the special case when $\phi$ tends to $+\infty$ at
the non-degenerate critical point $c$ of $V$.

\begin{prop}
\label{prop:linear}
Suppose that $\phi$ satisfies the quasilinear equation
\eqref{quasilinear} and that we have $\phi(x)\gg 1$ on a vicinity of
the non-degenerate critical point $c$ of $V$. Then $\phi$ is an
approximate solution of the following linear first order PDE:
\ben
\label{linear}
2 s_2(x) \lambda_i x^i \pd_i \phi= s_3(x)~~,
\een
which it satisfies up to corrections of order $\cO(\frac{e^{-2\phi}}{3V_c })$. 
\end{prop}

\begin{proof}
Dividing both terms by $3 V_c  e^{2\phi} $,  equation \eqref{quasilinear} can be written as:
\be
\lambda_i s_2(x)\!\left[t_i(x)\frac{e^{-2\phi}}{3V_c }\!+\!2 s_2(x) s_3(x)\right] \!x^i \pd_i \phi\!=\!
-\lambda_1^3\lambda_2^3(\lambda_1\!-\!\lambda_2)^2 x_1^2 x_2^2 s_1(x) \frac{e^{-2\phi}}{3V_c }+\! s_2(x) s_3(x)^2
\ee
and the terms containing exponentials can be neglected when $\phi(x)\gg 1$.
\end{proof}

To find the solutions of \eqref{linear}, consider the polar coordinate system $(r,\theta)$ defined through:
\ben
\label{polar}
x_1=r\cos \theta~~,~~x_2=r\sin \theta~~
\een
and notice the relations:
\be
\label{pdpolar}
\pd_1 =\cos\theta \frac{\pd}{\pd r} -\frac{\sin\theta}{r}\pd_\theta~\quad~,~\quad~\pd_2 =\sin\theta \frac{\pd}{\pd r} +\frac{\cos\theta}{r}\pd_\theta~~,
\ee
which imply:
\ben
\label{xpdpolar}
x^1\pd_1\phi_0= (\cos^2\theta \, r\pd_r-\sin\theta\cos\theta\pd_\theta) \phi~~,~~x^2\pd_2\phi_0=(\sin^2\theta \, r\pd_r + \sin\theta \cos\theta\pd_\theta)\phi~~.
\een

\begin{prop}
Suppose that $\lambda_1\neq \lambda_2$. Then the general solution of the linear equation \eqref{linear} is:
\ben
\label{gensol}
\phi(r,\theta)=\phi_0(\theta)+Q_0\big(\frac{\lambda_2-\lambda_1}{\lambda_1\lambda_2} \log r+\frac{1}{\lambda_1}\log|\cos\theta|-\frac{1}{\lambda_2}\log|\sin\theta|\big)~~,
\een
where:
\ben
\label{partsol}
\phi_0(\theta)=\frac{1}{4} \log(\lambda_1^2 \cos ^2\theta+\lambda_2^2 \sin ^2\theta)-
  \frac{1}{2} \frac{\lambda_2 \log |\cos\theta|-\lambda_1\log |\sin\theta|}{\lambda_2-\lambda_1}
\een
and $Q_0$ is a function of a single variable which we take to be smooth.
\end{prop}

\noindent Notice that the solution \eqref{gensol} need not be of class
$\cC^1$ if we take $Q_0$ to be smooth. Figures \ref{fig:linear1},
\ref{fig:linear2} and \ref{fig:linear3} show plots of the functions
$\phi_0$ and:
\be
f(r,\theta)\eqdef\frac{\lambda_2-\lambda_1}{\lambda_1\lambda_2} \log r+\frac{1}{\lambda_1}\log|\cos\theta|-\frac{1}{\lambda_2}\log|\sin\theta|
\ee
for a few choices of $\lambda_1$ and $\lambda_2$.

\begin{figure}[H]
\centering
\begin{minipage}{0.45\textwidth}
\centering ~~ \includegraphics[width=1.2\linewidth]{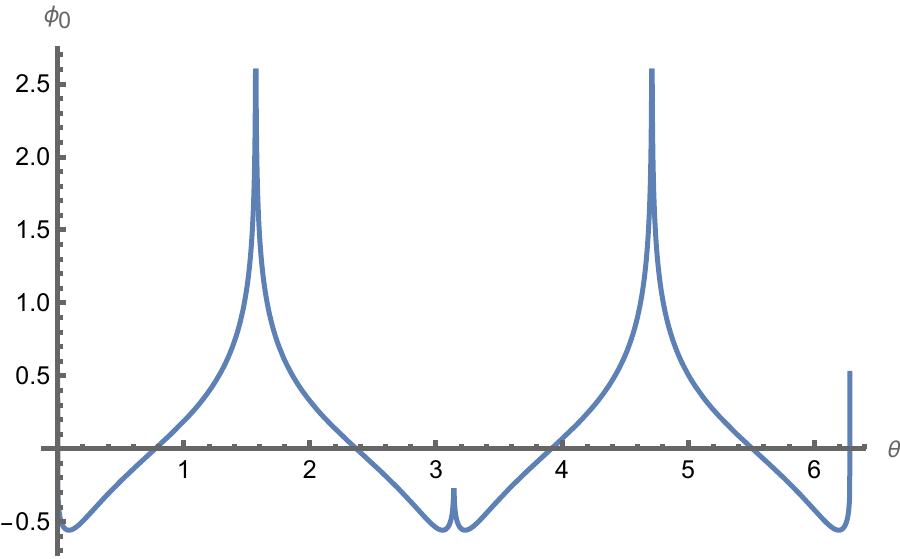}
\subcaption{$\phi_0$ as a function of $\theta$.}
\end{minipage}
\hfill
\centering
\begin{minipage}{.45\textwidth}
\centering ~~ \includegraphics[width=0.9\linewidth]{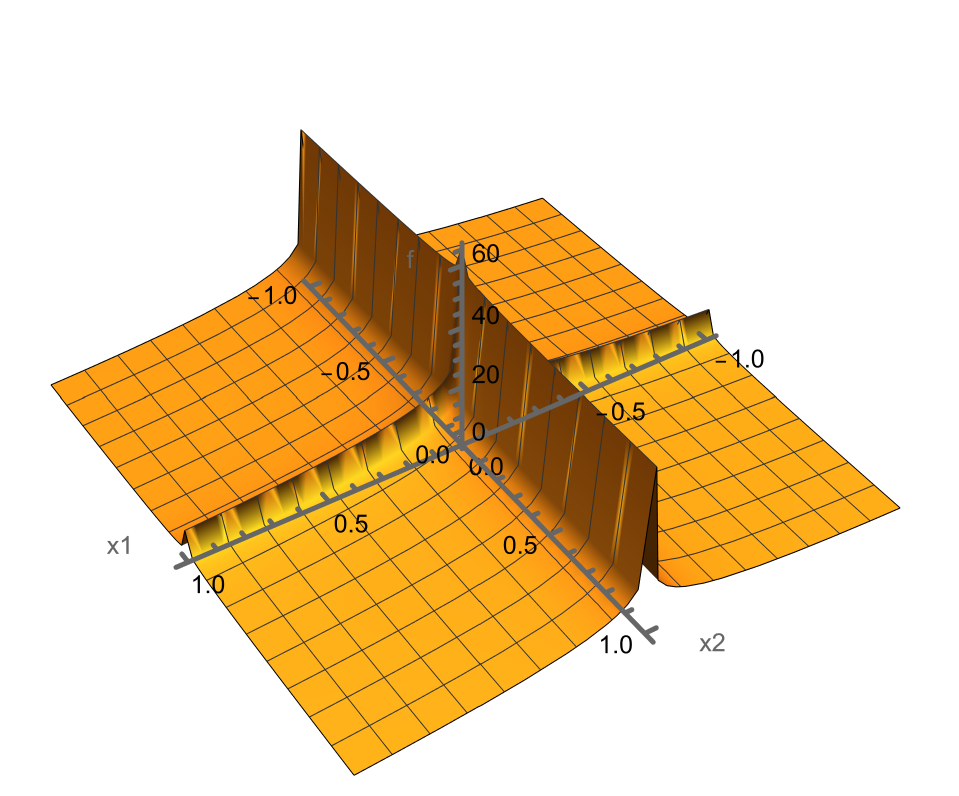}
\subcaption{$f$ as a function of $x_1$ and $x_2$}
\end{minipage}
\hfill
\caption{Plots of $\phi_0$ and $f$ for $\lambda_1=-1/5$ and $\lambda_2=1$.}
\label{fig:linear1}
\end{figure}

\begin{figure}[H]
\centering
\begin{minipage}{0.45\textwidth}
\centering ~~ \includegraphics[width=1.2\linewidth]{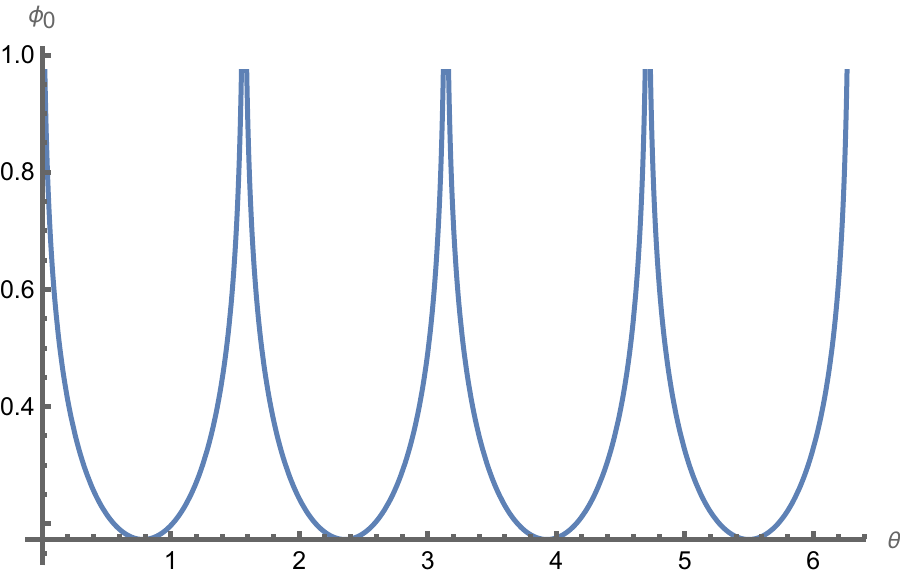}
\subcaption{$\phi_0$ as a function of $\theta$.}
\end{minipage}
\hfill
\centering
\begin{minipage}{.45\textwidth}
\centering ~~ \includegraphics[width=0.9\linewidth]{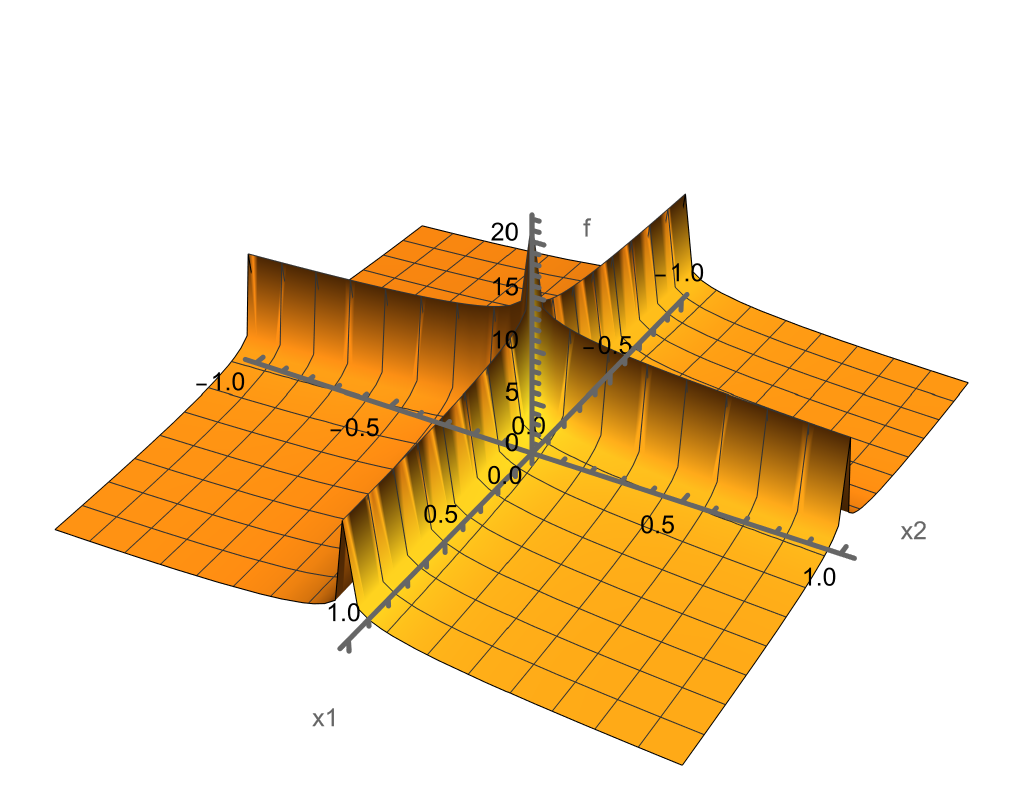}
\subcaption{$f$ as a function of $x_1$ and $x_2$}
\end{minipage}
\hfill
\caption{Plots of $\phi_0$ and $f$ for $\lambda_1=-1$ and $\lambda_2=1$.}
\label{fig:linear2}
\end{figure}

\begin{figure}[H]
\centering
\begin{minipage}{0.45\textwidth}
\centering ~~ \includegraphics[width=1.2\linewidth]{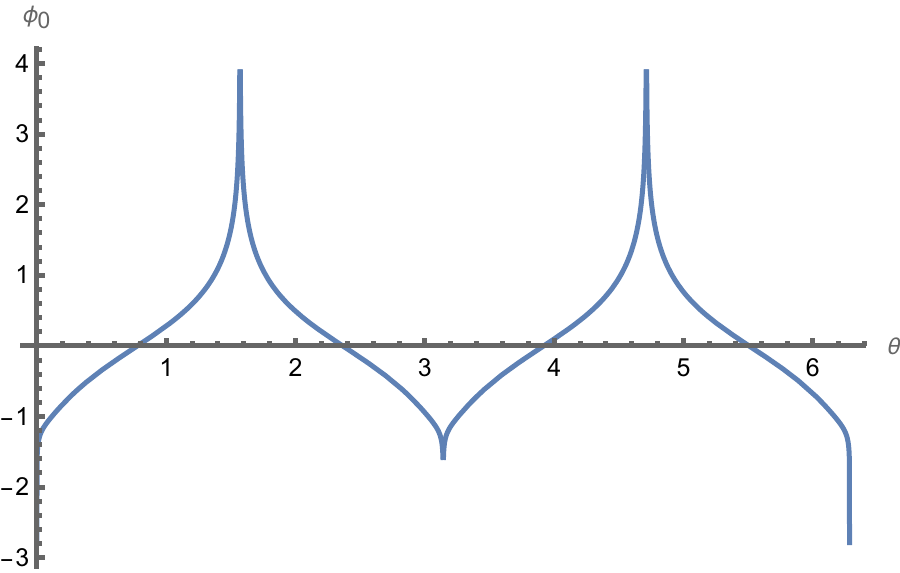}
\subcaption{$\phi_0$ as a function of $\theta$.}
\end{minipage}
\hfill
\centering
\begin{minipage}{.45\textwidth}
\centering ~~ \includegraphics[width=1\linewidth]{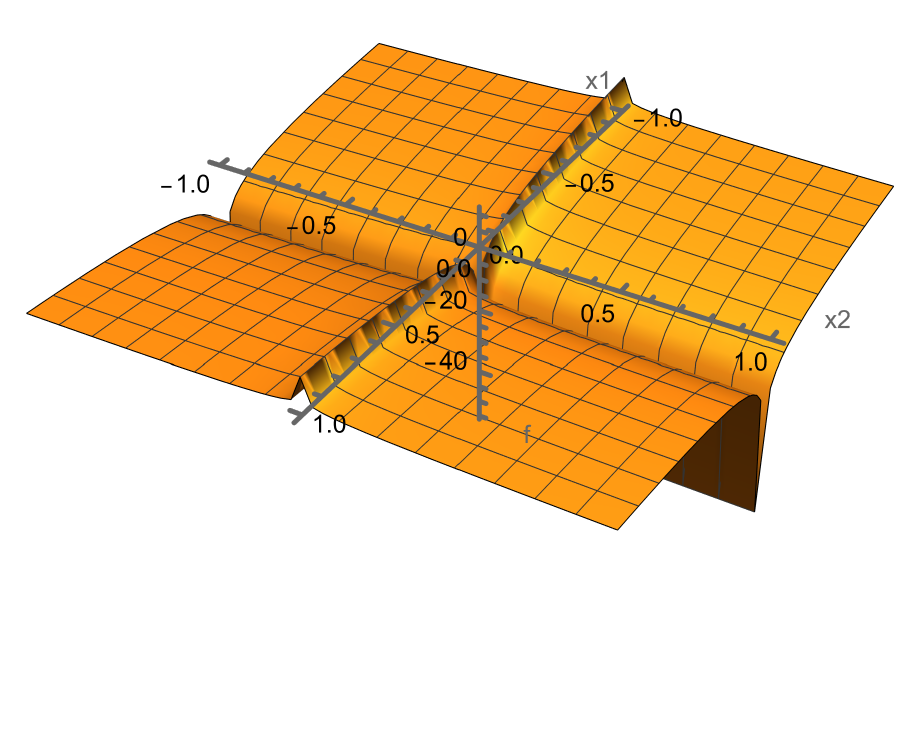}
\subcaption{$f$ as a function of $x_1$ and $x_2$}
\end{minipage}
\hfill
\caption{Plots of $\phi_0$ and $f$ for $\lambda_1=1/5$ and $\lambda_2=1$.}
\label{fig:linear3}
\end{figure}

\begin{proof}
The homogeneous equation associated to \eqref{linear} is equivalent with: 
\be
\lambda_1 x^1 \pd_1 \phi+\lambda_2 x^2 \pd_2\phi=0~~.
\ee
The change of variables $x_i=\sign(x_i)e^{\sigma_i}$ brings this equation to the form:
\be
\lambda_1 \frac{\pd \phi}{\pd \sigma_1}+\lambda_2 \frac{\pd \phi}{\pd \sigma_2}=0~~.
\ee
Let:
\be
\sigma_\pm =\frac{1}{\lambda_1}\sigma_1\pm \frac{1}{\lambda_2}{\sigma_2}~~.
\ee
Changing coordinates from $(\sigma_1,\sigma_2)$ to $(\sigma_+,\sigma_-)$,
the homogeneous equation becomes $\frac{\pd \phi}{\pd \sigma_+}=0$,
which means that $\phi$ depends only on $\sigma_-$. 
Hence the general smooth solution of the homogeneous equation is: 
\be
\phi=Q_0\big(\frac{1}{\lambda_1} \sigma_1-\frac{1}{\lambda_2}\sigma_2\big)=Q_0\big(\frac{1}{\lambda_1} \log|x_1|-\frac{1}{\lambda_2}\log|x_2|\big)~~,
\ee
where $Q_0$ is a function which we take to be smooth. A particular solution
$\phi_0$ of the inhomogeneous equation can be obtained by requiring
that $\phi_0$ depends only on the angular coordinate $\theta$.
Using relations \eqref{xpdpolar}, the inhomogeneous equation for
$\phi_0$ reduces to:
\be
2(\lambda_2-\lambda_1)s_2(\cos\theta,\sin\theta) (\cos\theta\sin\theta) \phi'_0(\theta)=s_3(\cos\theta,\sin\theta)~~,
\ee
where we used the fact that $s_k$ are homogeneous polynomials of degree two in $x_1$ and $x_2$. The last equation reads:
\be
\phi'_0(\theta)=\frac{1}{2(\lambda_2-\lambda_1)} R(\theta)
\ee
where:
\beqa
R(\theta)\eqdef \frac{s_3(\cos\theta,\sin\theta)}{s_2(\cos\theta,\sin\theta)\cos\theta\sin\theta}=
\frac{\lambda_1^3\cos^2\theta+\lambda_2^3\sin^2\theta}{(\lambda_1^2\cos^2\theta+\lambda_2^2\sin^2\theta)\cos\theta\sin\theta}~~.
\eeqa
This gives \eqref{partsol} up to an arbitrary constant which can be absorbed in $Q_0$. 
\end{proof}

\begin{prop}
\label{prop:lineqequal}
Suppose that $\lambda_1=\lambda_2:=\lambda$. Then the linear equation \eqref{linear} reduces to:
\ben
\label{lindeg}
x^i\pd_i \phi=\frac{1}{2}~~,
\een  
whose general solution is: 
\ben
\label{gensoldeg}
\phi(r,\theta)=\frac{1}{2} \log r+\phi_0(\theta)~~,
\een
where $\phi_0\in \cC^\infty(\rS^1)$ is a real-valued function of a single variable defined on the unit circle, which we take to be smooth.
\end{prop}

\begin{proof}
Since $s_k(x)=\lambda^k ||x_0||^2$, equation \eqref{linear} reduces to
\eqref{lindeg}, which takes the following form in polar coordinates:
\ben
\label{lindegpolar}
r\pd_r \phi (r,\theta)=\frac{1}{2} ~~.
\een
The corresponding homogeneous equation is equivalent with: 
\be
\pd_r\phi=0~~,
\ee
with general solution:
\be
\phi=\phi_0(\theta)~~,
\ee
where $\phi_0\in \cC^\infty(\rS^1)$. The conclusion follows by
noticing that $\phi_0=\frac{1}{2}\log r$ is a particular solution of
the inhomogeneous equation.
\end{proof}

\begin{remark}
When $\lambda_1=\lambda_2:=\lambda$, we have $t_1=t_2=0$ and: 
\be
a_1(x,\phi)=a_2(x,\phi)=6 \lambda^8 V_c  e^{2\phi} ||x||_0^6~\quad~,~\quad~b(x,\phi)= 3 \lambda^8 V_c  e^{2\phi} ||x||_0^6 ~~.
\ee
Hence the quasilinear equation \eqref{quasilinear} with $\lambda_1=\lambda_2=\lambda$ 
coincides with the linear first order PDE \eqref{lindeg}.
\end{remark}

\subsection{Natural asymptotic conditions at a non-degenerate critical point with distinct principal values}
\label{subsec:asymptotics}
\
Suppose that $\lambda_1\neq \lambda_2$. Defining:
\be
Q(w)\eqdef  Q_0\left(\frac{\lambda_2-\lambda_1}{\lambda_1\lambda_2} w\right)~~,
\ee
the general solution \eqref{gensol} of the linear equation \eqref{linear} takes the form:
\be
\phi(r,\theta)=\phi_0(\theta)+Q\left(\log r+\frac{\lambda_2 \log|\cos\theta|-\lambda_1\log|\sin\theta|}{\lambda_2-\lambda_1}\right)~~.
\ee
This satisfies the condition
$\lim_{r\rightarrow 0}\phi(r,\theta)=+\infty$ iff $\lim_{w\rightarrow -\infty} Q(w)=+\infty$. In this case, we have:
\be
\phi\approx Q(\log r)~~\mathrm{for}~~r\ll 1~~,
\ee
so $\phi$ becomes rotationally-invariant near the critical point. 
Since in this case $\phi$ approximates a solution of the nonlinear
PDE \eqref{Feq} when $r\rightarrow 0$, it follows that
the scalar field metric $\cG$ defined by that solution of \eqref{Feq}
becomes asymptotically rotationally-invariant near the critical
point. For $r\ll 1$, the Gaussian curvature of this metric is:
\be
K\approx-e^{-2\phi} \Delta \phi\approx - e^{-2\phi} \frac{1}{r} \frac{\dd}{\dd r}(r\frac{\dd\phi}{\dd r})=-e^{-2Q(\log r)} Q''(\log r)~~.
\ee
To further constrain $Q$, it is natural to require that the asymptotic
Gaussian curvature be constant, i.e. $K=K_c$ for some constant
$K_c$. This gives the equation:
\be
e^{-2Q(w)} Q''(w)=K_c~.
\ee
It is also desirable that the metric $\cG$ be geodesically complete
near the critical point (except at that point, which is not part of
$\cM_0$), and we assume this from now on.

If we choose $K_c=0$, then we can take $Q(w)=-w$, which gives:
\be
\phi(r,\theta)\approx_{r\ll 1} -\log r
\ee
and:
\be
\dd s^2\approx_{r\ll 1} \frac{1}{r^2}(\dd r^2+r^2\dd \theta^2)=\dd \rho^2+\dd \theta^2~~,  
\ee
where $\rho\eqdef -\log r$. In this case, $\cG$ asymptotes for $r\ll
1$ to the metric on a flat cylinder.

If we choose $K_c=-1$, then $\cG$ will asymptote at the critical point
$c$ to a hyperbolic metric defined near $c$ which coincides with the
hyperbolic cusp metric. In this case, we have:
\ben
\label{CuspMetric}
\dd s^2\approx \frac{1}{(r\log r)^2}(\dd r^2+r^2\dd \theta^2)
\een
for $r\ll 1$ and: 
\ben
\label{hcusp}
Q(w)=-\log(|w| e^w)=-w-\log |w|~~\quad~~\mathrm{i.e.}~\quad~\phi(r,\theta)\approx_{r\ll 1} -\log(r\log(1/r))~~.
\een
When $V$ is a Morse function defined on $\cM$, relation \eqref{hcusp}
can be viewed as an asymptotic condition for the solution $\phi$ of
the contact Hamilton-Jacobi equation \eqref{Feq}, to be imposed at
each critical point where $V$ has distinct principal values. In
practice, this can be done approximately by choosing simple closed
curves $C_R(c)$ around each critical point $c$ of $V$ which correspond
to Euclidean circles of radius $R\ll 1$ centered at the origin in the
isothermal coordinates centered at $c$ and imposing the Dirichlet
boundary conditions:
\be
\phi\vert_{C_R(c)}=-\log(R\log(1/R))~~.
\ee
This formula inspired the somewhat idiosyncratic choice of boundary
conditions made for the numerical examples contained in the paper.

Now suppose that $\lambda_1=\lambda_2=\lambda$. In this case, relation
\eqref{gensoldeg} gives:
\be
\phi=\frac{1}{2}\log r+Q_0(\theta)\approx_{r\ll 1} \frac{1}{2}\log r~~,
\ee
which tends to {\em minus} infinity when $r\rightarrow 0$. In this
case, $\phi$ need {\em not} approximate a solution of the contact
Hamilton-Jacobi equation \eqref{Feq}.

\subsection{Characteristics in the quasilinear approximation}

The quasilinear equation \eqref{quasilinear} can be solved locally by
the method of characteristics, which in this case reduces to the
Lagrange-Charpit method. The characteristic system is:
\beqan
\label{CharQuasilinear}
&&\frac{\dd x^i}{\dd t}=a_i(x,u)x^i~,\nn\\
&& \frac{\dd u}{\dd t}=b(x,u)~~.
\eeqan
The Dirichlet problem asks for a solution of \eqref{quasilinear}
which satisfies:
\be
\phi(\gamma(q))=\phi_0(q)\quad \forall q\in I~~,
\ee
where $\gamma:I\rightarrow U_0$ is a non-degenerate smooth curve and
$\phi_0:I\rightarrow \R$ is a given smooth function. The Dirichlet
data $(\gamma,\phi_0)$ is called {\em regular} (or {\em
  noncharacteristic}) if the following condition is satisfied:
\be
a_1(\gamma(q),\phi_0(q))\gamma_1(q)\gamma'_2(q)-a_2(\gamma(q),\phi_0(q))\gamma_2(q)\gamma'_1(q)\neq 0\quad \forall q\in I~~.
\ee
For regular data, the solution of the Dirichlet problem is
obtained by computing the one-parameter family of solutions $x(t,q)$,
$u(t,q)$ of \eqref{CharQuasilinear} which satisfy the initial
conditions:
\be
x(0,q)=\gamma(q)~~,~~u(0,q)=\phi_0(q)\quad\forall q\in I~~.
\ee
Then $u$ is obtained by eliminating $t$ and $q$ between the equations:
\be
x^i=x^i(t,q)~~,~~u=u(t,q)~~,
\ee
which produces a functional relation $u=\phi(x^1,x^2)$. The fact that the
function obtained in this manner satisfies equation
\eqref{quasilinear} follows from the relations:
\be
b(x,u)=\frac{\pd u}{\pd t}=\pd_i \phi \frac{\pd x^i}{\pd t}=a_i(x,u)x^i\pd_i \phi~~.
\ee

\subsection{Numerical solutions and characteristics of the quasilinear equation near a critical point}
\label{subsec:numquasilinear}
Let us illustrate the discussion above with the behavior of solutions
of the quasilinear equation \eqref{quasilinear} and their
characteristics. The potential is given by \eqref{Vquad}, where we
take $V_c>0$ and where $\lambda_1,\lambda_2$ are real numbers. We take
the open set $U_0\subset \R^2$ to be a domain contained inside the
region where the potential is positive. In this case, the functions
$a_1,a_2$ and $b$ are given explicitly by:
{\beqan
 \!\!\!\!\!\!a_1(x_1,x_2,u)&\!=\!&\lambda_1 (\lambda_1^2 x_1^2\!+\!\lambda_2^2 x_2^2) \Big[6 e^{2 u} V_c (\lambda_1^2 x_1^2\!+\!\lambda_2^2 x_2^2)
   (\lambda_1^3 x_1^2\!+\!\lambda_2^3 x_2^2)\nn\\
   &&\quad\quad\quad\quad\quad\quad\quad + (\lambda_1\!-\!\lambda_2) \lambda_1 \lambda_2^2 x_2^2 \big[\lambda_1
   (\lambda_1\!-\!3 \lambda_2) x_1^2\!-\!2 \lambda_2^2 x_2^2\big]\Big] \nn\\
\!\!\!\!\!\!a_2(x_1,x_2,u)&\!=\!&\lambda_2 (\lambda_1^2 x_1^2\!+\!\lambda_2^2 x_2^2) \Big[6 e^{2 u} V_c (\lambda_1^2 x_1^2\!+\!\lambda_2^2 x_2^2)(\lambda_1^3 x_1^2\!+\!\lambda_2^3 x_2^2) \nn\\
&& \quad\quad\quad\quad\quad\quad\quad+(\lambda_1\!-\!\lambda_2) \lambda_2 \lambda_1^2 x_1^2 \big[2 \lambda_1^2
   x_1^2\!+\!(3 \lambda_1\!-\!\lambda_2) \lambda_2 x_2^2\big]\Big] \nn \\
\!\!\!\!\!\!\! b(x_1,x_2,u)&\!=\!&3 e^{2 u} V_c (\lambda_1^2 x_1^2\!+\!\lambda_2^2 x_2^2) (\lambda_1^3 x_1^2\!+\!\lambda_2^3 x_2^2)^2 \nn\\
&&\quad\quad\quad\quad\quad\quad\quad\!-\!\lambda_1^3
  (\lambda_1\!-\!\lambda_2)^2 \lambda_2^3 x_1^2 x_2^2 (\lambda_1 x_1^2\!+\!\lambda_2 x_2^2) ~.
\eeqan}
\!\!As before, we consider the Dirichlet condition
$\phi=-\log[R\log(1/R)]$ on a circle of radius $R<1$ centered at the
origin. For the indicated values of $V_c$, $\lambda_1$, $\lambda_2$
and $R$, Figures \ref{fig:Quasilinear1}, \ref{fig:Quasilinear2}, and
\ref{fig:Quasilinear3} show the scalar potential, some characteristic
curves of the quasilinear equation \eqref{quasilinear} and some
solutions of its viscosity perturbation for various values of the
viscosity parameter $\fc$. The Dirichlet problem is irregular when
$\lambda_1=\lambda_2=1$. In that case, the viscosity solution of the
Dirichlet problem is constant; we do not consider this case in the
figures.


\begin{figure}[H]
\centering
\begin{minipage}{.45\textwidth}
\centering \includegraphics[width=\linewidth]{ContourPot-5inv11.pdf}
\vspace{0.5em}
\subcaption{Contour plot of the potential.}
\end{minipage}
\hfill
\begin{minipage}{.45\textwidth}
\centering  \includegraphics[width=1.15\linewidth]{3dPot-5inv11.pdf}
\vspace{2em}
\subcaption{3D plot of the potential.}
\end{minipage}
\hfill
\vspace{2em}
\centering
\begin{minipage}{.45\textwidth}
\centering \includegraphics[width=1.02\linewidth]{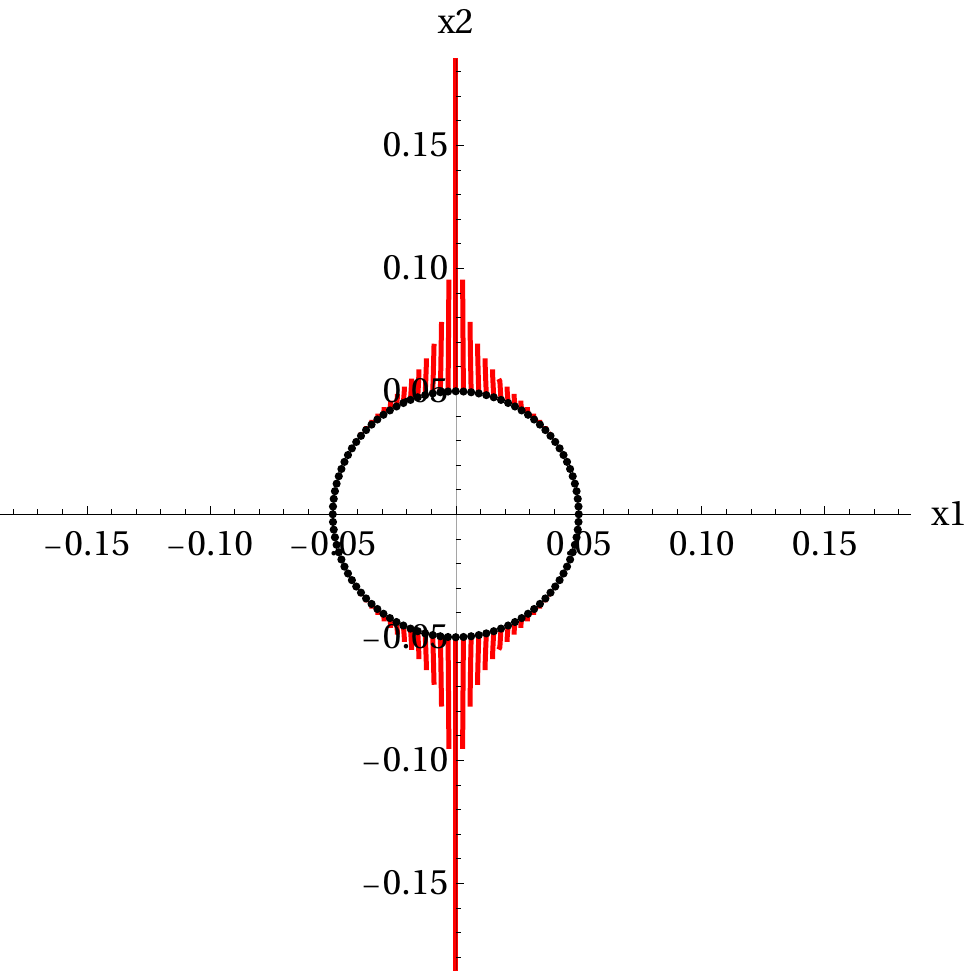}
\vspace{0.5em}
\subcaption{Projected characteristic curves.}
\end{minipage}
\hfill
\begin{minipage}{.5\textwidth}
\centering~~\includegraphics[width=1.05\linewidth]{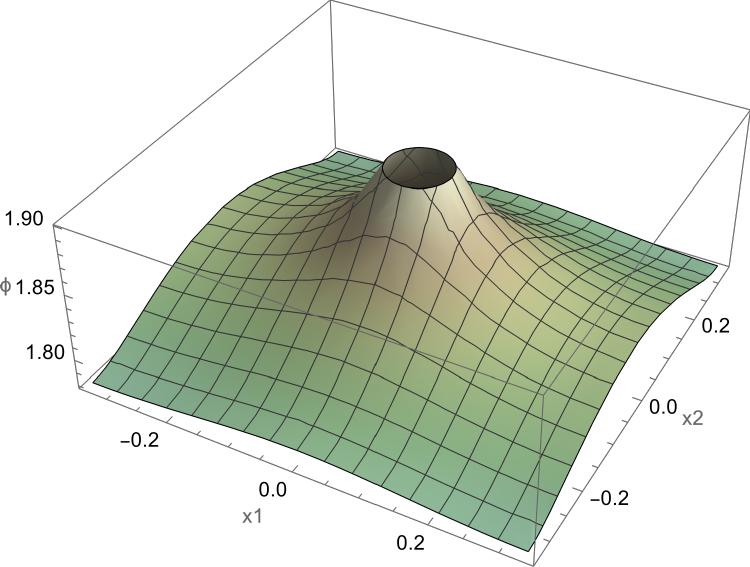}
\vspace{1em}
\subcaption{Viscosity approximant of the solution to the Dirichlet problem for $\fc=e^{-8}$.}
\end{minipage}
\hfill
\vspace{1em}
\caption{The potential, characteristics and a viscosity approximant of
  the solution to the Dirichlet problem for a quasilinear equation
  with $V_c=1/90$, $\lambda_1=-1/5$, $\lambda_2=1$ and $R=1/20$.}
\label{fig:Quasilinear1}
\end{figure}

\vspace{2em}

\pagebreak

\begin{figure}[H]
\vspace{3em}
\centering
\begin{minipage}{.45\textwidth}
\centering  \includegraphics[width=1\linewidth]{ContourPot-11.pdf}
\vspace{0.5em}
\subcaption{Contour plot of the potential.}
\end{minipage}
\hfill
\begin{minipage}{.5\textwidth}
\centering \includegraphics[width=1.05\linewidth]{3dPot-11.pdf}
\vspace{1em}
\subcaption{3D plot of the potential.}
\end{minipage}
\hfill
\vspace{2em}
\centering
\begin{minipage}{.5\textwidth}
\centering \includegraphics[width=0.95\linewidth]{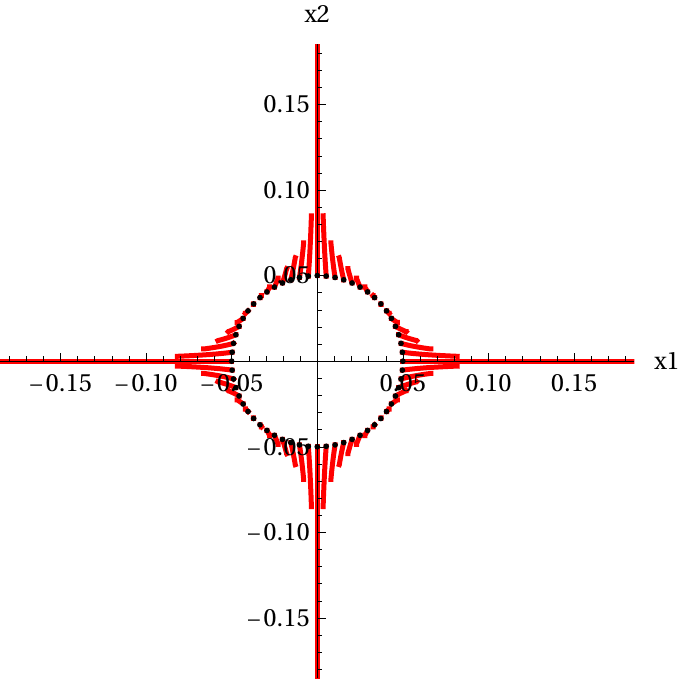}
\subcaption{Projected characteristic curves.}
\end{minipage}
\hfill
\begin{minipage}{.47\textwidth}
\centering \!\!\!\!\!\includegraphics[width=1.15\linewidth]{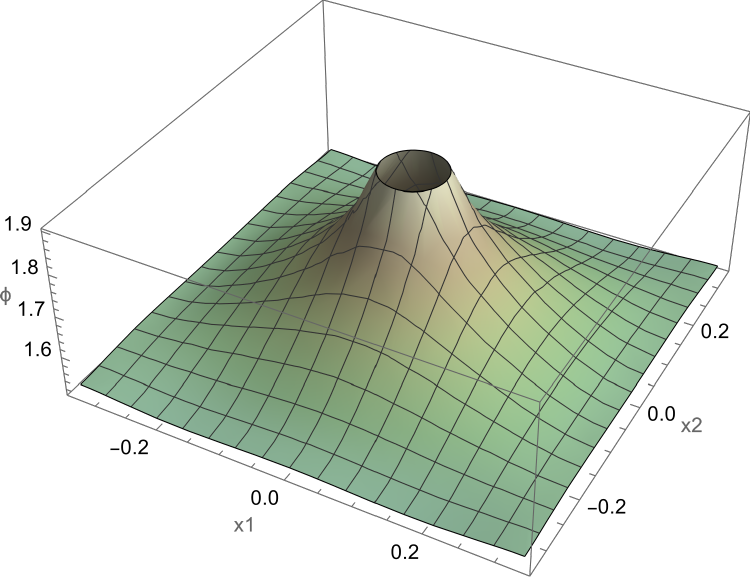}
\vspace{1em}
\subcaption{Viscosity approximant of the solution to the Dirichlet problem for $\fc=e^{-8}$.}
\end{minipage}
\hfill
\vspace{1em}
\caption{The potential, characteristics and a viscosity approximant of
  the solution to the Dirichlet problem for the quasilinear equation
  with $V_c=1/18$, $\lambda_1=-1$, $\lambda_2=1$ and $R=1/20$.}
\label{fig:Quasilinear2}
\end{figure}

\begin{figure}[H]
\vspace{1em}
\centering
\begin{minipage}{.45\textwidth}
\centering \!\!\!\!\!\!\!\!\! \includegraphics[width=1\linewidth]{ContourPot5inv11.pdf}
\subcaption{Contour plot of the potential.}
\end{minipage}
\hfill
\begin{minipage}{.47\textwidth}
\centering \!\!\!\!\!\!\!\!\!\!\!\!\!\!\!\includegraphics[width=1.15\linewidth]{3dPot5inv11.pdf}
\subcaption{3D plot of the potential.}
\end{minipage}
\hfill
\\
\centering
\begin{minipage}{.5\textwidth}
\centering \!\!\!\!\!\!\! \!\!\includegraphics[width=0.95\linewidth]{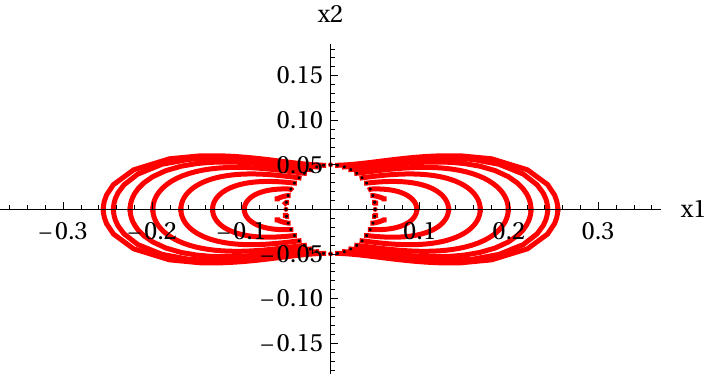}
\vspace{1em}
\subcaption{Projected characteristic curves.}
\end{minipage}
\hfill
\begin{minipage}{.46\textwidth}
\centering \!\!\!\!\!\!\!\!\!\!\!\!\!\!\! \includegraphics[width=1.12\linewidth]{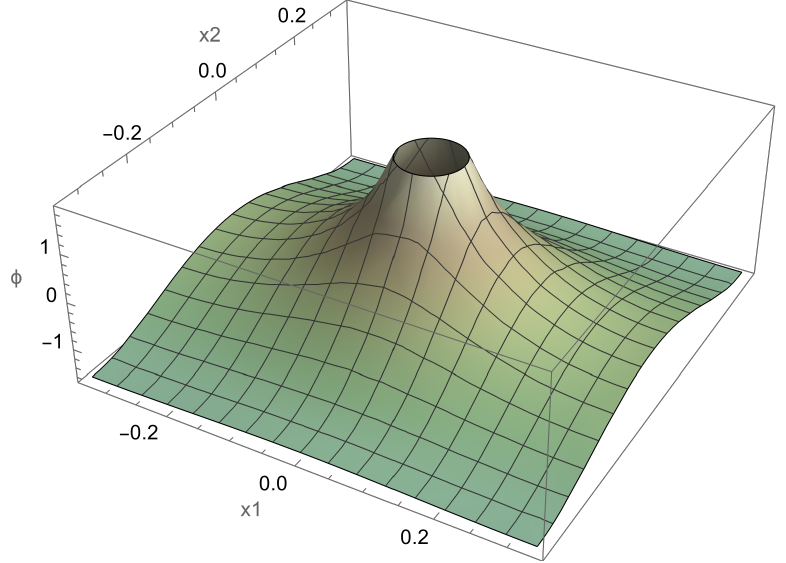}
\subcaption{Viscosity approximant of the solution to the Dirichlet problem for $\fc=e^{-8}$.}
\end{minipage}
\hfill
\caption{The potential, characteristics and a viscosity approximant of
  the solution to the Dirichlet problem for the quasilinear equation
  with $V_c=10^{-2}$, $\lambda_1=1/5$, $\lambda_2=1$ and $R=1/20$. All
  characteristic curves are followed for the same duration of their
  parameter up to limits imposed by the computational capacity of our
  system.}
\label{fig:Quasilinear3}
\end{figure}

\section{Models with 2-torus target}
\label{sec:torus}

In this section, we consider two-field cosmological models whose
target is a 2-torus $\rT^2$. Such models appear in simple examples of
the well-known setting of axion cosmology \cite{Axiverse}, for
example in an M-theory axiverse \cite{AxiverseM} or in a string
axiverse \cite{AxiverseB} realized as an orientifold compactification
of type IIB string theory on a Calabi-Yau manifold which produces two
axions in the uncompactified directions. We refer the reader to
\cite{Marsh} for a review of axion cosmology. While models with two
axions are only toy examples, they are interesting in the
context of our consistency conditions, since the strong SRRT equation
allows one, at least in principle, to determine the fiducial metrics
for which strong rapid turn inflation is likely to be possible.

The scalar potential of axions is likely to be easier to determine in
axiverse cosmology than their scalar field metric -- which, in
compactifications with low or broken supersymmetry, receives
complicated corrections that tend to be hard to control. In this
context, our approach of taking $V$ as given and determining those
scalar field metrics which are ``good'' for strong rapid turn
inflation is natural. Throughout this section, we fix an orientation
of the target manifold $\cM=\rT^2$. Following our general method, we
pass to the complex parameterization of the model by encoding the
information carried by the scalar field metric into its volume form
and the complex structure $J$ determined by its conformal class. The
latter makes $\rT^2$ into an elliptic curve which we denote by $E$,
while the former determines and is determined by the conformal factor
of the metric in a complex coordinate chart relative to the complex
structure $J$. Since any elliptic curve can be presented as a quotient
$E=\C^2/\Lambda$ of the complex plane by a full lattice
$\Lambda\subset \C$, such a curve admits a global {\em periodic}
complex coordinate which is induced by the natural complex coordinate
$z=x_1+\i x_2$ of $\C^2$. This allows one to lift the contact
Hamilton-Jacobi equation to an equation defined on $\C$ whose defining
function is $\Lambda$-periodic. Moreover, the globally-defined
solutions of the contact Hamilton-Jacobi equation on $E$ correspond to
$\Lambda$-periodic solutions of this lifted equation. Using this
observation, we discuss a simple class of exact solutions and the
corresponding metrics on $\rT^2$. In principle, the problem of solving
boundary or asymptotic condition problems for the lifted contact
Hamilton-Jacobi equation can be approached by numerical
methods. However, this requires implementations of the finite element
method for equations with {\em periodic} boundary or asymptotic
conditions, and we could not find such code that is publicly
available.

\subsection{Riemannian metrics on the torus}

Recall that a Riemannian metric $\cG$ on $\rT^2$ determines and is determined
by the following data:

\begin{itemize}
\item A complex structure $J$ on $\rT^2$ which makes the torus into an
  elliptic curve $E$ of modulus $\tau$ determined by $J$. In
  turn, this determines a presentation $E_\tau=\C/\Lambda$, where
  $\Lambda$ is a full lattice in the complex plane.
\item A conformal factor $\phi\in \cC^\infty(\rT^2)$.
\end{itemize}

More precisely, any conformal structure on an oriented two-torus
determines a complex coordinate $z$ on the universal cover, thus
presenting $\rT^2$ as a quotient:
\ben
\label{ComplexTorus}
\rT^2=\C/\Lambda~~,
\een
where $\Lambda$ is a full lattice in $\C\simeq_\R \R^2$. This
coordinate is determined up to transformations of the form
$z\rightarrow az+b$, where $a,b$ are complex constants with $a\neq 0$.
Equivalently, $\Lambda$ is determined up to translations, rotations
and homotheties of $\R^2$. A presentation \eqref{ComplexTorus}
determines a distinguished point $\mathrm{o}\in \rT^2$ which
corresponds to the origin of $\C$ and its $\Lambda$-translates and
hence endows $\rT^2$ with an Abelian group structure for which
$\mathrm{o}$ is the zero element. If we require that the
transformations which preserve the complex structure fix this point,
then we must set $b=0$ in the transformation rule above. Then we can
always choose $a$ such that $\Lambda$ admits a basis (as a free
Abelian group) of the form $e_1=1$ and $e_2=\tau$, in which case
$\tau\in \C$ is the modular parameter of the corresponding complex
structure $J$ on the torus. When endowed with this complex structure,
the torus becomes the elliptic curve $E$ of modulus $\tau$. The
freedom of choosing $e_2$ amounts to the action of the modular group
$\PSL(2,\Z)$ on $\tau$ by fractional transformations:
\be
\tau\rightarrow \frac{a\tau+b}{c\tau+d}\quad\forall 
A:=\left[\begin{array}{cc} a & b \\ c & d\end{array}\right]\in \PSL(2,\Z)~~. 
\ee
A fundamental domain for this action is the well-known domain of
definition of the modular $j$-function.

\subsection{Periodic isothermal coordinates}

Let $\cG$ be a Riemannian metric on $\rT^2$ which belongs to the
conformal class defined by $J$ and ${\hat \cG}$ be its
lift\footnote{For ease of notation, we will omit hats indicating
lifts to the covering space in what follows.} to the universal cover
$\C$ of $E$. The real and imaginary parts $x_1=\Re z$ and $x_2=\Im z$
of the complex coordinate $z$ are global isothermal coordinates for
${\hat \cG}$ on $\C=\R^2$, which descend to coordinates on any open
subset of $\rT^2$ obtained by choosing a fundamental domain, i.e. a
basis for the lattice $\Lambda$. We can also view $z$ and $(x_1,x_2)$
as bi-periodic coordinates on $\rT^2$, though the period vectors need
not be ``aligned'' with $x_1$ and $x_2$. In these coordinates, the
metric $\cG$ has the form:
\be
\dd s^2=e^{2\phi(z,\bar{z})} |\dd z|^2=e^{2\phi(x_1,x_2)}(\dd x_1^2+\dd x_2^2)~~,
\ee
where $\phi$ is a $\Lambda$-periodic smooth positive function defined on
$\C$ (equivalently, a smooth positive function defined on $\rT^2$). We
denote by $\la~,~\ra$ the Euclidean scalar product on $\C=\R^2$:
\be
\la z,z'\ra =\Re(\bar{z} z')=x_1x'_1+x_2x'_2\quad\forall~z=x_1+\i x_2~~,~~z'=x'_1+\i x'_2\in \C~~.
\ee
This determines a flat Riemannian metric $\cG_0$ on $\rT^2$.

\subsection{Positive scalar potentials on the torus}

A scalar potential on $\rT^2$ lifts to a $\Lambda$-periodic smooth
function $V$ defined on $\C$ which admits a Fourier expansion:
\ben
\label{Fourier}
V(z,\bar{z})=\sum_{q\in \Lambda^\ast} V_q e^{2\pi \i \la q, z\ra}~~,
\een
with
\ben
V_q\eqdef \frac{1}{(2\pi)^2}\int_{\rT^2} V(z,\bar{z}) e^{-2\pi \i \la q, z\ra}\vol_0
\een
for all $q\in \Lambda^\ast$, where
\ben
\Lambda^\ast\eqdef \{q\in \C~| \la q,\lambda \ra \in \Z\quad\forall \lambda\in \Lambda\}\subset \C
\een
is the lattice dual to $\Lambda$. Here $\vol_0=\dd x_1\wedge \dd
x_2=\frac{\i}{2}\dd z \wedge \dd \bar{z}$ is the volume form of the
flat metric $\cG_0$. Since $V$ is square integrable on the compact
manifold $\rT^2$, the expansion \eqref{Fourier} converges in the
$\mathrm{L}^2$ norm defined by $\cG_0$. Since we assume that $V$ is
smooth, the series is in fact convergent in the $\cC^k$-norm for all
$k\in \Z_{\geq 0}\cup\{\infty\}$ and in particular is absolutely
convergent. The Parseval formula gives:
\ben
\label{Parseval}
\sum_{q\in \Lambda^\ast}|V_q|^2=\frac{1}{(2\pi)^2}\int_{\rT^2} |V(z,\bar{z})|^2\vol_0~.
\een
Since $V$ is real, its Fourier modes must satisfy the relations:
\be
\overline{V_q}=V_{-q}\quad \forall q\in \Lambda^\ast~~.
\ee
Writing
\be
V_q=A_q e^{2\pi \i\alpha_q}~,~\mathrm{with}~~ A_q\eqdef |V_q|\geq 0~~\mathrm{and}~~\alpha_q\in \R/\Z~~,
\ee
these relations give:
\be
A_{-q}=A_q~~\mathrm{and}~~\alpha_{-q}=-\alpha_q\quad\forall q\in \Lambda^\ast~~
\ee
and imply $\alpha_0=0$. The Parseval formula \eqref{Parseval} becomes:
\ben
\label{ParsevalA}
\sum_{q\in \Lambda^\ast} A_q^2=\frac{1}{(2\pi)^2}\int_{\rT^2} |V(z,\bar{z})|^2\vol_0~.
\een
Since $V$ is real, equation \eqref{Fourier} implies:
\ben
\label{FourierReal}
V(z,\bar{z})=\Re \sum_{q\in \Lambda^\ast} V_q e^{2\pi \i \la q, z\ra }=A_0+\!\!\sum_{q\in \Lambda^\ast\setminus \{0\}} \!\! A_q \cos(2\pi [\la q, z\ra+\alpha_q])~,
\een
where we used the fact that $\alpha_0=0$. The condition that $V$ is positive amounts to:
\be
A_0>-\sum_{q\in \Lambda^\ast\setminus\{0\}}  A_q \cos(2\pi [\la q, z\ra+\alpha_q])~~.
\ee

\subsection{Lattice-adapted coordinates}

Choosing a basis $(e_1,e_2)$ of $\Lambda$, we write $z=\theta^1
e_1+\theta^2 e_2$ with $\theta^1,\theta^2\in \R$. Such real
coordinates $(\theta^1,\theta^2)$ on the universal cover of $E$ are called {\em
  lattice-adapted coordinates}. Notice that a translation by a lattice
vector $\lambda=\lambda^1 e_1+\lambda^2 e_2\in \Lambda$ on the
universal cover corresponds to $\theta^i\rightarrow \theta^i+\lambda^i$,
where $\lambda^i\in \Z$. In lattice-adapted coordinates, we have:
\be
\la q, z\ra=q_1\theta^1+q_2\theta^2\quad\forall q\in \Lambda^\ast~,
\ee
where we defined:
\be
q_i\eqdef \la q,e_i\ra \in \Z~~.
\ee
In axion cosmology, the quantity $q=(q_1,q_2)$ is called the {\em
  covector of charges} determined by $q\in \Lambda^\ast$ in the basis
$(e_1,e_2)$ of $\Lambda$. In lattice-adapted coordinates, the Fourier
expansion \eqref{FourierReal} reads:
\ben
\label{FourierBasis}
V(z,\bar{z})=\sum_{(q_1,q_2)\in \Z^2} A_{q_1,q_2} \cos(2\pi [q_1\theta^1+q_2\theta^2+\alpha_{q_1,q_2}])~~,
\een
where $A_{q_1,q_2}\eqdef A_q$ and $\alpha_{q_1,q_2}\eqdef \alpha_q$. The metric $\cG$ has the squared line element:
\be
\dd s^2=e^{2\varphi(\theta^1,\theta^2)} \big[|e_1|^2\dd \theta_1^2+|e_2|^2\dd \theta_2^2+(e_1\overline{e_2}+e_2\overline{e_1})\dd \theta_1\dd\theta_2 \big]~~.
\ee
If we choose $e_1=1$ and $e_2=\tau$, then the last relation becomes:
\be
\dd s^2=\varphi(\theta^1,\theta^2) \big[\dd \theta_1^2+|\tau|^2\dd \theta_2^2+2(\Re \tau)\dd \theta_1\dd\theta_2 \big]~~.
\ee

\subsection{Charge-adapted coordinates}

Choosing a basis $(e_1^\ast, e_2^\ast)$ of the dual lattice
$\Lambda^\ast$, we can write $q=n_1 e_1^\ast+n_2 e_2^\ast$ with $n_1,
n_2\in\Z$. Then:
\be
\la q, z\ra=n_1 s^1+n_2 s^2\quad\forall z\in \C~,
\ee
where we defined:
\be
s^i\eqdef \la e_i^\ast, z\ra =\Re(\overline{e_i^\ast} z)\in \R~~.
\ee
The real coordinates $(s^1,s^2)$ on the universal cover are called
{\em charge-adapted coordinates}. A translation by a lattice vector
$\lambda\in \Lambda$ on the universal cover corresponds to
$s^i\rightarrow s^i+\lambda_i$, where $\lambda_i\eqdef \la e_i^\ast,
\lambda \ra\in \Z$. The Fourier expansion \eqref{FourierReal} reads:
\ben
\label{FourierExp}
V(z,\bar{z})=\sum_{(n_1,n_2)\in \Z^2} A_{n_1,n_2} \cos(2\pi [n_1 s^1+n_2 s^2+\alpha_{n_1,n_2}])~~.
\een
where $A_{n_1,n_2}\eqdef A_n$ and $\alpha_{n_1,n_2}\eqdef \alpha_n$. 

\begin{remark}
In axion cosmology it is often assumed that only a finite
number of the amplitudes $A_q$ are nonzero. In that case, the sums
above become finite. 
\end{remark}

\subsection{The strong SRRT equation and its lift to the covering space}

The discussion of Section \ref{sec:HJ} shows that the strong SRRT
equation for torus models can be written as:
\be
F(x_1,x_2,\phi, \pd_1\phi,\pd_2\phi)=0~~,
\ee
where $(x_1,x_2)$ are periodic isothermal coordinates and the contact
Hamiltonian $F$ is given in $\eqref{Floc}$. This form of the equation
is valid globally on $\rT^2$ if we interpret $x^1,x^2$ as periodic
coordinates on the torus. Equivalently, this equation can be viewed as
an equation for the $\Lambda$-periodic function ${\hat \phi}\eqdef
\phi\circ \pi$, where $\pi:\C\rightarrow T^2=\C/\Lambda$ is the
quotient projection. In that case, $F$ is viewed as a real-valued
function defined on $\C\times \R\times \R^2$ which is periodic under
translations by vectors of the lattice $\Lambda$. Boundary and
asymptotic condition problems for the contact Hamilton-Jacobi equation
lift in the obvious manner to the covering space.

\subsection{A simple example}

Let us fix  $q\in \Lambda^\ast\setminus \{0\}$, $A>0$, $A_0>A$ and $\alpha\in [0,1)$
  and consider the positive {\em single charge potential}:
\be
V(z,\bar{z})= A_0+A \cos(2\pi [\la q, z\ra +\alpha])~~,
\ee
whose critical points lie on the parallel lines: 
\be
L_+(n)\!=\!\{z\in \C~\vert~\la q,z\ra \!=\!-\alpha+n\}~~,
~~L_-(n)\!=\!\{z\in \C~\vert~\la q,z\ra\!=\!-\alpha+n+\frac{1}{2}\}~~(n\in \Z)~~
\ee
along which the potential attains its extremum values $A_0+A$ and $A_0-A$
respectively. These lines project to two parallel closed curves
$\mathcal{C}_+$ and $\mathcal{C}_-$ on the torus which have the same
rational slope. In particular, the potential $V$ on the torus is of
Bott-Morse type.  Notice that $V$ is invariant under translations of
the form $z\rightarrow z+t$ for any vector $t\in \R^2$ which is
$\cG_0$-orthogonal to $q$, i.e. which satisfies $\la q,t\ra=0$.
Accordingly, the strong SRRT equation \eqref{SSThermalRed} is invariant under
such translations and we can seek special solutions $\phi$ which
respect this symmetry, i.e. which satisfy:
\be
\phi(z+t)=\phi(z)\quad\forall t\perp_0 q~~,
\ee
where $\perp_0$ indicates orthogonality with respect to the scalar
product $\la ~,~\ra$. Notice that the gradient flow lines of $V$ are
parallel to the vector $q$. The normalized gradient vector
$n_0=\frac{\grad_0 V}{||\dd V||_0}$ of $V$ with respect to the flat
metric $\cG_0$ is given on the noncritical set of the lift of $V$ to
the covering space $\C$ of the torus by:
\be
n_0(z)=\zeta (z) {\hat q}\quad\forall z\in \C\setminus \cup_{n\in \Z}(L_+(n)\cup L_-(n))~~,
\ee
where ${\hat q}\eqdef \frac{q}{||q||_0}$ and: 
\ben
\zeta(z)\eqdef-\sign[\sin(2\pi [\la q,z\ra+\alpha]).
\een
Let $\tau_0(z)=\i n_0(z)=\zeta(z) q'$, where $q'\eqdef \i q$. This is the $\cG_0$-normalized
vector which is $\cG_0$-orthogonal to $n_0(z)$ such that
$(n_0(z),\tau_0(z))$ is a positively-oriented frame. Since $q'\perp_0
q$, translations which preserve $V$ have the form $t=\lambda q'$ with
$\lambda\in \R$. Thus $\pd_{\tau_0}\phi=0$ and the strong SRRT
equation \eqref{SSThermalRed} reduces to:
\ben
\label{eq0}
\tilde{H}_0^2 \,\big[||\dd V||_0\pd_{n_0}\phi+\Delta_0V - H_0\big]=3 e^{2\phi} V  \big[||\dd V||_0\pd_{n_0}\phi-H_0\big]^2~~. 
\een
The components of the Riemannian $\cG_0$-Hessian of $V$ in the frame $(n_0,\tau_0)$ are:
\be
V_{n_0\tau_0}=V_{\tau_0 n_0}=\pd_{n_0}\pd_{\tau_0}V=0~~\mathrm{and}~~V_{\tau_0\tau_0}=\pd_{\tau_0}^2V=0
\ee
(which give ${\tilde H}_0=0$) and:
\be
H_0=\Hess_0(V)(n_0,n_0)=\pd_{n_0}^2V=-(2\pi ||q||_0)^2 A\cos(2\pi [\la q, z\ra+\alpha])~~.
\ee
Hence \eqref{eq0} reduces to:
\ben
\label{eq1}
||\dd V||_0\pd_{n_0}\phi=-(2\pi ||q||_0)^2 A\cos(2\pi [\la q, z\ra+\alpha])~~.
\een
Since:
\be
||\dd V||_0=|\pd_{n_0}V|=2\pi ||q||_0 A \,\big|\!\sin(2\pi [\la q, z\ra +\alpha])\big|~~,
\ee
equation \eqref{eq1} reads:
\ben
\label{eq}
\pd_{{\hat q}}\phi(z)=2\pi ||q||_0 \cot(2\pi [\la q, z\ra +\alpha])~~.
\een
The general solution of \eqref{eq} which satisfies $\pd_{q'}\phi=0$ is:
\ben
\label{phisingle}
\phi=\log|\!\sin (2\pi [\la q,z \ra +\alpha])|+C~~,
\een
where $C\in \R$ is an integration constant. Notice that $\phi$ has
logarithmic singularities along the lines $L_\pm(n)$, which on the
torus project to the two connected components $\cC_+$ and ${\cC}_-$ of
the critical locus of $V$. Figure \ref{fig:Torus} shows the single
charge potential and the solution \eqref{phisingle} for the parameter
values listed in the caption.

\begin{figure}[H]
  \centering
\begin{minipage}{.45\textwidth}
\centering \includegraphics[width=.95\linewidth]{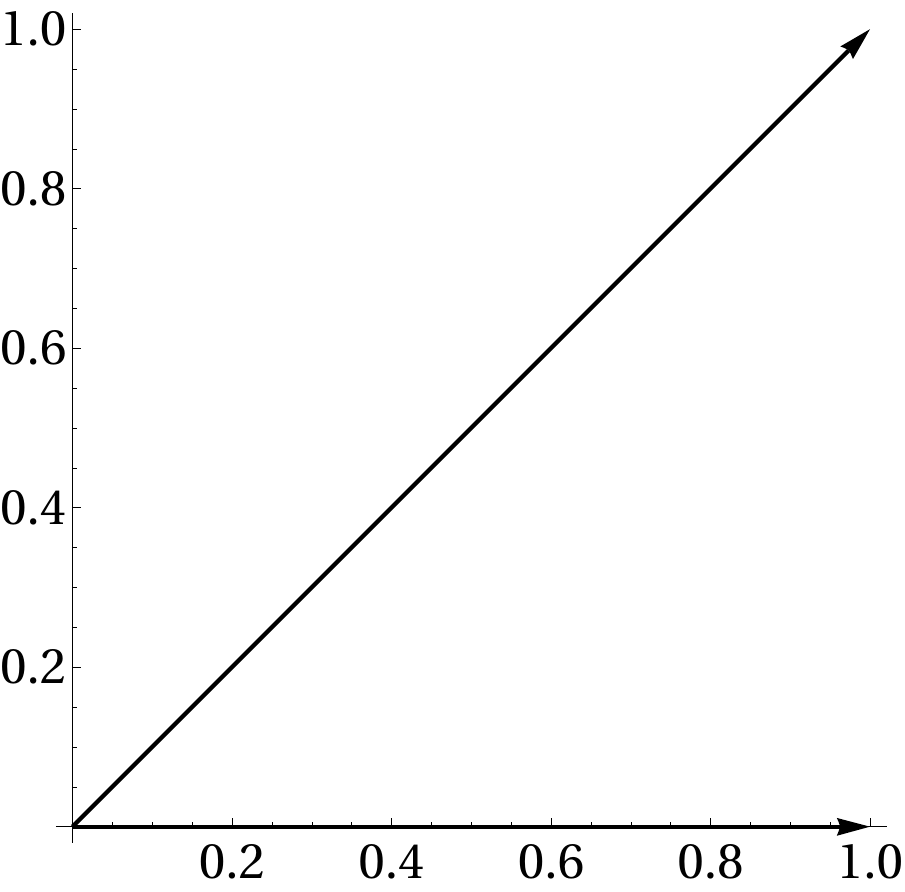}
\subcaption{The basis $(e_1,e_2)$ of the lattice $\Lambda$.}
\end{minipage}
\hfill
\begin{minipage}{.49\textwidth}
  \centering ~~ \includegraphics[width=.6\linewidth]{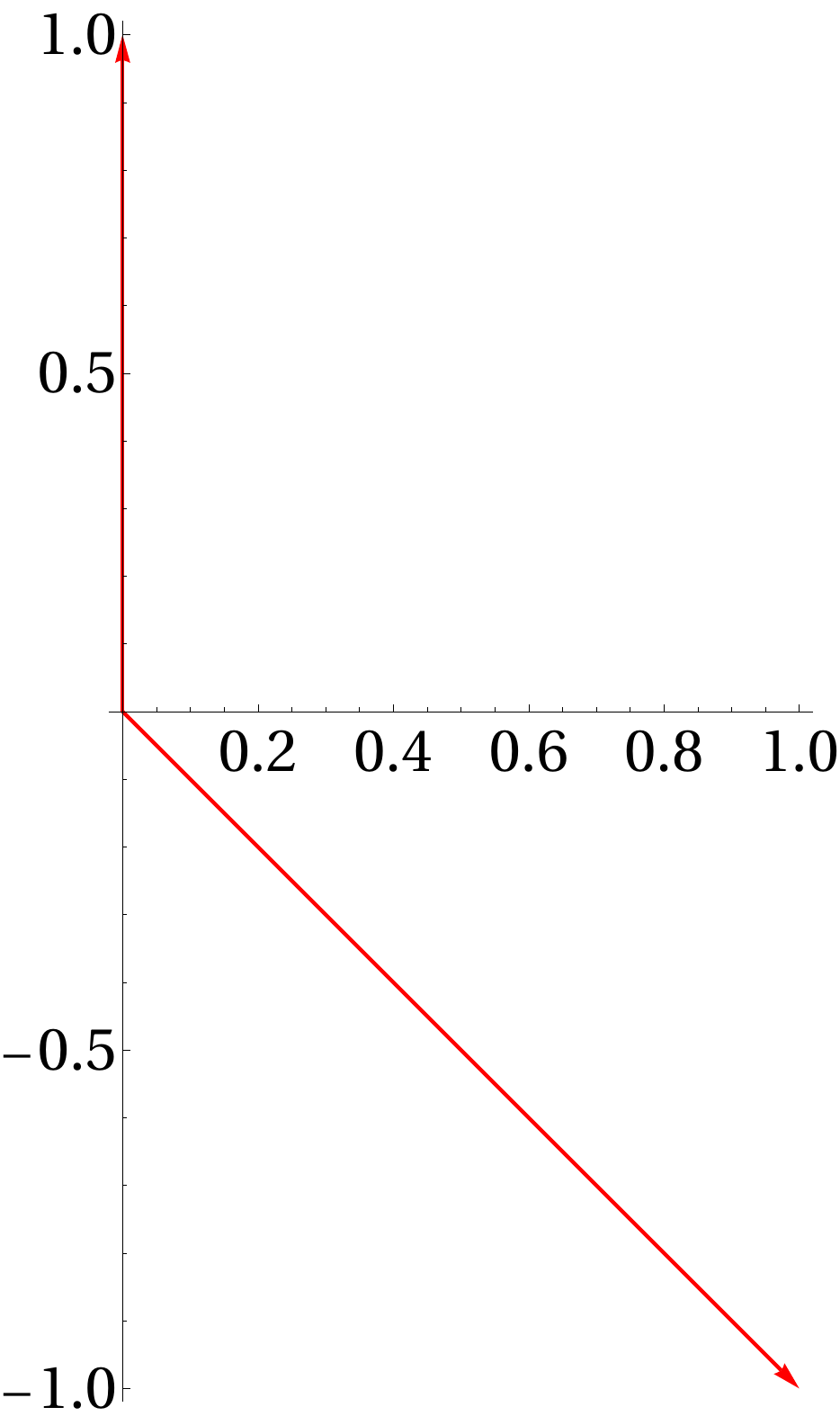}
\subcaption{The basis $(e_1^\ast,e_2^\ast)$ of the dual lattice $\Lambda^\ast$.}
\end{minipage}
\hfill
\centering
\begin{minipage}{.44\textwidth}
\vspace{1em}
\centering \includegraphics[width=1.05\linewidth]{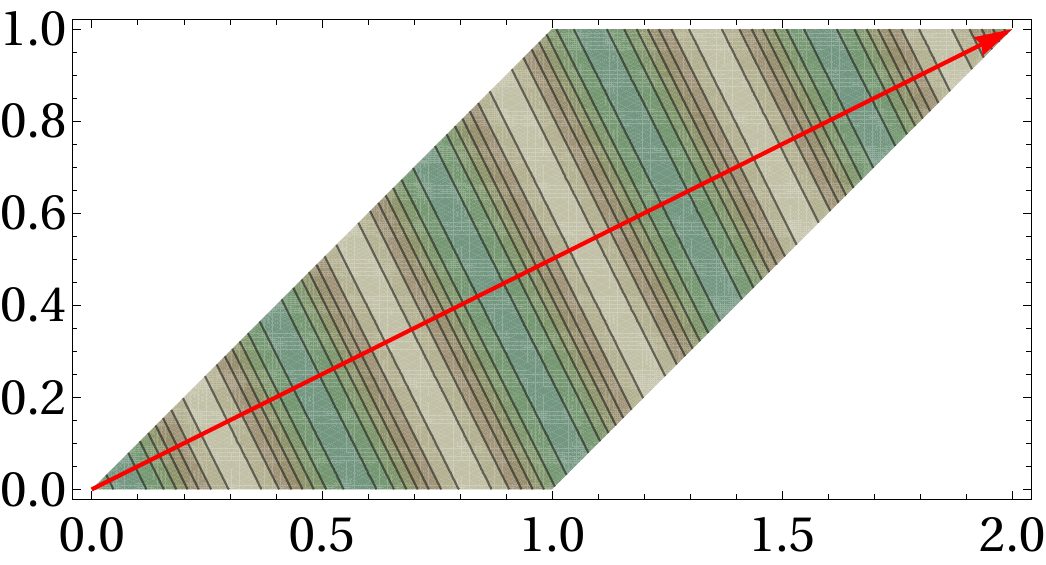}
\subcaption{Contour plot of the potential and its vector of charges $q$ (shown in red).}
\end{minipage}
\hfill
\quad\quad\quad \begin{minipage}{.46\textwidth}
\vspace{1em}
\centering  \includegraphics[width=1.05\linewidth]{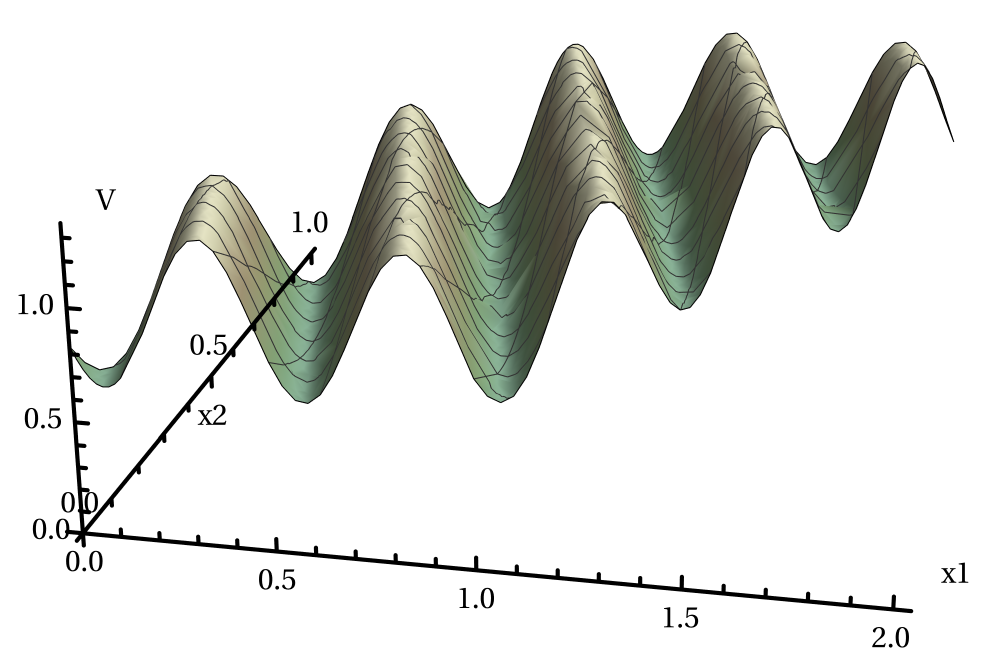}
\subcaption{3D plot of the potential.}
\end{minipage}
\hfill
\centering
\begin{minipage}{.44\textwidth}
\vspace{2em}
\centering \includegraphics[width=1.05\linewidth]{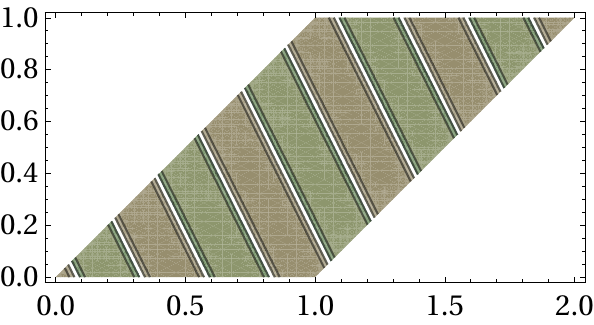}
\subcaption{Contour plot of the solution $\phi$ for $C=0$.}
\end{minipage}
\hfill
\begin{minipage}{.47\textwidth}
\vspace{1em}
\centering \includegraphics[width=\linewidth]{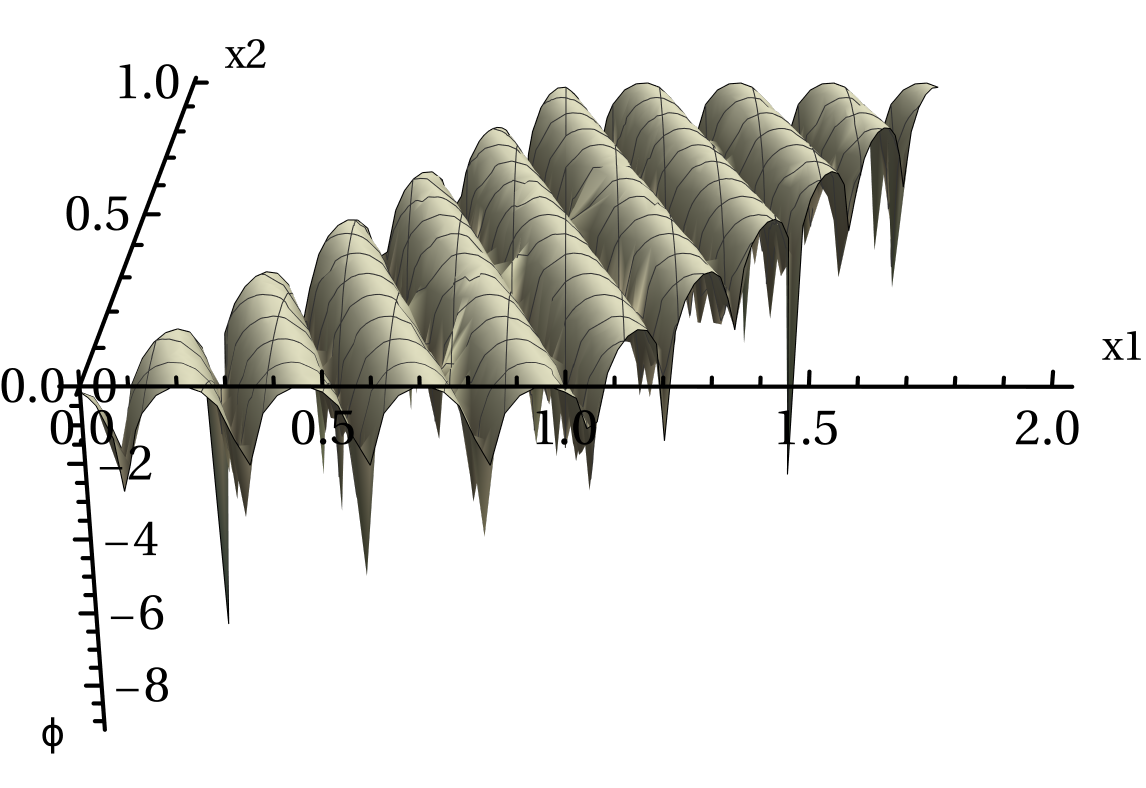}
\subcaption{3D plot of the solution $\phi$ for $C=0$. }
\end{minipage}
\hfill
\caption{The single charge potential and symmetry-adapted solution \eqref{phisingle}
  for the elliptic curve with $\tau=1+\i$ for $A_0=1$, $A=1/3$, $q_1=2$,
  $q_2=3$ and $\alpha=1/3$. We took $e_1=(1,0)$, $e_2=(1,1)$ and
  $e_1^\ast=(1,-1)$, $e_2^\ast=(0,1)$. The first two plots show these
  bases of $\Lambda$ and $\Lambda^\ast$. The remaining plots take
  $(x_1,x_2)$ to lie in the fundamental domain of $\Lambda$. The
  solution is given by $\phi=\log|\sin\left[2 \pi (2
    x_1+x_2+1/3)\right]|$, where we take the integration constant $C$ to be zero. }
\label{fig:Torus}
\end{figure}

\pagebreak

\section{Conclusion and further directions}
\label{sec:conclusions}

We extracted the frame-free form of the strong SRRT equation of
two-field cosmological models with oriented target space derived in
\cite{cons}, showing that it gives a geometric PDE which relates the
scalar field metric and potential of such models. Using the complex
parameterization of the former, we showed that, when the potential and
conformal class of the metric are fixed, this PDE reduces to a
geometric contact Hamilton-Jacobi equation for the volume form of the
scalar field metric, where the latter is viewed as a section of the
positive determinant half-line bundle $L_+$ of the target manifold
$\cM$ of the model. This volume form thus plays the role of abstract
Hamilton-Jacobi action for a geometric contact mechanical system
defined on the total space of the first jet bundle of $L_+$, which is
endowed with its Cartan contact structure.

In local isothermal coordinates defined by the complex structure $J$
determined by the conformal class of the scalar field metric, this
geometric PDE for the volume form becomes an ordinary contact
Hamilton-Jacobi equation for the conformal factor. In such
coordinates, the contact Hamiltonian is cubic in momenta, with
coefficients which depend in a complicated but explicitly known manner
on the scalar potential of the model. We analyzed this equation
locally using the classical method of characteristics and studied its
quasilinearization close to a non-degenerate critical point of the
potential, discussing various qualitative features and asymptotics
which are revealed by this approach.

Since the equation is nonlinear, its locally-defined classical
solutions generally do not extend globally, as illustrated by the
phenomenon of crossing characteristics. However, our contact
Hamilton-Jacobi operator turns out to satisfy the properness condition
of the theory of viscosity solutions, which allows us to apply that
theory to our situation. In particular, one can use the viscosity
perturbation of the equation in order to find numerical approximants
of its viscosity solutions.  We used this approach to compute
approximants of the viscosity solution of the Dirichlet problem in a
few simple cases of interest. Finally, we considered the case of
two-field models with toroidal scalar manifold, which arise as toy
models of axion cosmology. In this case, we found the general
symmetry-adapted solution of our contact Hamilton-Jacobi equation for
potentials with a single charge vector.

The present work opens up numerous new questions and directions for
further research. First, one can perform a similar analysis for the
{\em weak} SRRT equation derived in \cite{Achucarro} and
\cite{LiliaCons}, and we plan to do so in upcoming work.  Second, one
could write specialized code to compute efficiently solutions of our
contact Hamilton-Jacobi equation for general Riemann surfaces of
geometrically finite type (i.e. generally non-compact surfaces with
finitely-generated fundamental group), and in particular bi-periodic
viscosity solutions which are relevant to the case of modes with
toroidal scalar manifold. Third, one could investigate the problem of
existence and uniqueness of globally-defined viscosity solutions with
prescribed asymptotics of the type described in Subsection
\ref{subsec:asymptotics}. Such asymptotic condition problems seem to
be natural in the setting of two-field scalar cosmology. Fourth,
general lore suggests that the scale-invariant solutions of the
quasilinear approximation of our contact Hamilton-Jacobi equation,
which we investigated in Subsection \ref{subsec:scale}, should play a
distinguished role in a dynamical renormalization group approach to
studying that equation, in line with the general ideology advocated in
\cite{ren}. Last but not least, the geometric contact Hamilton-Jacobi
equation derived in the present paper may deserve detailed
mathematical study.

More ambitiously and for the longer term, one could aim to use the
fiducial two-field models which satisfy the strong SRRT equation (and
which can be determined numerically by computing appropriate viscosity
solutions of our contact Hamilton-Jacobi equation) in order to perform
a systematic study of the class of two-field models which allow for
``strong'' rapid turn inflation, as a function of the scalar potential
and of the topology of the target Riemann surface. Since the topology
of non-compact Riemann surfaces can be quite nontrivial, carrying out this
program requires the development of efficient software for computing
viscosity solutions which satisfy appropriate periodicity conditions
encoded by the fundamental polygon of such a surface. In principle,
the results and methods of the present paper allow one to construct
fiducial models for {\em any} strictly-positive and smooth scalar
potential defined on a given scalar manifold $\cM$. As a first
instance, it would be interesting to do so for scalar potentials which
satisfy the Morse condition and hence have isolated critical points.

Another direction of interest involves the systematic consideration of
higher order corrections in the approximations which led to the result
of \cite{cons}. This involves setting up an appropriate expansion
scheme which would extend the strong SRRT equation by including
corrections to all orders and in particular would contain (and widely
generalize) the weak SRRT condition derived in references
\cite{Achucarro} and \cite{LiliaCons}. Analyzing the resulting
conditions order by order should lead to a hierarchy of interesting
geometric PDEs relating the conformal factor and scalar potential on a
Riemann surface, when using the complex parameterization discussed in
the present work.

\section*{Acknowledgments}
\noindent This work was supported by national grant PN 23210101/2023
and by a Maria Zambrano Excellency Fellowship. We also acknowledge
funding by COST Action COSMIC WISPers CA21106, which is supported by
COST (European Cooperation in Science and Technology).

\end{document}